\documentclass[12pt,twoside]{article} 
\usepackage[T1]{fontenc} % font

\usepackage{amsthm, amsmath, amssymb, mathrsfs, upgreek, enumerate, mathtools,
  comment, natbib, subfigure, graphicx, mathabx, booktabs, threeparttable,
tikz, relsize, exscale, epstopdf, scrextend, multirow, xr-hyper, pdfpages, bm}
\usepackage[headheight=110pt,margin=1.25in]{geometry}
\usepackage[section]{placeins}

\usepackage[toc,page,header]{appendix}

\usepackage{fancyhdr, setspace} % headers and footers
\usepackage[pagebackref=false]{hyperref} % link colors
\hypersetup{
    colorlinks=true,
    citecolor=blue,
    filecolor=black,
    linkcolor=black,
    urlcolor=blue,
    bookmarksopen=true,
    pdfstartview=FitH
}

\newtheorem{theorem}{Theorem}[section]
\newtheorem{assump}{Assumption}[section]
\newtheorem{lemma}{Lemma}[section]

\theoremstyle{definition}
\newtheorem{definition}{Definition}[section]

\newtheorem{example}{Example}[section]

\newtheorem{remark}{Remark}[section]

\def\equationautorefname~#1\null{(#1)\null} % for autoref to work as eqref

\def\Snospace~{\S{}}
\makeatletter
\def\thm@space@setup{
  \thm@preskip=15pt \thm@postskip=15pt % controls spacing before and
}
\makeatother

\def\indep{\perp\!\!\!\perp}

\newcommand{\cov}{\text{Cov}}
\newcommand{\var}{\text{Var}}
\newcommand{\E}{{\bf E}}
\newcommand{\R}{\mathbb{R}}
\newcommand{\RR}{\bm{R}}

\newcommand{\prob}{\mathbf{P}}
\newcommand{\plimarrow}{\stackrel{p}\longrightarrow}
\newcommand{\dlimarrow}{\stackrel{d}\longrightarrow}

\newcommand{\X}{\mathcal{X}}

\newcommand{\zero}{\bm{0}}
\newcommand{\W}{\bm{W}}
\newcommand{\cc}{\bm{c}}

\newcommand*{\medcap}{\mathbin{\scalebox{1.5}{\ensuremath{\cap}}}}
\newcommand*{\medcup}{\mathbin{\scalebox{1.5}{\ensuremath{\cup}}}}

\providecommand{\abs}[1]{\lvert#1\rvert} 
\providecommand{\norm}[1]{\lVert#1\rVert}

\let\emptyset\varnothing

\renewcommand{\qed}{\hfill \mbox{\raggedright \rule{0.08in}{0.08in}}} % black QED box
\renewenvironment{proof}[1][\proofname]{{\noindent\sc#1. }}{\qed\vspace{15pt}} % "proof" small caps

\numberwithin{equation}{section} % include section number in equation numbering

\title{Inference in Models of Discrete Choice with Social Interactions Using Network Data\bf\sc \thanks{This research is supported by National Science Foundation grant SES-1755100.}}

\author{Michael P.\ Leung\thanks{We thank audiences at the Network Econometrics Juniors' Conference at Northwestern and the North American Winter Meetings. Department of Economics, University of Southern California. E-mail: leungm@usc.edu.}}

\begin{document}
\maketitle

\begin{abstract}

  {\sc Abstract.} This paper studies inference in models of discrete choice with social interactions when the data consists of a single large network. We provide theoretical justification for the use of spatial and network HAC variance estimators in applied work, the latter constructed by using network path distance in place of spatial distance. Toward this end, we prove new central limit theorems for network moments in a large class of social interactions models. The results are applicable to discrete games on networks and dynamic models where social interactions enter through lagged dependent variables. We illustrate our results in an empirical application and simulation study.
  
  \vspace{15pt}

  \noindent {\sc JEL Codes}: C22, C31, C57

  \noindent {\sc Keywords}: social networks, peer effects, empirical games, HAC estimator
 
\end{abstract}

\addcontentsline{toc}{part}{Main Paper}
\newpage

%----------------------------------------------------------------------
\section{Introduction}\label{sintro}
%---------------------------------------------------------------------- 
\onehalfspacing

Threshold models of social influence are the subject of a large theoretical literature in the social sciences \citep{ellison1993learning,granovetter1978threshold,jackson2010,morris2000contagion}. These are models of social interactions with binary decisions, which have been used to study, for example, product adoption \citep{godinho2014peer}, risky behavior \citep{bauman1996importance}, health choices \citep{christakis2004social}, protests \citep{gonzalez2017collective}, and voting \citep{bond201261}. The empirical content of these and related models has long been of interest in the econometric literature \citep{brock2001discrete,manski1993identification}. Much of this work has focused on static models with cross-sectional data, where social interactions only operate within groups and the number of groups is large. This paper instead studies inference in static and dynamic models when the data consists of a single large network. The motivation for using network data is that it can be used to more accurately model social interactions, which is often heterogeneous by nature. For example, students in different classrooms may interact, and even within classrooms, they may only interact with a subset of their classmates. %The increasing availability of network data has enabled researchers to better study heterogeneity in social interactions and subject to empirical inquiry a large body of theoretical work on the impact of network topology on the diffusion of behaviors, products, and information \citep{jackson2017economic}.

The goal of this paper is to provide theoretically justified inference procedures for models of discrete choice with social interactions. We consider spatial and network heteroskedasticity and autocorrelation consistent (HAC) variance estimators, which have seen increasing use in applied work for inference in the single network setting. Our main theoretical results are new central limit theorems for static and dynamic models of social interactions and conditions under which the HAC estimators are consistent, or possibly conservative, for the variance. These results hold under an asymptotic approximation that sends the size of a single large network to infinity.

The motivation for large-network asymptotics is that network data commonly consists of observations on a single network, as opposed to many independent groups. This is theoretically challenging because social interactions inherently induce autocorrelation between different network subunits. Inference thus requires a large-sample theory under which the amount of ``independent information'' in a network grows with the network size. This is analogous to limit theory in time series that sends the number of time periods to infinity rather than the number of time series.

There are few inference procedures presently available in the single-network setting. In the medical literature, network autocorrelation is often ignored in practice, and i.i.d.\ standard errors are commonly used \citep{lee2019network}. In economics, the strategy of clustering on subnetworks is quite common. By this we mean dividing a single network into many subnetworks using, say, geographic boundaries or a community detection algorithm \citep[e.g.][]{blondel2008fast}, and then clustering standard errors on the subnetworks. However, there is often a sizeable share of links that bridge clusters, which implies agents can interact across subnetworks. This renders implausible the assumption necessary for the validity of clustered standard errors, that subnetworks are independent. HAC estimators provide a viable alternative, as they do not rely on partitioning or independence of clusters. Instead, they account for correlation between each agent and the alters within a neighborhood of the agent, which can be thought of as defining agent-specific, overlapping clusters.

To provide a clearer sense of the practical import of our results, consider the study of \cite{conley2010learning} on social learning among pineapple farmers in Ghana. Their main specification is a logistic model of the probability that a farmer's agricultural inputs change in response to above- or below-expectation returns from his network neighbors' crop yields. For the purposes of constructing clustered standard errors, it is unclear how to reasonably divide the network using geographic boundaries, as there are only three villages in their data. The authors instead use a spatial HAC estimator, which is valid when the data is spatially autocorrelated. We provide conditions for validity of the estimator when the data is potentially {\em network} autocorrelated, for example when the errors of network neighbors are dependent.

The spatial HAC estimator requires data on agents' spatial positions, but this is not always available. For instance, rather than physical space, the underlying space may correspond to a latent social space, whereby socially similar agents are more likely to form connections. Given network data, we can instead define an alternate notion of ``distance,'' namely {\em path distance}, which is the shortest number of links it takes to travel from one agent to another in the network, and use this in place of spatial distance to construct a network HAC estimator. This idea has been used in practice \citep[e.g.][]{acemoglu2015state,eckles2016estimating}. We provide conditions under which the estimator is valid. 

Our results apply to static and dynamic models of social interactions. In dynamic models, an agent's action is a function of her network neighbors' lagged actions. Such models are useful when a short time series on a large network is available. We discuss how our results can be used for parametric and nonparametric inference in dynamic models. Static models instead allow actions to depend on neighbors' contemporaneous actions. They formally correspond to discrete games of complete information, which are useful when the data consists of a single network snapshot. In the existing literature, \cite{xl2015} propose a simulated method of moments estimator and \cite{li2016partial} characterize the identified set using subnetwork moments. We discuss how our results can be used to conduct inference using their procedures.

We illustrate our results in a simulation study and empirical application. The simulation evidence shows that the normal approximation works well, and inference using the HAC estimators properly controls size, whereas naive i.i.d.\ standard errors can be highly anti-conservative. In the empirical application, we reanalyze data from \cite{conley2010learning} and find that both spatial and network HAC estimators deliver similar standard errors. 

%It should be noted that there are still reasons to use the spatial HAC, if the data is available. First, the network HAC requires good measurement of linkages to obtain path distances. Second, consistency of the network HAC requires the network to be exogenous in static settings, whereas no such assumption is needed for the spatial HAC. Third, we provide results under which the spatial HAC can be used when the usual first-order stationarity condition fails to hold, whereas corresponding results for the network HAC appear difficult to establish. This is important for moment inequality models.

Finally, we would like to highlight several technical contributions. 
\begin{enumerate}
  \item To our knowledge, all existing results on HAC estimators require first-order stationarity conditions that essentially require moments to be centered at their conditional expectations. This holds in GMM-type settings. However, it is not satisfied in moment inequality models and often violated in nonparametric settings, for example nonparametric inference on the conditional choice probability. For the spatial HAC, we provide two positive results. The first is a general (albeit strong) sufficient condition under which first-order stationarity holds asymptotically for any network moments, whether centered or not. In this case, the spatial HAC is asymptotically conservative, which appears to be a new result. Second, our results justify use of the generalized spatial HAC proposed by \cite{leung2019normal}, which is consistent regardless of whether stationarity holds.

  \item To prove our CLTs, we first establish a general CLT under high-level ``stabilization'' conditions, which is a modification of a result due to \cite{leung2019normal}. The difference is we consider an increasing domain setting to obtain consistency of the spatial HAC, whereas they study a quasi-infill setting.\footnote{Specifically, under our asymptotics, positions are scaled to be sufficiently spread out in space, which is the usual increasing domain asymptotics used to analyze spatial HAC estimators. Under the asymptotics of \cite{leung2019normal}, positions are instead sampled from a bounded region, but preferences are scaled so that agents increasingly prefer nearby alters as the network size grows. \label{quasiinfill}} To derive primitive conditions for stabilization in our setting, we follow the methodology in \cite{leung2019normal}, drawing on results in branching process theory. While the conditions we derive are qualitatively similar to those used in \cite{leung2019normal} to establish CLTs for strategic models of network formation, our results are {\em not} special cases of theirs. They consider applications to network formation and network regressions, but their results do not apply to social interactions models.

  \item Our results on the validity of the network HAC require an exogenous network in the static model. We provide intuition suggested by our proof strategy on why endogeneity is difficult to allow in general.
\end{enumerate}

{\bf Related Literature.} To our knowledge, the only prior work on network HAC estimators is a recent paper by \cite{kojevnikov2019limit}. \cite{kojevnikov2019bootstrap} proposes novel bootstrap procedures for network data. Both papers assume the data satisfies a new conditional $\psi$-weak dependence notion modified from the time series literature but do not discuss applications to models of social interactions. We utilize a different notion of weak dependence proposed by \cite{leung2019normal} and apply branching process results to verify weak dependence in our applications. See \autoref{relit} for further comparison of our approach and theirs.

There is a large literature on spatial autoregressive models, of which the widely used linear-in-means model of social interactions is a special case. For relevant results on spatial HAC estimators, see for example \cite{conley1999gmm}, \cite{conley2007spatial}, \cite{jenish2016}, and \cite{kelejian2007hac}. In the context of network formation, \cite{boucher2017my} and \cite{leung2019normal} show that spatial HAC estimators can be used for valid inference.

A growing literature in econometrics studies models with strategic interactions and many agents. \cite{menzel2015inference} and \cite{shang2011two}, among others, consider inference on discrete games of complete information. Unlike our setting, there is no network structure; instead, social interactions enter payoffs through a vector of aggregate statistics, such as the average action of all agents. \cite{li2016partial} and \cite{xl2015} study settings with networked interactions and complete information. \cite{eraslan2017identification} and \cite{xu2018social} derive results for corresponding games of incomplete information. Large-network asymptotics are different in this setting because actions are i.i.d.\ conditional on observed characteristics and the network. \cite{kuersteiner2018dynamic} develop large-sample theory for a dynamic version of the linear-in-means model where social interactions enter through lagged dependent variables. \cite{he2018measuring} propose inference methods for nonparametric measures of diffusion in dynamic models with binary outcomes.

{\bf Outline.} The next two sections respectively introduce dynamic and static models of social interactions with binary outcomes. Then \autoref{svar} presents the HAC estimators and conditions for their asymptotic validity. We state formal conditions for central limit theorems in \autoref{ssm} for the static model and \autoref{sdm} for the dynamic model. Next, \autoref{seill} discusses results from a simulation study and empirical application using data from \cite{conley2010learning}. Finally, \autoref{sconclude} concludes.

{\bf Notation.} If $f$ is a density function or random vector, let $\text{supp}(f)$ be its support. Given $n$ i.i.d.\ vectors $X_1, \dots, X_n$ and $H \subseteq \{1, \dots, n\}$, let $X_H$ be the submatrix $(X_i\colon i \in H)$ and $X_{-i} = (X_1, \dots, X_{i-1}, X_{i+1}, \dots, X_n)$. We use boldface letters to denote the entire collection $\bm{X} = (X_i)_{i=1}^n$ (as opposed to $X$ to avoid confusing this with a generic draw). For a symmetric $n\times n$ matrix $\bm{A}$, we let $A_{ij}$ denote the $ij$th entry, $A_i$ denote the $i$th column, and $A_{-i}$ denote $\bm{A}$ with the $i$th row and column deleted. Also for $H$ defined previously, let $A_H$ be the submatrix $(A_{ij}\colon i,j \in H)$ and $A_{H,i} = (A_{ij}\colon j \in H)$. For all of these submatrices and subvectors, rows and columns are ordered in the same way as the original matrices and vectors.

Throughout this paper we will only be concerned with undirected networks with no self-links. Accordingly, we represent a network $\bm{A}$ on $n$ agents as an $n\times n$ adjacency matrix with zeros on the diagonals. The $ij$th entry $A_{ij}$, which we call the {\em potential link} between $i$ and $j$, takes values in $\{0,1\}$. The {\em degree} of an agent $i$ in $\bm{A}$ is $\sum_j A_{ij}$. A {\em path} in $\bm{A}$ from agent $i$ to $j$ is a sequence of distinct agents starting with $i$ and ending with $j$ such that for each consecutive $k,k'$ in this sequence, $A_{kk'}=1$. The number of agents in this sequence minus one is the {\em length} of this path. The {\em path distance} between two agents is the length of the shortest path that connects them, assuming one exists; if one does not, then it is defined as $\infty$. The {\em $K$-neighborhood} of an agent $i$ in $\bm{A}$, denoted $\mathcal{N}_{\bm{A}}(i,K)$, is the set of all agents $j$ for which $\ell_{\bm{A}}(i,j) \leq K$. Note that this includes $i$. Finally the {\em component} of an agent $i$ in $\bm{A}$ is the set of all agents $j$ for which $\ell_{\bm{A}}(i,j) < \infty$. 

%----------------------------------------------------------------------
\section{Dynamic Model of Social Interactions}\label{sdynamic}
%----------------------------------------------------------------------

A large literature in microeconomics and computer science studies dynamic models of social influence \citep{ellison1993learning,kempe2003maximizing,montanari2010spread,morris2000contagion}. In these models, a subset of agents are initially seeded as adopters (choosing action 1), and then agents myopically best respond in subsequent periods to the number of neighbors who are adopters. For example, if the share of period $t-1$ adopters in $i$'s neighborhood exceeds some threshold $\tau_i$, then agent $i$ might adopt in period $t$. This section introduces an econometric version of the model with observed and unobserved heterogeneity. Models of this sort are widely used in applied work in economics, marketing, and network science.\footnote{E.g.\ \cite{ameri2017structural}, \cite{banerjee2013diffusion}, \cite{christakis2007spread}, \cite{iyengar2011opinion}, \cite{katona2011network}, and \cite{park2018social}.} They are useful when the econometrician observes a short time series of a large network.

%------------------------------------
\subsection{Model}
%------------------------------------ 

The econometrician observes a set of agents $\mathcal{N}_n = \{1, \dots, n\}$ connected through a time-invariant network $\bm{A}$. Agents interact over a small number of time periods $t = 0, 1, \dots, T$ for $0<T<\infty$. Associate each agent $i$ and period $t$ with a binary outcome $Y_i^t$ and type $\tau_i^t$. We decompose $\tau_i^t = (X_i^t, \alpha_i, \varepsilon_i^t)$, where $X_i^t$ is a vector of observed covariates, $\alpha_i$ a fixed effect, and $\varepsilon_i^t$ an idiosyncratic error. We assume $X_i \equiv (X_i^0, \dots, X_i^T, \alpha_i)$ is i.i.d.\ across agents, and independent of $\varepsilon_i^t$ for any $t$, while the errors are i.i.d.\ across agents and time periods. 

Outcomes are realized according to myopic best-response dynamics: for all agents $i$ and periods $t = 1, \dots, T$,
\begin{equation}
  Y_i^t = \bm{1}\left\{ U(S_i^t, \tau_i^t) > 0 \right\}, \label{dymodel}
\end{equation}

\noindent where $U(\cdot)$ is a real-valued function representing net utility. The first term $S_i^t$ is a finite-dimensional vector of statistics given by
\begin{equation*}
  S_i^t \equiv S(Y^{<t}, \tau_i^{\leq t}, \tau_{-i}^{\leq t}, A_i, A_{-i})
\end{equation*}

\noindent for a vector-valued function $S(\cdot)$ with $Y^{<t} = (Y_j^s\colon j \in \mathcal{N}_n, s < t)$, $\tau_i^{\leq t} = (\tau_i^s\colon s \leq t)$, and $\tau_{-i}^{\leq t}$ the collection of $\tau_i^{\leq t}$'s excluding $i$ (see the definition of the $-i$ subscript in the introduction). Hence, $S_i^t$ can capture peer effects through lagged outcomes and types of other agents in the network. As a function of $\tau^{\leq t}_{-i}$, $S_i^t$ can also depend on the unobserved subvector of neighbors' types, which can generate network autocorrelation in unobservables.

\begin{example}\label{e1}
  A typical payoff specification is the linear in parameters model
  \begin{equation*}
    U(S_i^t, \tau_i^t) = (X_i^t)'\beta_1 + \beta_2 \frac{\sum_j A_{ij} Y_j^{t-1}}{\sum_j A_{ij}} + \alpha_i + \nu_i^t.
  \end{equation*}

  \noindent Social interactions enter through the average outcome of neighbors in the previous period. The idiosyncratic term $\nu_i^t$ might simply equal $\varepsilon_i^t$, or it might be autocorrelated. For example, the errors may be jointly normal across agents with nonzero covariances for linked pairs. Alternatively, we could have
  \begin{equation*}
    \nu_i^t = \frac{\sum_j A_{ij} \varepsilon_j^t}{\sum_j A_{ij}} + \varepsilon_i^t,
  \end{equation*}

  \noindent where the first term captures exogenous peer effects in the errors and induces contemporaneous network autocorrelation between the $\nu_i^t$'s. In the latter case, $S_i^t$ is a two-dimensional vector consisting of this term and the term multiplying $\beta_2$.
\end{example}

\begin{example}\label{e2}
  A large literature dating back to at least \cite{granovetter1978threshold} studies threshold models of behavior, where for each agent $i$, $Y_i^t = 1$ if and only if the share of neighbors choosing action one in the previous period exceeds a threshold $\xi_i^t$ \citep{jackson2010,schelling2006micromotives}. This is captured in our framework by setting
  \begin{equation*}
    U(S_i^t, \tau_i^t) = \frac{\sum_{j\neq i} A_{ij} Y_j^{t-1}}{\sum_{j\neq i} A_{ij}} - \xi_i^t.
  \end{equation*}

  \noindent If $\xi_i^t$ only depends on own type $\tau_i^t$, then in this model, $S_i^t$ is a scalar corresponding to the average lagged outcome on the right-hand side. In practice, heterogeneity in the threshold is often of interest, since this determines the extent of diffusion, so a typical exercise would be to parametrize it as a function of type and estimate the parameters.
\end{example}

Both of these examples restrict the dependence of $S_i^t$ on the network. We next impose this restriction more generally. Let $\mathcal{N}_{\bm{A}}^-(i,M) = \mathcal{N}_{\bm{A}}(i,M)\backslash\{i\}$, which is $i$'s $M$-neighborhood, excluding $i$ herself. Let $Y_H^{<t} = (Y_j^s\colon j \in H, s < t)$, and recall the notation for submatrices in the introduction. 

\begin{assump}[Local Interactions]\label{S2}
  There exists $M \in \mathbb{N}$ such that, for all $n\in\mathbb{N}$, $i \in \mathcal{N}_n$, and $t = 1, \dots, T$,
  \begin{equation*}
    S_i^t = S(Y_{\mathcal{N}(i)}^{<t}, \tau_i^{\leq t}, \tau_{\mathcal{N}(i)}^{\leq t}, A_{\mathcal{N}(i)}) \quad\text{for}\quad \mathcal{N}(i) \equiv \mathcal{N}_{\bm{A}}^-(i,M).
  \end{equation*}
\end{assump}

\noindent This states that $S(\cdot)$ is only a function of its arguments through the $M$-neighbors of $i$, which is clearly satisfied in the previous examples for $M=1$.

Model \autoref{dymodel} governs the evolution of the process from period 1 onward. It remains to specify a model for the initial condition. This model need not be known in practice to use our results, but we will need to impose some (nonparametric) restrictions for weak dependence in order to establish a CLT. If this process is observed shortly after its inception, it is reasonable to draw $Y_i^0$ from a single-agent discrete choice model, e.g.\ $Y_i^0 = \bm{1}\{U(\zero, \tau_i^0) > 0\}$, where we zero out the regressors that are functions of lagged dependent variables because there is no previous time period. Alternatively, the initial condition might be viewed as the long-run outcome of a dynamic process. This can be reasonably approximated by a static model of social interactions, which is discussed in \autoref{smodel}. Since this nests the single-agent discrete choice model, we will assume a general static model for the initial condition, whose formal statement is postponed to \autoref{sdm}. 

%------------------------------------
\subsection{Network Moments}\label{sdynetmoms}
%------------------------------------

We next define the class of dynamic network moments, objects for which we seek to prove a CLT and construct variance estimators. Let $Y_i = (Y_i^t)_{t=0}^T$, $\tau_i$ be defined in the analogous way, and $Y_{-i}, \tau_{-i}$ be defined according to the notation in the introduction. We consider moments that are averages of {\em agent statistics} $\psi_i$, namely
\begin{equation*}
  \frac{1}{n} \sum_{i=1}^n \psi_i \quad\text{for}\quad \psi_i \equiv \psi(Y_i, Y_{-i}, A_i, A_{-i}, \tau_i, \tau_{-i}), 
\end{equation*}

\noindent where $\psi(\cdot)$ is $\R^m$-valued. The main technical contributions of this paper are twofold. In \autoref{sdm}, we provide conditions under which, as $n\rightarrow\infty$,
\begin{equation}
  \frac{1}{\sqrt{n}} \sum_{i=1}^n \bm{\Sigma}_n^{-1/2} \left( \psi_i - \E[\psi_i] \right) \dlimarrow \mathcal{N}(\zero, \bm{I}_m) \quad\text{and}\quad \liminf_{n\rightarrow\infty} \lambda_\text{min}(\bm{\Sigma}_n) > 0, \label{clt}
\end{equation}

\noindent where $\bm{\Sigma}_n$ is the variance and $\bm{I}_m$ the $m\times m$ identity matrix. In \autoref{svar}, we show that spatial and network HAC estimators are valid estimators for $\bm{\Sigma}_n$. The formal sequence of models along which we take these limits is defined in \autoref{lna}.

We consider agent statistics $\psi(\cdot)$ satisfying the following {\em $K$-locality} condition. Recall the notation introduced prior to \autoref{S2}.

\begin{assump}[$K$-Locality]\label{klocal}
  There exists $K \in \mathbb{N}$ such that, for any $n$ and $i\in\mathcal{N}_n$,
  \begin{equation*}
    \psi_i = \psi(Y_i, Y_{\mathcal{N}(i)}, A_{\mathcal{N}(i),i}, A_{\mathcal{N}(i)}, \tau_i, \tau_{\mathcal{N}(i)}) \quad\text{for}\quad \mathcal{N}(i) \equiv \mathcal{N}_{\bm{A}}^-(i,K).
  \end{equation*}
\end{assump}

\noindent This states that, for some positive constant $K$, the agent statistic of $i$ is only a function of its arguments through the agents in $i$'s $K$-neighborhood. We first walk through two basic illustrative examples and then discuss network moments useful for inference on social interactions. 

\begin{example}\label{ccp}
  Consider $\psi_i = Y_i^t$ for some specified time period $t$. Then the network moment corresponds to the average outcome or empirical choice probability at time $t$. A related example is the average outcome over all observed time periods, where $\psi_i = T^{-1} \sum_{t=0}^T Y_i^t$. Both of these satisfy \autoref{klocal} for $K=0$ because they are only functions of the first argument of $\psi(\cdot)$.
\end{example}

\begin{example}\label{subnetmoms}
  \autoref{klocal} also encompasses subnetwork moments, such as the number of dyads (linked pairs) such that both agents choose action 1 in period $t$. This corresponds to $\psi_i = \sum_j A_{ij}Y_i^tY_j^t$, which satisfies \autoref{klocal} for $K=1$, since it is only a function of $i$'s neighbors.\footnote{Note that the network moment is then $n^{-1} \sum_{i,j} A_{ij}Y_i^tY_j^t$. We always scale by $n^{-1}$ because our assumptions will ensure a sparse network, which implies that $\sum_j A_{ij} = O_p(1)$ for any $i$. See \autoref{rsparse}. \label{sparsescale}} Another example is the number of intransitive triads (triplets with only two links) such that all agents choose action 0. This corresponds (up to scale) to $\psi_i = \sum_j \sum_k A_{ij}A_{jk}(1-A_{ik}) (1-Y_i^t)(1-Y_j^t)(1-Y_k^t)$, which satisfies \autoref{klocal} for $K=2$, since $k$ may be a 2-neighbor of $i$.
\end{example}

\noindent We next provide examples of network moments useful for parametric and nonparametric inference on social interactions.

%------------------------
\subsubsection{Method of Moments}\label{sdysmm}
%------------------------

In practice, a common exercise is to parametrize payoffs $U(\cdot)$ and the conditional distributions of $\alpha_i$ and $\varepsilon_i^t$ given $X_i^t$. For example, consider a two-period version of \autoref{e1} with $\alpha_i = 0$ for all $i$ and $\nu_i^t \sim \mathcal{N}(0,1)$. We can estimate the model using a probit regression of $Y_i^1$ on $(X_i^1, \sum_j A_{ij} Y_j^0 / \sum_j A_{ij})$. However, the usual probit likelihood may not be a true likelihood because the errors $\nu_i^t$ are potentially network autocorrelated. Nonetheless, we can treat this as a pseudo-likelihood, as in \cite{poirier1988probit}.

To obtain a normal limit for the pseudo-MLE estimator, we need the average of the scores to obey a CLT, in addition to the usual regularity conditions. Here we take $\psi_i$ to be the score for agent $i$. This satisfies \autoref{klocal} for $K=1$ because the regressors are a function of $i$'s 1-neighbors. The information matrix equality does not hold, since this is a pseudo-likelihood, but a consistent estimate of the asymptotic variance can obtained using the sandwich formula with the variance of the scores estimated using either our proposed spatial or network HAC.

Since the model is fully specified up to a vector of parameters $\theta$, we can also apply simulated method of moments using, for example, the moments introduced in Examples \ref{ccp} and \ref{subnetmoms}. This consists of computing their empirical analogs from data and matching them to simulated analogs. For inference, our results can be applied to construct GMM standard errors that account for network autocorrelation using the sandwich formula.

%------------------------
\subsubsection{Nonparametric Inference}\label{snonparam}
%------------------------

We next consider nonparametric inference on the average structural function (ASF), a parameter which provides a nonparametric measure of social interactions. Partition $(S_i^t, \tau_i^t) = (\tilde X_i^t, \tilde \alpha_i, \tilde\varepsilon_i^t)$, where $\tilde X_i^t$ contains all observed quantities, $\tilde\alpha_i$ all time-invariant unobservables, and $\tilde\varepsilon_i^t$ all time-varying unobservables that are independent across time. Thus we rewrite $U(S_i^t, \tau_i^t) = U(\tilde X_i^t, \tilde \alpha_i, \tilde\varepsilon_i^t)$. Using this representation, the ASF is given by
\begin{equation*}
  \mu(x) = \E[U(x, \tilde\alpha_i, \tilde\varepsilon_i^t)],
\end{equation*}

\noindent where $x$ is a vector of constants. Then if $\tilde X_i^t$ includes some function of the lagged outcome of neighbors, $\mu(x) - \mu(x')$ provides a nonparametric measure of peer effects. 

The ASF is typically not point identified in this context \citep{chamberlain2010binary}, so we use the partial identification approach of \cite{chernozhukov2013average}. Let $\tilde X_i = (\tilde X_i^t)_{t=0}^T$, whose support is required to be discrete. Define
\begin{align*}
  &\X^t(x) = \left\{ \tilde X_i\colon \tilde X_i^t=x, \tilde X_i^r \neq x \,\,\forall r<t, r = 1, \dots, T \right\},  \\
  &\X^\neq(x) = \left\{ \tilde X_i\colon \tilde X_i^t\neq x \,\,\forall t = 1, \dots, T \right\}. 
\end{align*}

\noindent Note that $\X^t(x)$ is the set of values of $\tilde X_i$ for which the $t$th component first equals $x$ at time $t$. Then $\left\{ \X^\neq(x), \X^0(x), \dots, \X^T(x) \right\}$ partitions the support of $\tilde X_i$. Let 
\begin{equation*}
  \hat Y_i(x) = \sum_{t=1}^T \bm{1}\{\tilde X_i \in \X^t(x)\} Y_i^t, \quad P_i(x) = \bm{1}\{\tilde X_i \in \X^\neq(x)\}, 
\end{equation*}

\noindent $\mu_\ell(x) = \E[\hat Y_i(x)]$, and $\mu_u(x) = \mu_\ell(x) + \E[P_i(x)]$. \cite{chernozhukov2013average} show that
\begin{equation*}
  \mu_\ell(x) \leq \mu(x) \leq \mu_u(x)
\end{equation*}

\noindent and that the bounds collapse to a point at a rate exponential in $T$.
% to condition on X_i^t=x, see clt_nf/calculations/14--.jpg

To use these bounds in practice, we construct sample analog estimators 
\begin{equation*}
  \hat\mu_\ell(x) = \frac{1}{n} \sum_{i=1}^n \hat Y_i(x) \quad\text{and}\quad \hat\mu_u(x) = \hat\mu_\ell(x) + \frac{1}{n} \sum_{i=1}^n P_i(x).
\end{equation*}

\noindent Given a joint CLT for these estimators, we can apply the method of \cite{woutersen2006simple} asymptotic version of the GMS test \citep[][p.\ 135]{andrews2010inference} to construct confidence intervals for the ASF using a HAC estimate of the variance.

As a first step to apply our CLT, we need to verify that the analog estimators fall within the set of moments satisfying \autoref{klocal}.  The agent statistic is $\psi_i = (\hat Y_i(x), P_i(x))$. By definition, $\hat Y_i(x)$ is a function of the outcome time series through $\{(Y_i^t, \tilde X_i^t)\}_{t=0}^T$. By \autoref{S2}, $\tilde X_i^t$ (in particular the subvector corresponding to observed components of $S_i^t$) is a function of outcomes only through the 1-neighborhood of agent $i$. Hence, $\hat Y_i(x)$ satisfies \autoref{klocal} for $K=1$. The same argument holds for $P_i(x)$.

%------------------------------------
\subsection{Network Formation}\label{snf}
%------------------------------------

This section introduces a (nonparametric) stochastic model of network formation. The model need not be known in practice, but it is used to prove our asymptotic results. Readers interested in inference methods can skip ahead to \autoref{svar}.

We assume each agent $i$ is endowed with a {\em network type} $(\rho_i, \mu_i)$, which is a time-invariant subvector of $\tau_i^t$. Additionally, each agent pair $(i,j)$ is endowed with a time-invariant pair-specific random utility shock $\zeta_{ij}$, i.i.d.\ across pairs and independent of all other model primitives. For all $i,j \in \mathcal{N}_n$ with $i\neq j$, potential links in $\bm{A}$ satisfy
\begin{equation}
  A_{ij} = \bm{1}\left\{ V(\norm{\rho_i-\rho_j}, \mu_i, \mu_j, \zeta_{ij}) > 0 \right\}, \label{modelnf}
\end{equation}

\noindent where $\norm{\cdot}$ is a norm on $\R^d$ and $V(\cdot)$ a real-valued latent-index function, which we will later assume is eventually decreasing in its first argument (\autoref{sparsity}).\footnote{Since $A_{ij}$ is an undirected network, we assume $V(\norm{\rho_i-\rho_j}, \mu_i, \mu_j, \zeta_{ij}) = V(\norm{\rho_i-\rho_j}, \mu_j, \mu_i, \zeta_{ji})$.} In our applications, $V(\cdot)$ is an unknown function because the usual object of interest is some feature of $U(\cdot)$. Likewise, network types may be unobserved by the econometrician. 

The second and third arguments of $V(\cdot)$ contain agent-specific characteristics that may influence link formation. If $V$ is monotonic in any of these elements, then this captures what \cite{graham2017econometric} refers to as {\em degree heterogeneity}, where agents with more attractive characteristics $\mu_i$ are likely to have more connections (high degree). Alternatively, $V(\cdot)$ could depend on $\mu_i,\mu_j$ through $\bm{1}\{\mu_i = \mu_j\}$, which captures {\em homophily} in $\mu_i$. Homophily refers to the widely observed tendency for similar individuals to associate. 

The first argument of $V(\cdot)$ requires homophily in {\em positions} $\rho_i$. This aspect of the model lends it a spatial dimension, which is essential for showing validity of the spatial HAC estimator. Under additional weak dependence conditions stated in \autoref{sdm}, spatial homophily implies that the correlation between $\psi_i$ and $\psi_j$ is smaller when $\norm{\rho_i-\rho_j}$ is larger. 

Note that the space on which agents are located need not correspond to physical space. They may be homophilous in other social dimensions, in which case we can define distance in terms of their social, rather than geographic, characteristics. In the case where positions are unobserved by the econometrician, the model is a nonlinear version of a latent space model commonly used in social network analysis \citep{hrh2002,breza_using_2017}. In these models, agents are positioned on a latent ``social space,'' and connections form at a higher rate among socially similar agents. In this case, clearly a spatial HAC cannot be computed, since positions are unobserved, but our results show that the network HAC can be used instead.

\begin{remark}
  The setup allows for correlation between $\mu_i$ and $\alpha_i$ (for example), in which case the network is endogenous. If agents are homophilous in $\mu_i$, then this captures {\em unobserved homophily}, a well-known hindrance to identifying social interactions \citep{shalizi2011homophily}. For example, identification of peer effects in product adoption is confounded by peers with similar product preferences $\alpha_i$ forming connections at a higher rate.
\end{remark}

\begin{remark}
  The model does not allow for strategic interactions in link formation, meaning that $V(\cdot)$ is not a function of $\bm{A}$.  Most papers in the econometric literature that address the problem of network endogeneity use similar models with no strategic interactions, viewing them heuristically as reduced-form approximations \citep[e.g.][]{auerbach2018identification,hsieh2016social,johnsson2019estimation}. \cite{leung2019normal} prove a CLT for static models of network formation with strategic interactions. In principle, our CLTs can be generalized to their model, but we do not pursue this generalization because it mildly complicates the proofs in predictable ways without providing any new intuition.
\end{remark}

%------------------------------------
\subsection{Large-Network Asymptotics}\label{lna}
%------------------------------------

This section formalizes the sequence of models along which we take limits in our asymptotic results. Recall that $\rho_i$, defined in the previous subsection, is a time-invariant subvector of $\tau_i^t$ for any $t$ and $\tau_i$ is the time series of $i$'s type $(\tau_i^t)_{t=0}^T$. Take the same time series but omit $\rho_i$ from each $\tau_i^t$, and call the result $Z_i$. Recall from the introduction that we define $\bm{\rho} = (\rho_i)_{i=1}^n$, and similarly define $\bm{Z}$ and $\bm{\zeta}$. Then the model is fully characterized by the tuple 
\begin{equation}
  (U, \lambda, V, \bm{\rho}, \bm{Z}, \bm{\zeta}), \label{finitemodel}
\end{equation}

\noindent where $U(\cdot)$ is the payoff function in \autoref{dymodel} and $V(\cdot)$ the latent index in the network formation model \autoref{modelnf}. The new term $\lambda(\cdot)$ concerns the initial conditions model; being a static model with strategic interactions, initial outcomes will be determined by a selection mechanism $\lambda(\cdot)$, whose formal definition we postpone to \autoref{ssm}.

To establish the validity of spatial HAC estimators, we need positions $\rho_i$ to be sufficiently removed from each other asymptotically, which is the usual assumption of {\em increasing domain asymptotics} that is standard in spatial econometrics. 

\begin{assump}[Increasing Domain]\label{as}
  Let $\tilde\rho_1, \tilde\rho_1, \dots$ be i.i.d., continuously distributed vectors in $\R^d$ with density $f$ bounded away from zero and infinity. Let $\omega_n = (n/\kappa)^{1/d}$ for some universal constant $\kappa>0$. Define $\rho_i = \omega_n\tilde\rho_i$ for all $i$.  The observed outcome time series is realized according to the $n$th model of the sequence $\{(U, \lambda, V, \bm{\rho}, \bm{Z}, \bm{\zeta})\}_{n \in \mathbb{N}}$.
\end{assump}

\noindent Under this sequence, we derive a CLT \autoref{clt} in \autoref{sdm} and asymptotic properties of the HAC estimators in \autoref{svar}.

\begin{remark}
  In spatial econometrics, it is typically assumed that positions $\rho_i$ are non-random and separated by a universal minimum distance \citep[e.g.][]{conley1999gmm,jenish_spatial_2012}. \autoref{as} is a slightly different model that follows the spatial graphs literature \citep{penrose2003} and is also sometimes used in spatial statistics \citep{lahiri2006resampling}. Both assumptions have the same implication, that in the limit, any ball of fixed radius centered at an agent's position contains only an asymptotically finite number of other agents' positions. If $\omega_n$ diverges faster (slower) than the stated rate, then this ball will be asymptotically empty (save for the central agent); if it diverges slower, then it will contain an infinite number of agents in the limit. 
\end{remark}

%----------------------------------------------------------------------
\section{Static Model of Social Interactions}\label{smodel}
%----------------------------------------------------------------------

This section introduces a static analog of the model studied in \autoref{sdynamic}, which is useful when the econometrician observes a snapshot of a large network. As before, we let $\mathcal{N}_n = \{1, \dots, n\}$ denote the set of agents, which are connected through an undirected network $\bm{A}$ realized according to the model in \autoref{snf}. Each agent $i$ is endowed with a type $\tau_i = (X_i, \varepsilon_i)$, i.i.d.\ across agents, where $X_i$ is observed by the econometrician and $\varepsilon_i$ unobserved. Agents take a binary action, and payoffs may depend on the actions taken by others in the network. Agent $i$'s observed action satisfies
\begin{equation}
  Y_i = \bm{1}\left\{ U(S_i, \tau_i) > 0 \right\} \label{model}
\end{equation}

\noindent where the net payoff function $U(\cdot)$ depends on $i$'s type and a finite-dimensional vector of statistics
\begin{equation*}
  S_i \equiv S(Y_{-i}, \tau_i, \tau_{-i}, A_i, A_{-i}).
\end{equation*}

\noindent This captures strategic interactions through its dependence on $Y_{-i}$. Model \autoref{model} states that observed actions are realized according to a pure-strategy Nash equilibrium; agents choose the action that maximizes payoffs given the actions of others in the network. 

\begin{example}\label{e3}
  \cite{bramoulle2009identification} study a network analog of the standard linear-in-means model \citep{manski1993identification}, which is a model with continuous outcomes. Outcomes depend on the average outcome of peers (endogenous peer effects) and the average characteristics of peers (exogenous peer effects). The analogous specification for discrete choice is
  \begin{equation*}
    U(S_i, \tau_i) = S_i'\theta_1 + X_i'\theta_{-1} + \varepsilon_i, \quad S_i = \left( \frac{\sum_{j\neq i} A_{ij} Y_j}{\sum_{j\neq i} A_{ij}}, \frac{\sum_{j\neq i} A_{ij} X_j}{\sum_{j\neq i} A_{ij}} \right) 
  \end{equation*}

  \noindent \citep{brock2001discrete,xu2018social}. The first component of $S_i$ is the average action taken by neighbors. We can also consider type-weighted versions of the average action or nonlinear functions of $Y_{-i}$ and $A_i$ such as the minimum or maximum action \citep{hoxby2005taking}.
\end{example}

Model \autoref{model} is entirely analogous to the dynamic model \autoref{dymodel} but with contemporaneous rather than lagged dependent variables entering $S_i$. Hence, we can also consider analogs of Examples \ref{e1} and \ref{e2} in the static setting. In all of these examples, strategic interactions only operate through network neighbors of the ego $i$. We next impose this restriction more generally. Recall that $\mathcal{N}_{\bm{A}}^-(i,1)$ is the $1$-neighborhood of $i$, excluding $i$ herself.

\begin{assump}[Local Interactions]\label{S}
  For all $n$ and $i \in \mathcal{N}_n$,
  \begin{equation*}
    S_i = S(Y_{\mathcal{N}(i)}, \tau_i, \tau_{\mathcal{N}(i)}, A_i) \quad\text{for}\quad \mathcal{N}(i) \equiv \mathcal{N}_{\bm{A}}^-(i,1).
  \end{equation*}
\end{assump}

\noindent This states that $S_i$ is only a function of its arguments through agents connected to $i$. 

Unlike the typical linear-in-means model for continuous outcomes, the nonlinearity of the discrete choice model typically gives rise to multiple equilibria \citep{brock2001discrete}. This is obvious when $n=2$ and the two agents connected, since this is a $2 \times 2$ game of complete information. With multiple equilibria, the econometric model is incomplete in the sense that a reduced form does not (yet) exist \citep{tamer2003incomplete}. To complete the model, we follow the empirical games literature and introduce a selection mechanism. Whether this is required to be known to the econometrician depends on the inference procedure, as discussed in the next subsection. Recall that $\bm{\tau} = (\tau_i)_{i=1}^n$. Let $\mathcal{E}(\bm{A},\bm{\tau}) \subseteq \{0,1\}^n$ be the set of Nash equilibria, i.e.\ the set of binary outcome vectors such that, for each $\bm{Y} \in \mathcal{E}(\bm{A},\bm{\tau})$, the $i$th component $Y_i$ satisfies \autoref{model} for each $i$. 

\begin{assump}[Selection Mechanism]\label{select}
  (a) A Nash equilibrium exists for any network size $n$, i.e.\ $|\mathcal{E}(\bm{A},\bm{\tau})| \geq 1$. (b) There exists a {\em selection mechanism} $\lambda(\bm{A},\bm{\tau})$ with range $\mathcal{E}(\bm{A},\bm{\tau})$ such that $\bm{Y} = \lambda(\bm{A},\bm{\tau})$ for any $n$.
\end{assump}

\noindent Part (b) defines the selection mechanism as a reduced form mapping from the model primitives to the observed equilibrium outcome. Note that if $U(\cdot)$ does not vary in $S_i$, then $\abs{\mathcal{E}(\bm{A},\bm{\tau})} = 1$, since this is simply a discrete-choice model with no strategic interactions, in which case $\lambda(\cdot)$ is trivial.  In economic terms, $\lambda(\cdot)$ represents the process by which agents coordinate on a Nash equilibrium to play in the observed data. A simple example is a function that picks an element of $\mathcal{E}(\bm{A},\bm{\tau})$ uniformly at random. This is an econometric model used by \cite{bjorn1984simultaneous} and \cite{soetevent2007discrete}.\footnote{To see how this is formally represented in our notation, without loss of generality let the first component of $\varepsilon_i$, for each $i$, be a random variable $\gamma_i$ uniformly distributed on $[0,1]$ that is payoff-irrelevant. That is, it does not enter $U(\cdot)$. Partition the unit interval into $\lvert \mathcal{E}(\bm{A},\bm{\tau}) \rvert$ equally sized intervals and arbitrarily order the elements of $\mathcal{E}(\bm{A},\bm{\tau})$. Then let $\lambda(\cdot)$ be the function that selects the $k$th equilibrium if $\gamma_i$ (for any arbitrarily chosen $i$, say $i=1$) lies in the $k$th interval of the partition.} Another example is myopic best-response dynamics.
  
\begin{example}\label{mrbd}
  As discussed in \autoref{sdynamic}, the microeconomic literature on dynamic models of social influence predominantly considers models in which agents react myopically to the decisions of their peers in the previous period. Formally, fix $\bm{A},\bm{\tau}$ and an arbitrary $\bm{Y}^0 \in \{0,1\}^n$. Generate $Y^1 \in \{0,1\}^n$ by setting
  \begin{equation*}
    Y_i^1 = \bm{1}\left\{ U(S(Y_{-i}^0,\tau_i,\tau_{-i},A_i,A_{-i}), \tau_i) > 0 \right\}
  \end{equation*}

  \noindent for each $i \in \mathcal{N}_n$, and likewise generate $Y^2, Y^3, \dots$. This process is often referred to as {\em myopic best-response dynamics}. It is well known that if $\bm{Y}^0 = (1, \dots, 1)$, then under a game of strategic complements, this process converges to the ``largest'' Nash equilibrium $\bm{Y}^*$ in a precise sense \citep[e.g.][p.\ 1279--80]{jia2008happens}. In that case, this process constitutes a mapping $\lambda(\cdot)$ from the primitives $(\bm{A},\bm{\tau})$ to a unique outcome $\bm{Y}^* \in \mathcal{E}(\bm{A},\bm{\tau})$.
\end{example}

%------------------------------------
\subsection{Network Moments}\label{snetmoms}
%------------------------------------

Similar to \autoref{sdynetmoms}, we consider network moments that are averages of $\R^m$-valued {\em agent statistics} of the form
\begin{equation*}
  \frac{1}{n} \sum_{i=1}^n \psi_i \quad\text{for}\quad \psi_i \equiv \psi(Y_i, Y_{-i}, A_i, A_{-i}, \tau_i, \tau_{-i}). 
\end{equation*}

\noindent Our objective is to establish a CLT \autoref{clt} and validity of the HAC estimators under increasing domain asymptotics, as in the dynamic case (\autoref{as}). We assume $\psi(\cdot)$ satisfies \autoref{klocal}, now under the new notation of the static setting. This condition simply states that $\psi_i$ is only a function of its arguments through $i$'s $K$-neighborhood in $\bm{A}$, which is satisfied for a wide variety of network moments useful for inference on social interactions. The remainder of this subsection provides illustrative examples.

\begin{example}
  A trivial example of a network moment is $n^{-1}\sum_{i=1}^n Y_i$, which is the empirical choice probability and satisfies \autoref{klocal} for $K=0$. A weighted version of the choice probability moment is used in the simulated method of moments estimator of \cite{xl2015} discussed in \autoref{ssmm}. Another example is subnetwork moments
  \begin{equation*}
    \frac{1}{n} \sum_{i_1=1}^n \cdots \sum_{i_\ell=1}^n \bm{1}\{Y_H = y_H, A_H=a_H\}, 
  \end{equation*}

  \noindent where $H = \{i_1, \dots, i_\ell\}$, $y_H \in \{0,1\}^\ell$, $a_H$ is a connected network on $H$, and $A_H$ is the subnetwork of $\bm{A}$ on $H$. Its expectation is proportional to the probability that agents in $H$ form subnetwork $a_H$ and choose outcomes $y_H$. As discussed below, this satisfies $K$-locality and can be used to construct moment inequalities for structural inference on $U(\cdot)$.
\end{example}

%------------------------
\subsubsection{Method of Moments}\label{ssmm}
%------------------------

The next two subsections discuss applications to parametric inference on $U(\cdot)$ that utilize moments falling within the scope of \autoref{klocal}. \cite{xl2015} consider a linear latent index model 
\begin{equation*}
  U(S_i, \tau_i) = \theta_1 \frac{\sum_{j\neq i} A_{ij} Y_j}{\sum_{j\neq i} A_{ij}} + X_i'\theta_{-1} + \varepsilon_i,
\end{equation*}

\noindent where $\theta_1$ captures endogenous peer effects. They assume $\theta_1 \geq 0$, which implies that observed outcomes constitute a Nash equilibrium of a supermodular game. It is well known that the set of equilibria forms a complete lattice and therefore that there exists a ``largest'' equilibrium. For estimation, \cite{xl2015} assume the following.

\begin{assump}\label{netexog}
  (a) Realized outcomes correspond to the ``largest'' Nash equilibrium. (b) $(\bm{X},\bm{A}) \indep \bm{\varepsilon}$. (c) The distribution of $\varepsilon_1$ given $X_1$ is known up to $\theta$.
\end{assump}

\noindent Part (a) is equivalent to assuming that $\lambda(\cdot)$ is given as in \autoref{mrbd}. By assuming a particular selection mechanism, this enables simulation of the model moments. Part (b) assumes an exogenous network.

The authors propose to estimate $\theta$ using simulated method of moments (SMM) based on the conditional choice probability, namely $n^{-1}\sum_i \psi_i$ for
\begin{equation*}
  \psi_i = (Y_i - \prob_\theta(Y_i=1 \mid \bm{X}, \bm{A})) h(X_i,X_{-i},A_i,A_{-i}). 
\end{equation*}

\noindent Here, $\prob_\theta$ refers to the probability under the model with structural parameters given by $\theta = (\theta_1, \theta_{-1})$, and $h(\cdot)$ is a vector-valued instrument function that converts to unconditional moments. In practice, the conditional probability can be simulated due to \autoref{netexog}. For inference, they propose to use the parametric bootstrap.

For the SMM estimator to have a normal limit, a key property to verify is a central limit theorem for the moments themselves (this is assumption (d) of the authors' Theorem B.1). \cite{xl2015} establish a CLT for the moments by invoking a CLT for near-epoch dependent data due to \cite{jenish_spatial_2012}. Using our CLT instead requires significantly weaker restrictions on the network formation model.\footnote{In particular, we do not require their Assumption 7, which states that if two agents are linked, then their spatial distance must fall below some universal constant.} Additionally, our CLT applies to other network moments that can be used for estimation, such as subnetwork moments. Finally, our results enable the use of HAC variance estimators for inference. This consists of taking the usual GMM sandwich formula but replacing the sample variance of the moments in the middle of the sandwich with one of our HAC estimators.

For our CLT to apply to these moments, we need to verify $K$-locality (\autoref{klocal}). This holds, for example, under the restriction
\begin{equation*}
  h(X_i,X_{-i},A_i,A_{-i}) = h(X_i,X_{\mathcal{N}_{\bm{A}}^-(i,1)},A_i),
\end{equation*}

\noindent in which case \autoref{klocal} is satisfied for $K=1$. The restriction states that the instruments only depend on the observed types and links involving neighbors of $i$. In their simulation study, \cite{xl2015} choose $h(X_i,X_{-i},A_i,A_{-i}) = (1, X_i, \sum_{j\neq i} A_{ij} X_j / \sum_{j\neq i} A_{ij})$, which satisfies this condition. This choice is likely motivated by the intuition in \cite{bramoulle2009identification} that the average covariates of peers is an instrument for the endogenous peer effect.

%------------------------
\subsubsection{Set Inference}\label{ssubinf}
%------------------------

Without imposing restrictions on equilibrium selection, $\theta$ is typically partially identified, and the identified set can be characterized in terms of moment inequalities \citep{beresteanu2011sharp,galichon2011set}. For discrete games on networks, \cite{li2016partial} propose moment inequalities that provide a conservative characterization of the identified set. These are based on subnetwork moments of the form $\prob(Y_H = y_H \mid X_H, A_H)$ for $H\subseteq \mathcal{N}_n$, which gives the conditional joint distribution of outcomes of agents in $H$ given the subnetwork $A_H$.

The main identification result of \cite{li2016partial} constructs a function $G(\cdot)$ such that
\begin{equation}
  \prob(Y_H = y_H \mid X_H, A_H) - G(y_H, X_H, A_H; \theta_0) \leq 0 \label{mi}
\end{equation}

\noindent for any $y_H \in \{0,1\}^{\lvert H\rvert}$ when $\theta_0$ is the true parameter. The bound $G(\cdot)$ can be computed via simulation, while the conditional probability can be estimated nonparametrically, provided a law of large numbers holds.

While \cite{li2016partial} focus on estimation, we next discuss how their bounds may be used for inference on $\theta_0$. Since their bounds come from conditional moment inequalities, we can apply a procedure due to \cite{andrews2013inference}. To correct for autocorrelation, we need to use the asymptotic version of their test, consisting of steps 1 and 2 in their section 9, with a valid variance estimator in place of the sample variance. Asymptotic validity requires a CLT for the network moments, and our results can be applied for this purpose.

The procedure of \cite{andrews2013inference} first converts to unconditional moment inequalities by multiplying moments with real-valued instrument functions $h(X_H, A_H)$, which can be done without loss of generality by using a large enough set of functions (see their section 3.3). Define
\begin{equation*}
  m(y_H; \theta) = \E\left[\left(\bm{1}\{Y_H = y_H\} - G(y_H, X_H, A_H; \theta)\right) h(X_H, A_H)\right].
\end{equation*}

\noindent Then by \autoref{mi}, we have $m(y_H; \theta) \leq 0$ when $\theta$ is the true parameter.

Consider the ``sample analog'' of these moments
\begin{equation*}
  \hat m(y_H; \theta) = \frac{1}{n} \sum_{i_1=1}^n \cdots \sum_{i_\ell=1}^n \left( \bm{1}\{Y_H = y_H\} - G(y_H, X_H, A_H; \theta) \right) h(X_H, A_H),
\end{equation*}

\noindent where $H = \{i_1, \dots, i_\ell\}$; see \autoref{sparsescale} for discussion of the $n^{-1}$ scaling. Its expectation is proportional to $m(y_H; \theta)$, where the constant of proportionality does not depend on $n$ or $\theta$. Since $A_H$ has finite support, we can consider instrument functions of the form 
\begin{equation*}
  h(X_H, A_H) = \tilde h(X_H) \bm{1}\{A_H = a_H\},
\end{equation*}

\noindent where $a_H$ is a network on $H$. 

The sample moments can be rewritten as $n^{-1} \sum_i \psi_i$ for
\begin{equation*}
  \psi_i = \sum_{i_2=1}^n \cdots \sum_{i_\ell=1}^n \left( \bm{1}\{Y_H = y_H\} - G(y_H, X_H, A_H; \theta) \right) h_k(X_H, A_H), 
\end{equation*}

\noindent where $H = \{i, i_2, \dots, i_\ell\}$. To apply our CLT, we need to verify \autoref{klocal}. If we use any subnetwork moment(s) in which $a_H$ is a connected network, then the assumption holds for $K=\ell$.%\footnote{The set of all connected subnetworks already provides more than enough moments to use in practice, and there does not seem to be a reason to prefer disconnected to connected subnetworks in general. Nonetheless, if the econometrician chooses to use moments in which the subnetwork $a_H$ is disconnected, our CLT cannot be applied. In this case, we require a generalization of the theorem of \cite{leung2019normal} to stabilizing U-statistics, which is beyond the scope of this paper.}

%----------------------------------------------------------------------
\section{Variance Estimators}\label{svar}
%----------------------------------------------------------------------

We postpone formal conditions for CLTs for the static and dynamic models to \autoref{ssm} and \autoref{sdm}, respectively. This section is concerned with the properties of spatial and network HAC estimators for the variance $\bm{\Sigma}_n$ in \eqref{clt}. The notation used is applicable to either the static or dynamic model. The first two subsections discuss the validity of these estimators under first-order stationarity conditions. To our knowledge, such conditions are always used to establish consistency of HAC estimators in time series and spatial econometrics. However, they are typically only satisfied in GMM-type settings and not, for example, in moment inequality models or many nonparametric models. The third subsection is concerned with inference when first-order stationarity fails.

%------------------------------------
\subsection{Spatial HAC}\label{sstation}
%------------------------------------

The standard spatial HAC estimator is given by
\begin{equation*}
  \hat{\bm{\Sigma}}_\rho = \frac{1}{n} \sum_{i=1}^n \sum_{j=1}^n (\psi_i-\bar{\psi}) (\psi_j-\bar{\psi})' K( (\rho_i-\rho_j) / h_n),
\end{equation*}

\noindent where $h_n \in \R_+$ is the bandwidth, $K(\cdot)$ a real-valued kernel function with domain $\R^d$, and $\bar{\psi} = n^{-1}\sum_{i=1}^n \psi_i$, the vector of network moments of interest. In cases where $\E[\bar{\psi}]$ is a known constant, for example zero in GMM models, we can replace $\bar{\psi}$ in the HAC estimator with this constant. In practice, it is common to use a product kernel, which has the form
\begin{equation*}
  K((\rho_i-\rho_j) / h_n) = \prod_{k=1}^d \tilde K( (\rho_{ik}-\rho_{jk}/ h_n) ),
\end{equation*}

\noindent where $d$ is the dimension of $\rho_i$, $\rho_{ik}$ is the $k$th component of $\rho_i$, and $\tilde K$ is a real-valued kernel function with domain $\R$. A typical choice for $\tilde K(\cdot)$ is the Bartlett kernel $\tilde K(x) = (1 - \abs{x}) \bm{1}\{\abs{x} \leq 1\}$. For other examples, see e.g.\ \cite{andrews1991heteroskedasticity}, p.\ 821. 

In practice, it is best to show robustness of the standard errors across a reasonable range of bandwidths, as in \cite{conley2010learning}. For example, one can recompute the HAC estimator for several bandwidth values in a neighborhood of some reference value. This value may be obtained from domain knowledge, say by examining the distribution of distances. Alternatively, one can calibrate $h_n$ by Monte Carlo simulation with a parametric submodel. Data-driven choice of $h_n$ is an important but difficult topic for future research and will not be addressed in this paper.

We next study the theoretical properties of $\hat{\bm{\Sigma}}_\rho$. To our knowledge, all existing results on the validity of HAC estimators, whether for time series, spatial processes, or network settings, require first-order stationarity \citep{andrews1991heteroskedasticity,conley1999gmm,kojevnikov2019limit}. In our setting, this corresponds to the following.

\begin{assump}[First-Order Stationarity]\label{station}
  $\E[\psi_1] = \E[\psi_1 \mid \rho_1]$.\footnote{In spatial econometrics, the usual stationarity condition is that $\E[\psi_i \mid \bm{\rho}]$ does not depend on $i$. This is analogous to ours except conditional on the set of positions, since they are treated as fixed \citep[e.g.][]{conley1999gmm,jenish2016}. Our asymptotic results are unconditional, which is why the assumption is slightly different.}
\end{assump}

\noindent This says that the absolute value of an agent's position is mean independent of her agent statistic. Its technical purpose is that the long-run covariance aggregates over covariances of agents $i$ and $j$ conditional on their positions. This covariance is a function of $\E[\psi_i \mid \rho_i]$, which needs to be consistently estimated. Under \autoref{station}, this is possible using the sample mean $\bar{\psi}$.

\begin{example}\label{fosex}
  \autoref{station} typically holds in GMM-type settings. For example, consider the moments in \autoref{ssmm}. We have 
  \begin{equation}
    \E[\psi_i \mid \rho_i] = \E\big[Y_i h(X_i,X_{-i},A_i,A_{-i}) \mid \rho_i\big] - \E\big[ \E_\theta[Y_i \mid \bm{X},\bm{A}] h(X_i,X_{-i},A_i,A_{-i}) \mid \rho_i], \label{vj023becx2}
  \end{equation}

  \noindent where $\E_\theta[\cdot]$ is the expectation under the model with structural parameters given by $\theta$ and $\E[\cdot]$ is the expectation under the true model. Also,
  \begin{align*}
    \E\big[Y_i h(X_i,X_{-i},A_i,A_{-i}) \mid \rho_i\big] 
    &= \E\big[ \E[Y_i h(X_i,X_{-i},A_i,A_{-i}) \mid \rho_i, \bm{X}, \bm{A}] \mid \rho_i\big] \\
    &= \E\big[\E[Y_i \mid \rho_i, \bm{X}, \bm{A}] h(X_i,X_{-i},A_i,A_{-i}) \mid \rho_i\big] \\
    &= \E\big[\E[Y_i \mid \bm{X}, \bm{A}] h(X_i,X_{-i},A_i,A_{-i}) \mid \rho_i\big],
  \end{align*}

  \noindent where the last line holds if positions only enter payoffs through the network $\bm{A}$; \autoref{netexog}(b) suffices for this. If $\theta$ is the true parameter, then $\E[Y_i \mid \bm{X}, \bm{A}] = \E_\theta[Y_i \mid \bm{X}, \bm{A}]$, so \eqref{vj023becx2} $=\zero$ and \autoref{station} holds. A similar argument can be used to verify the assumption for the application in \autoref{sdysmm}.
\end{example}

\begin{assump}[HAC Kernel]\label{kernel}
  $K(0)=1$; $K(x) = 0$ for all $x \in \R^d$ such that $\norm{x} > 1$; $\int \abs{K(x)} \,\text{d}x < \infty$; $K$ is continuous at zero; and $K^* \equiv \sup_x \abs{K(x)} < \infty$.
\end{assump}

\noindent This assumption imposes standard restrictions on the kernel. Now, let $\lambda_\text{min}(\bm{M})$ denote the smallest eigenvalue of the matrix $\bm{M}$ and $\norm{\bm{M}}$ the max norm of $\bm{M}$.

\begin{theorem}\label{spatialhac}
  Suppose $h_n = O(n^{1/(3d)})$ and $h_n \rightarrow \infty$. Under Assumptions \ref{station} and \ref{kernel} and the conditions required for a CLT (Theorems \ref{staticclt} and \ref{dynamicclt} in the static and dynamic cases, respectively), $\norm{\hat{\bm{\Sigma}}_\rho - \tilde{\bm{\Sigma}}_n} \plimarrow 0$ for some sequence of matrices $\{\tilde{\bm{\Sigma}}_n\}_{n\in\mathbb{N}}$ such that $\lambda_\text{min}(\bm{\Sigma}_n - \tilde{\bm{\Sigma}}_n) \geq 0$ for all $n$. If there exists a constant vector $c$ such that $\E[\psi_1]=c$ for any $n$, then $\bm{\Sigma}_n = \tilde{\bm{\Sigma}}_n$ for all $n$.
\end{theorem}
\begin{proof}
  See \autoref{sstationpf} for the formal proof and a proof sketch.
\end{proof}

\noindent The second conclusion of the theorem states that $\hat{\bm{\Sigma}}_\rho$ is consistent if $\E[\psi_1]$ does not vary with $n$. This is satisfied if $\psi(\cdot)$ is a moment function used for GMM estimation, since then $\E[\psi_1]=\zero$. More generally, if the moment is centered at the right conditional expectation, usually both \autoref{station} and $\E[\psi_1]=c$ hold. The applications in \autoref{sdysmm} and \autoref{ssmm} are centered. 

When moments are not centered, typically $\E[\psi_1]$ varies with $n$, and the first conclusion of \autoref{spatialhac} shows that $\hat{\bm{\Sigma}}_\rho$ is asymptotically conservative in this case, which appears to be a new result. However, this also requires first-order stationarity to hold, which is often violated when moments are uncentered. In moment inequality models (\autoref{snonparam} and \autoref{ssubinf}) and often in nonparametric models, moments are typically not centered, so the results here are inapplicable. We discuss alternatives in \autoref{snostat}.

%------------------------------------
\subsection{Network HAC}
%------------------------------------

Next we study the network HAC estimator for $\bm{\Sigma}_n$ used by \cite{acemoglu2015state} and \cite{eckles2016estimating}. This essentially consists of taking the usual spatial HAC estimator and replacing spatial distance with path distance. Recall from \autoref{sintro} the definition of path distance $\ell_{\bm{A}}(i,j)$. Let $h_n \in \R_+$ be a bandwidth and $K\colon \R \rightarrow \R$ a kernel function. The network HAC estimator is given by 
\begin{equation*}
  \hat{\bm{\Sigma}}_{\bm{A}} = \frac{1}{n} \sum_{i=1}^n \sum_{j=1}^n (\psi_i-\bar{\psi}) (\psi_j-\bar{\psi})' K( \ell_{\bm{A}}(i,j) / h_n).
\end{equation*}

\noindent For consistency, the bandwidth $h_n$ will be required to grow essentially at a logarithmic rate. In practice, it is best to either show robustness of the standard errors across a reasonable range of bandwidths in some neighborhood of a reference value, as in \cite{conley1999gmm}. This value may be obtained from domain knowledge, say by examining typical path lengths in the network. Another reference value is simply setting $h_n = \log n$, which is used in our simulation study. Finally, one can calibrate $h_n$ by Monte Carlo simulation with a parametric submodel. 

The basic idea behind our proof for the validity of the network HAC estimator is as follows. In our model, agents are homophilous in positions, which implies that agents that are close in path distance should also be close in spatial distance. This suggests we can use the former to approximate the latter and expect $\hat{\bm{\Sigma}}_{\bm{A}} \approx \hat{\bm{\Sigma}}_\rho$.

\begin{remark}[Related Literature]\label{relit}
  A recent paper by \cite{kojevnikov2019limit} prove consistency of $\hat{\bm{\Sigma}}_{\bm{A}}$ for the variance of a class of moments satisfying a novel notion of network weak dependence they call ``conditional $\psi$-weak dependence.'' This is a modification of a concept in the time series literature using path distance in place of temporal distance. The condition is distinct from stabilization, the weak dependence conditions used in our CLTs, which do not condition on the network.\footnote{See \autoref{smaster} for the formal stabilization conditions. These are taken from \cite{leung2019normal}, which are, in turn, modifications of assumptions first proposed in the literature on geometric graphs \citep{penrose2003}.} We view these as complementary contributions, since they consider applications for which stabilization cannot be used but do not consider applications to social interactions models. Additionally, \cite{kojevnikov2019limit} impose general conditions on the network structure without assuming a specific model of network formation. We instead assume a particular (nonparametric) model of network formation, which enables us to derive lower-level conditions. 
\end{remark}

\begin{assump}[First-Order Stationarity]\label{networkstation}
  Let $N$ be any random variable supported on the natural numbers, independent of all other primitives. Suppose the number of agents is given by $N$, and let $\psi_1^N$ denote agent 1's statistic (in either the static or dynamic model), $\bm{X}^N = (X_i)_{i=1}^N$, and $\bm{A}^N$ the network. Then $\E[\psi_1^N] = \E[\psi_1^N \mid \bm{X}^N, \bm{A}^N, N]$.
\end{assump}

\noindent The simplest way to understand the assumption is to consider the case where $N = n$ a.s., which is the model originally specified in \autoref{sdynamic} and \autoref{smodel}. Then this reduces to $\E[\psi_1] = \E[\psi_1 \mid \bm{X}, \bm{A}]$, which is analogous to \autoref{station}. As with the latter assumption, stationarity is satisfied in GMM-type settings or more generally when moments are centered at their conditional (on $\bm{X},\bm{A}$) expectations. This includes the applications in \autoref{sdysmm} and \autoref{ssmm}. Whether consistent variance estimators can be obtained when \autoref{networkstation} fails is an open question.

For technical reasons, we need stationarity to hold not only for models where $N=n$ but also for $N \sim \text{Poisson}(n)$, which is why the statement of \autoref{networkstation} is more complicated. This does not seem to rule out any applications of interest. Similar to the spatial HAC setting, the technical purpose of this assumption is that the long-run covariance aggregates over the conditional covariances of agents $i$ and $j$. Hence, it is a function of $\E[\psi_i \mid \bm{X}, \bm{A}]$, which needs to be consistently estimated. Under \autoref{networkstation}, this is possible using the sample mean $\bar{\psi}$.

In order to approximate spatial with path distance, we require the following intermediate-level condition on the network formation model.

\begin{assump}[Path Distance]\label{pd}
  (a) There exists $c>0$ such that 
  \begin{equation*}
    \prob(\ell_{\bm{A}}(i,j) > c \norm{\rho_i-\rho_j} \mid \ell_{\bm{A}}(i,j) < \infty) = O(n^{-3/4}).
  \end{equation*}
  
  \noindent (b) For any $c > 0$,
  \begin{equation*}
    \lim_{\epsilon\rightarrow\infty} \limsup_{n\rightarrow\infty} \prob(\ell_{\bm{A}}(i,j) > \epsilon \mid \ell_{\bm{A}}(i,j) < \infty, \norm{\rho_i-\rho_j}) \bm{1}\{\norm{\rho_i-\rho_j} < c\} = 0 \text{ a.s.}
  \end{equation*}
\end{assump}

\noindent This assumption concerns the relationship between spatial and path distance. Part (a) states that, for connected agents, their path distance is not asymptotically much larger than spatial distance with high probability. Part (b) requires the path distance of connected agents to be asymptotically bounded if their spatial distance is likewise bounded. 

This seems to be a reasonable condition since agents are spatially homophilous, so agents that are spatially distant should be similarly distant in the network. However, verification of this condition is challenging because the relationship between spatial and path distance is a complicated graph-theoretic problem. To our knowledge, in random graph theory, the requisite results have only been developed for the random geometric graph model, which is a widely studied spatial graph model \citep{penrose2003}. In this model, two agents connect if and only if their spatial distance falls below some fixed threshold. In \autoref{rggpd}, we draw on recent theoretical results to verify \autoref{pd} for this model.

While no corresponding results are presently available for more general spatial graphs, it seems reasonable to conjecture that \autoref{pd} holds for the general class of models in \autoref{snf} under the regularity condition that the linking probability decays exponentially with distance, a condition we require anyway for a CLT in \autoref{sparsity}. The basis of our conjecture is that proving a result under smooth exponential decay is usually expected once the corresponding result has been proven for a hard threshold decay, like the random geometric graph. For instance, compare limit theory for $M$-dependence, where neighbors beyond distance $M$ are uncorrelated with the ego, and $\alpha$-mixing, where the covariance instead may decay smoothly with distance. 

\begin{assump}[Network Exogeneity]\label{netcexog}
  For any $n\in\mathbb{N}$, $\bm{\varepsilon} \indep \bm{A} \mid \bm{X}$.
\end{assump}

\noindent This is another new condition not required by the spatial HAC. In the static case, it rules out network endogeneity, since $\bm{A}$ must be independent of unobservables given observables. Thus, unobserved homophily is not permitted.\footnote{The framework of \cite{kojevnikov2019limit} also seems to rule out network endogeneity, since errors terms must be weakly dependent conditional on the network. Under endogeneity, the conditional dependence structure appears very difficult to characterize.} In the dynamic case, recall that $X_i$ collects the observed covariates and the fixed effects, while $\varepsilon_i$ collects the idiosyncratic errors. Then \autoref{netcexog} implies that the errors are independent of the network, but network endogeneity is still allowed through dependence between the fixed effect and the network. 

\begin{theorem}\label{networkhac}
  Suppose $h_n \rightarrow \infty$ and $\log h_n / \log n \rightarrow 0$. Under Assumptions \ref{kernel}--\ref{netcexog} and the conditions required for a CLT (Theorems \ref{staticclt} and \ref{dynamicclt} in the static and dynamic cases, respectively), $\norm{\hat{\bm{\Sigma}}_{\bm{A}} - \tilde{\bm{\Sigma}}_n} \plimarrow 0$ for some sequence of matrices $\{\tilde{\bm{\Sigma}}_n\}_{n\in\mathbb{N}}$ such that $\lambda_\text{min}(\bm{\Sigma}_n - \tilde{\bm{\Sigma}}_n) \geq 0$ for all $n$. If there exists a constant vector $c$ such that $\E[\psi_1]=c$ for all $n$, then $\bm{\Sigma}_n = \tilde{\bm{\Sigma}}_n$ for all $n$.
\end{theorem}
\begin{proof}
  See \autoref{snethacpf} for the formal proof and a proof sketch.
\end{proof}

\noindent The conclusions are the same as those of \autoref{spatialhac}. When moments are centered, typically first-order stationarity and $\E[\psi_1]=c$ hold, in which case the result delivers consistency. On the other hand, if $\E[\psi_1]$ varies with $n$, then the estimator is asymptotically conservative. 

\begin{remark}[Bandwidth Rate]
  The bandwidth is required to grow at a sub-polynomial rate, unlike the spatial HAC. Actually, inspection of the proof shows that the result holds if $h_n$ diverges at a rate that is polynomial in $n$ with sufficiently small degree, but the degree depends on the dimension of $\rho_1$, which is unknown. Thus, a sub-polynomial rate is the price to pay for not observing positions.
\end{remark}

\begin{remark}[Network Endogeneity]\label{rnetendog}
  \autoref{netcexog} is important to establish \autoref{networkhac}. It serves to ensure that disconnected agents are uncorrelated, i.e.\ $\cov(\psi_i, \psi_j \mid \bm{X}, \bm{A}) = 0$ if $\ell_{\bm{A}}(i,j)=\infty$. If this covariance were non-zero, then the network HAC would be biased because it cannot account for the covariance between pairs of disconnected agents, since the kernel weight multiplying $(\psi_i-\bar{\psi}) (\psi_j-\bar{\psi})'$ in $\hat{\bm{\Sigma}}_{\bm{A}}$ would be zero. It may seem intuitive for this covariance to be zero, and it certainly is true when the network is exogenous. However, without exogeneity, the problem is that we are conditioning on the endogenous network $\bm{A}$, which can induce nonzero correlation between the agents. 
\end{remark}

\begin{remark}[Positive Semidefiniteness]
  In our simulation study, $\hat{\bm{\Sigma}}_{\bm{A}}$ is always positive semidefinite, but in general, this is not guaranteed in finite sample. To ensure positive semidefiniteness, we can modify the estimator using a correction proposed by \cite{kojevnikov2019bootstrap}. Let $\bm{Q}_n \bm{\Lambda}_n \bm{Q}_n'$ be the eigendecomposition of $\hat{\bm{\Sigma}}_{\bm{A}}$, $c$ a positive constant, and $\bm{I}_n$ the $n\times n$ identity matrix. Define
  \begin{equation*}
    \hat{\bm{\Sigma}}_{\bm{A}}^+ \equiv \bm{Q}_n \max\{\bm{\Lambda}_n, c\bm{I}_n\} \bm{Q}_n',
  \end{equation*}

  \noindent where ``max'' refers to the element-wise maximum. Then $\hat{\bm{\Sigma}}_{\bm{A}}^+$ is positive semidefinite because the smallest eigenvalue is at least $c>0$ by construction. Consistency requires $c \rightarrow 0$ as $n \rightarrow \infty$ \citep[see the proof of Proposition B.2 in][]{kojevnikov2019bootstrap}. In practice, we might choose, say, $c=0.1$.
\end{remark}

%------------------------------------
\subsection{Inference Without Stationarity}\label{snostat}
%------------------------------------

The theorems in the previous subsections justify the use of HAC estimators in GMM-type applications, such as \autoref{sdysmm} and \autoref{ssmm}, or more generally, settings where the moments are centered at the right conditional expectations. However, there are important settings where moments are uncentered, for example moment inequality models, which include the applications in \autoref{snonparam} and \autoref{ssubinf}. Another application in which moments may be uncentered is nonparametric inference, for example, on the conditional choice probability $\E[Y_1 \mid X_1]$, or other such moments. 

We next discuss two ways in which nonstationarity can be addressed. First, we note that the generalized spatial HAC proposed by \cite{leung2019normal} can be applied to our setting, an estimator that is consistent when both \autoref{station} and $\E[\psi_1]=c$ fail to hold. Second, we derive new general (albeit strong) sufficient conditions under which arbitrary moments (whether centered or not) obey spatial first-order stationarity {\em asymptotically}. Combined with \autoref{spatialhac}, the latter result enables the use of the spatial HAC in a larger set of environments, with the caveat that it is conservative rather than consistent (since $\E[\psi_1]$ will typically vary with $n$). We are unable to prove corresponding results for the network HAC, since it appears difficult to derive general sufficient conditions for an asymptotic version of stationarity (\autoref{networkstation}) outside of GMM-type models. 

\bigskip

\noindent {\bf Generalized Spatial HAC.} In settings where \autoref{station} fails to hold, we can use the generalized spatial HAC estimator proposed in \S 6.1 of \cite{leung2019normal}. This is given by
\begin{align}
  &\hat\sigma^2 - \hat\alpha\hat\alpha', \quad \text{where} \label{nonstathac} \\
  &\hat\sigma^2 = \frac{1}{n} \sum_{i=1}^n \psi_i \psi_i' + \frac{1}{n} \sum_{i=1}^n \sum_{j\neq i} \big(\xi_i \xi_j - \hat\theta(\rho_i) \hat\theta(\rho_j) \big) K\left( (\rho_i-\rho_j) / h_n \right), \nonumber\\
  &\hat\alpha = \bar{\psi} + \frac{1}{n} \sum_{i=1}^n \sum_{j\neq i} \big( \psi_i - \hat\theta(\rho_i) \big) K((\rho_i-\rho_j)/h_n), \quad\text{and} \nonumber\\
  &\hat\theta(p) = \frac{\sum_{i=1}^n \psi_i \tilde K(n^{-1/d}(p-\rho_j)/b_n)}{\sum_{i=1}^n \tilde K(n^{-1/d}(p-\rho_j)/b_n)}, \nonumber
\end{align}

\noindent where $\tilde K(\cdot)$ is a kernel function and $b_n$ a bandwidth for the kernel estimator $\hat\theta(\rho_1)$ of $\E[\psi_1 \mid \rho_1]$. Thus, we avoid having to impose stationarity because we estimate $\E[\psi_1 \mid \rho_1]$ directly, nonparametrically. Consistency of the generalized spatial HAC requires $h_n^d = O(n^{1/4})$, $h_n\rightarrow\infty$, and $\hat\theta(p)$ to be uniformly consistent and converging faster than the $n^{-1/4}$ rate, as usual. Under these requirements and high-level stabilization conditions, \cite{leung2019normal} prove consistency. Our CLTs verify these conditions and hence imply consistency for social interactions models.

\begin{remark}
  This estimator differs slightly from \cite{leung2019normal} due to the $n^{-1/d}$ scaling in $\hat\theta(\cdot)$ and the rate condition on $h_n$. The scaling addresses the problem that we want to estimate $\E[\psi_1 \mid \rho_1]$ by averaging over the agent statistics of agents with positions near $\rho_1$. However, positions are drifting apart as $n$ diverges, since $\rho_i = \omega_n \tilde\rho_i$. Fortunately, since $\omega_n$ drifts at rate $n^{1/d}$, we can simply scale down positions by this rate, converting $\rho_i$ to $\tilde\rho_i/\kappa^{1/d}$, the latter of which has bounded support and does not depend on $n$.  The bandwidth $h_n$ has a different rate because, in this paper, we model positions as $\omega_n\tilde\rho_i$, whereas they model positions as simply $\tilde\rho_i$, which may be more reasonable when agents are closely located in space. This modification has no effect on the consistency proof, since the required high-level conditions (stated in \autoref{smaster} and verified in our CLTs) are the same as those of \cite{leung2019normal} except we change positions from $\{\tilde\rho_i\}_{i=1}^n$ to $\{\omega_n\tilde\rho_i\}_{i=1}^n$.
\end{remark}

\noindent {\bf Asymptotic Stationarity.} We next provide general sufficient conditions applicable to {\em any} network moment under which it {\em asymptotically} satisfies first-order stationarity, i.e.\ $\abs{\E[\psi_1] - \E[\psi_1 \mid \rho_1=\omega_n p]} = o(n^{-1/3})$ for any $p$.\footnote{The proof of \autoref{spatialhac} uses exact stationarity (\autoref{station}), but an inspection of the proof shows that asymptotic first-order stationarity suffices.} The conditions are strong, but they are useful for understanding what widely used stationarity conditions essentially require when moments are not centered. Recall the definition of $Z_i$ from \autoref{ssma1} in the static setting and \autoref{lna} in the dynamic setting.

\begin{assump}\label{rhounif}
  (a) $Z_1 \indep \rho_1$. (b) $f$ is the uniform distribution on some bounded subset of $\R^d$ for which the origin is an interior point.
\end{assump}

\noindent Part (a) says that positions are independent of all other attributes, and (b) implies that agents are not more clustered in some regions of space than others. 

\begin{theorem}\label{genstat}
  Suppose the CLT assumptions hold (Theorems \ref{staticclt} and \ref{dynamicclt} in the static and dynamic cases, respectively). Under \autoref{rhounif}, there exists a sequence of constants $\{c_n\}_{n\in\mathbb{N}}$ such that $\abs{\E[\psi_1 \mid \rho_1=\omega_n p] - c_n} = o(n^{-1/3})$ for any $p \in \text{supp}(f)$.
\end{theorem}
\begin{proof}
  See \autoref{pogs}.
\end{proof}

\noindent Combined with \autoref{spatialhac}, this establishes that the spatial HAC is asymptotically conservative for any network moment under \autoref{rhounif}.

The remainder of this subsection discusses \autoref{rhounif}. Certainly it is a strong condition, but it clarifies what stationarity generally demands in spatial settings when moments are uncentered. The reason we need (b) is that, if some subset of $\text{supp}(f)$ had higher density, then agents positioned in that region would tend to have higher degrees because there are more nearby alters in that region, as well. Hence, degree would not be mean-independent of position, and we would expect other network statistics to similarly depend on an agent's position. We need (a) for the same reason; if an agent's position were correlated with some attribute of the agent, then agents with the right positions would have desirable realizations of the attribute and thus tend to have high degrees. In either case, the absolute value of an agent's position would influence her outcome, whereas, by definition, stationarity demands that only her relative position matters. %\footnote{This intuition only gives us {\em asymptotic} stationarity because for any fixed network size, agents near the boundary of $\text{supp}(f)$ have relatively fewer nearby alters with which to form connections.}

Thus, \autoref{rhounif} implies that {\em positions play no role in degree heterogeneity}. By this we mean the following. In general from specification \autoref{modelnf}, variation in $\mu_i$ and $\rho_i$ can increase or decrease $i$'s degree, but \autoref{rhounif} rules out variation due to $\rho_i$. Thus, positions only generate homophily, not heterogeneity in degree. Of course, $\mu_i$ can still do either. Consequently, under \autoref{rhounif}, the network can still be endogenous with respect to the outcome equation due to correlation between the error in the outcome equation and $\mu_i$.

%----------------------------------------------------------------------
\section{CLT for the Static Model}\label{ssm}
%----------------------------------------------------------------------

In this section, we state formal conditions for a CLT for network moments defined in \autoref{snetmoms}. The first two subsections present and motivate our main conditions for weak dependence and the third states regularity conditions. The weak dependence conditions are analogous to those used by \cite{leung2019normal} to prove a CLT for static models of network formation with strategic interactions. Their proof strategy is to verify high-level ``stabilization'' conditions for a normal approximation (their Theorem C.2) by constructing a dependency neighborhood for each observation, what they call the {\em relevant set}, and then using branching process techniques to bound the sizes of these neighborhoods. We use the same methodology for both the static and dynamic model, but the arguments need to be modified for the social interactions setting.

Additionally, rather than applying their Theorem C.2, we verify the high-level conditions of our \autoref{master-clt}, which is similar to Theorem C.2 but makes two modifications. First, we use an increasing domain setup (\autoref{as}), whereas \cite{leung2019normal} use a quasi-infill setup (see \autoref{quasiinfill}). Second, Theorem C.2 provides a closed-form expression for the limit variance, whereas our theorem does not, which allows us to dispense with a high-level continuity condition.

%------------------------------------
\subsection{Strategic Interactions}\label{ssma1}
%------------------------------------

The first key condition required for a CLT is a restriction on the strength of strategic interactions. For motivation, consider the standard linear-in-means model as studied in \cite{bramoulle2009identification}. The outcome equation is given by
\begin{equation*}
  Y_i = \alpha + \sum_j A_{ij} Y_j \beta + \sum_j A_{ij} X_j' \delta + X_i'\gamma + \varepsilon_i,
\end{equation*}

\noindent where $\bm{A}$ may be row-normalized. For this model to have a reduced form, 
\begin{equation}
  \abs{\beta} \lambda_\text{max}(\bm{A}) < 1,
  \label{<1}
\end{equation}

\noindent is required, where $\lambda_\text{max}(\bm{A})$ is the largest eigenvalue of the adjacency matrix $\tilde{\bm{A}}$. This is evidently a restriction on the magnitude of peer effects. It ensures equilibrium existence and weak dependence.

We impose an assumption analogous to \autoref{<1}. While it may appear superficially more complicated, this is unavoidable due to the nonlinearity of the model. First we require some definitions. Let
\begin{equation}
  \mathcal{R}^c_i = \bm{1}\left\{ \inf_s U(s, \tau_i) \leq 0 \medcap \sup_s U(s, \tau_i) > 0 \right\}. \label{Rci}
\end{equation}

\noindent The significance of this indicator is that its expectation equals
\begin{equation}
  \prob\big( \sup_s U(s, \tau_i) > 0 \big) - \prob\big( \inf_s U(s, \tau_i) > 0 \big), \label{ssi}
\end{equation}

\noindent (assuming measurability). This is our analog of $\beta$ in \autoref{<1}. It measures the strength of strategic interactions, since it corresponds to the average causal effect on choosing action 1 of changing $S_i$ from its ``lowest'' to its ``highest'' possible value.

The next assumption imposes a restriction on \eqref{ssi} in relation to the network topology. Recall that $i$'s position $\rho_i$ is a subvector of her type $\tau_i$. Let $Z_i$ be the subvector of $\tau_i$ that omits $\rho_i$ (so that $\tau_i = (\rho_i,Z_i)$), $d_z$ the dimension of $Z_i$, and $\Phi(\cdot \,|\, p)$ the conditional distribution of $Z_i$ given $\rho_i=p$. Define
\begin{equation}
  \varphi(p,z;p',z') = \E\left[ A_{ij} \mathcal{R}^c_j \mid \rho_i=p, Z_i=z, \rho_j=p', Z_j=z' \right]. \label{L_r}
\end{equation}

\noindent For any $h\colon \R^d \times \R^{d_z} \rightarrow \R$, define the mixed norm 
\begin{equation*}
  \norm{h}_{\bm{m}} = \sup_{p\in\R^d} \left( \int_{\R^{d_z}} h(p,z)^2 \,\text{d}\Phi^*(z) \right)^{1/2},
\end{equation*}

\noindent where $\Phi^*(\cdot)$ is a distribution given in the next assumption. Let $\bar{f} = \sup_{p\in\R^d} f(p)$.

\begin{assump}[Strength of Interactions]\label{dfrag}
  There exists a distribution $\Phi^*(\cdot)$ on $\R^{d_z}$ such that, for any $p,p' \in \R^d$ and $z \in \R^{d_z}$,
  \begin{equation*}
    \int_{\R^{d_z}} \varphi(p,z;p',z') \,\text{d}\Phi(z' \,|\, p') \leq \int_{\R^{d_z}} \varphi(p,z;p',z') \,\text{d}\Phi^*(z').
  \end{equation*}
  
  \noindent Furthermore,
  \begin{equation*}
    \norm{h}_{\bm{m}} < 1 \quad\text{for}\quad h(p,z) \equiv \kappa \bar{f}\, \int_{\R^d} \left( \int_{\R^{d_z}} \varphi(p,z;p',z')^2 \,\text{d}\Phi^*(z') \right)^{1/2} \,\text{d}p'.
  \end{equation*}
\end{assump}

The first part of the assumption is a regularity condition on the conditional distribution of $Z_i$. The substantive requirement is $\norm{h}_{\bm{m}} < 1$. To see how this assumption restricts \autoref{ssi}, note that, for any $n$,
\begin{equation}
  \E\left[ \sum_{j\neq i} A_{ij} \mathcal{R}_j^c \,\bigg|\, Z_i=z, \rho_i=\omega_n p \right] \leq \kappa \bar{f} \int_{\R^d} \int_{\R^{d_z}} \varphi(p,z;p',z') \,\text{d}\Phi^*(z') \text{d}p', \label{r90bf0wj2}
\end{equation}

\noindent using a change of variables and the fact that $n\omega_n^{-d} = \kappa$. Clearly $\norm{h}_{\bm{m}}$ is an upper bound on the right-hand side, so a necessary condition for \autoref{dfrag} is that the left-hand side is strictly less than one. The interpretation of the left-hand side becomes clearer in the special case where the network is exogenous in the sense that $A_{ij} \indep \mathcal{R}_j^c$. Then $\norm{h}_{\bm{m}} < 1$ implies
\begin{equation}
  \left[ \prob\big( \sup_s U(s, \tau_j) > 0 \big) - \prob\big( \inf_s U(s, \tau_j) > 0 \big) \right] \E[A_{ij}] < 1 \label{nec}
\end{equation}

\noindent for any $n$, since the left-hand side of \autoref{r90bf0wj2} equals the expectation of the left-hand side of \autoref{nec}. This is transparently a restriction on the strength of strategic interactions, requiring it to be bounded below the inverse of the expected degree. The intuition is that, when mean degree is higher, an agent's action can impact more neighbors, so weak dependence requires weaker strategic interactions. 

To see more clearly how \autoref{nec} compares with \autoref{<1}, note that $\lambda_\text{max}(\bm{A}) \leq \max_i \sum_j A_{ij}$. Hence, a sufficient condition is
\begin{equation*}
  \abs{\beta} \max_i \sum_j A_{ij} < 1,
\end{equation*}

\noindent which has precisely the form of \autoref{nec}. Whether $\bm{A}$ is row-normalized, our condition and \eqref{<1} have similar behavioral implications, as discussed in Remark A.1 of \cite{leung2019compute}. Finally, it should be noted that both \autoref{<1} and \autoref{dfrag} are restrictions on the effect of $S_i$ on $Y_i$ only in a partial-equilibrium sense. The general equilibrium effect can be substantially larger, as seen in the simulation study of \cite{leung2016} in the context of network formation. 

Conditions similar to \autoref{dfrag} have been utilized by \cite{xl2015} and \cite{xu2018social} for social interactions models. \cite{leung2016} and \cite{menzel2015large} use analogous conditions for weak dependence in network formation and matching models, respectively. The former paper discusses similarities between this condition and related assumptions for temporal and spatial autoregressive models.

%------------------------------------
\subsection{Equilibrium Selection}\label{ssma2}
%------------------------------------

Whereas in the linear-in-means model, \eqref{<1} is enough to guarantee equilibrium uniqueness, in our setting multiple equilibria are possible under \autoref{dfrag}. By \autoref{select}, the selection mechanism $\lambda(\cdot)$ can be any function of $\bm{A}$ and $\bm{\tau}$, which can potentially induce strong dependence between the outcomes of agents distant in the network. For this reason, we require a restriction on $\lambda(\cdot)$ for weak dependence. We leave the formal statement of the assumption to \autoref{scltassumps}. Our goals here are (1) to motivate why a restriction is needed and (2) note that the assumption is satisfied under variants of myopic best-response dynamics, such as those used to define the dynamic model.

To illustrate why we need to restrict $\lambda(\cdot)$, let us take an extreme example. Suppose types are realized such that $\bm{A}$ has two components (disconnected subnetworks). By definition, agents in separate components are infinitely far apart in terms of path distance. Under \autoref{S}, the set of Nash equilibria for the full game is the Cartesian product of the set Nash equilibria on each component alone, since payoffs in one component do not depend on actions in the other component. In this respect, they appear to be two separate games. Nonetheless, it is easy to construct examples of $\lambda(\cdot)$ such that the realization of equilibria is correlated between components. For example, suppose types are realized such that there are two possible equilibria in each component. Consider an equilibrium selection mechanism that chooses one equilibrium over the other in each component, depending on the realization of agent 1's type $\tau_1$. Then the same underlying random vector determines equilibrium selection in both components, so their outcomes are correlated. In the cross-sectional setting, this problem is analogous to having a selection mechanism simultaneously determine equilibrium realizations across a set of distinct games, a scenario implicitly ruled out when researchers assume separate realizations of the same game are independent \citep{epstein2016robust}.

In the previous example, all agents effectively coordinate on equilibria according to the realization of a common signal $\tau_1$. For weak dependence, we must rule out coordination between ``distant'' agents. In order to formalize what we mean by ``distant'' and state the required condition (\autoref{nocoord}), we need several new definitions that have to be motivated at length, which we defer to \autoref{scltassumps}. In brief, the basic idea of the condition is to only allow agents within certain subnetwork neighborhoods to coordinate on their equilibrium actions. Such a restriction holds under variants of {\em myopic best-response dynamics}. This is an interesting class of selection mechanisms to consider because they are widely used in theoretical work to model diffusion, as discussed in \autoref{sdynamic}, and they define the other of our main applications, the dynamic model.

To define these dynamics more formally, suppose an initial vector of outcomes $\bm{Y}^0$ is chosen randomly, and then agents iteratively best-respond according to \autoref{mrbd} until convergence. Under a game of strategic complements, the process necessarily converges to a Nash equilibrium \citep{milgrom1990rationalizability}. Otherwise, convergence requires $\bm{Y}^0$ to lie along an improving path \citep{jackson_evolution_2002}. In either case, this constitutes a selection mechanism because it maps $(\bm{A},\bm{\tau})$ to a unique Nash equilibrium. %More generally, we can allow for variations of the dynamics where only a random subset of agents adjusts their actions in best-response at each step. Mechanisms of this type clearly rule out coordination between components as illustrated in the example above. In particular, they satisfy \autoref{nocoord}. 

%------------------------------------
\subsection{Regularity Conditions}\label{sreg}
%------------------------------------

We next formalize the sequence of models along which we take limits. Similar to \autoref{lna}, the static model is given by the tuple
\begin{equation}
  (U, \lambda, V, \bm{\rho}, \bm{Z}, \bm{\zeta}), \label{staticmodseq}
\end{equation}

\noindent where $U(\cdot)$ is the payoff function in \autoref{model}, $\lambda(\cdot)$ the equilibrium selection mechanism given in \autoref{select}, $V(\cdot)$ the latent index in the network formation model \autoref{modelnf}, $\bm{\rho} = (\rho_i)_{i=1}^n$, $\bm{Z} = (Z_i)_{i=1}^n$ with $Z_i$ defined prior to \autoref{dfrag}, and $\bm{\zeta}$ is the matrix of $\zeta_{ij}$'s. We can then directly import \autoref{as} using this new notation.

Several regularity conditions are used to establish the CLT. The first imposes a restriction on the network formation model, requiring the linking probability to decay exponentially quickly in spatial distance. Exponential decay is useful for a CLT in the same way that mixing coefficients are usually required to tend to zero with distance at an exponential rate.

\begin{assump}[Network Formation]\label{sparsity} \hfill
  \begin{enumerate}[(a)]
    \item There exists $(\bar{\mu},\bar{\mu}')$ such that $V(\delta, \mu, \mu', \zeta) \leq V(\delta, \bar{\mu}, \bar{\mu}', \zeta)$ for any $\delta \in \R_+$, and $(\mu,\mu',\zeta) \in \text{supp}( (\mu_i,\mu_j,\zeta_{ij}) )$. 

    \item $\text{supp}(\zeta_{ij}) \subseteq \R$, and for any $\delta,\mu,\mu'$, $V(\delta,\mu,\mu',\cdot)$ is strictly monotonic. 

    \item Let $\tilde V^{-1}(r,\cdot)$ be the inverse of $V(r, \bar{\mu}, \bar{\mu}', \cdot)$. There exist $c_1,c_2>0$ such that 
      \begin{equation*}
	\tilde{\Phi}_\zeta\left( \tilde V^{-1}(r, 0) \right) \leq c_1 e^{-c_2 r}, 
      \end{equation*}

      \noindent where $\tilde{\Phi}_\zeta(\cdot)$ denotes the complementary CDF of $\zeta_{ij}$.
  \end{enumerate}
\end{assump}

\noindent Parts (a) and (b) are common regularity conditions; (a) says the $\mu_i$'s have a bounded impact on the latent index $V(\cdot)$. Part (c) is the main restriction, which limits the nonlinearity of $V(\cdot)$ and the tail mass of the distribution of $\zeta_{ij}$. It holds, for example, if $V(\cdot)$ is linear in its arguments and $\zeta_{ij}$ has exponential tails. Note that $\prob(A_{ij}=1 \mid \delta_{ij}=r) \leq \tilde{\Phi}_\zeta( \tilde V^{-1}(r, 0))$. Hence, part (c) implies that connection probabilities are exponentially decaying with distance. This is a typical requirement in latent space models \citep{breza_using_2017,hrh2002}.

\begin{remark}[Sparsity]\label{rsparse}
  Under the increasing domain asymptotics of \autoref{as}, the previous assumption implies that the observed network is {\em sparse} in the sense that the expected number of links involving any arbitrary agent is $O(1)$.\footnote{The expected number of links is $\E[\sum_j A_{ij}] \leq n \E[A_{ij}] \leq n\E[\tilde{\Phi}_\zeta( \tilde V^{-1}(\omega_n\norm{\tilde\rho_i-\tilde\rho_j}, 0)] \leq \kappa\bar{f} C \int_{r\geq 0} c_1 e^{-c_2 r} \,\text{d}r = O(1)$ for some constant $C$ that appears from a change of variables to hyperspherical coordinates..} This is a commonly used definition of network sparsity that captures the stylized fact that, in most real world social networks, the number of connections involving the typical agent is small relative to the network size \citep{chandrasekhar2016}.
\end{remark}

For the next assumption, define
\begin{equation}
  \mathcal{T} = \{(p,z)\colon p \in \R^d, z \in \text{supp}(\Phi(\cdot \mid p))\}.
  \label{Sset}
\end{equation}

\begin{assump}[Non-Degeneracy]\label{dreg}
  Either $\varphi(p,z;p',z')=0$ for any $(p,z,p',z')$, or for $\Phi^*(\cdot)$ defined in \autoref{dfrag},
  \begin{equation*}
    \inf_{(p,z) \in \mathcal{T}} \int_{\R^d} \int_{\R^{d_z}} \varphi(p,z;p',z') \,\text{d}\Phi^*(z') \,\text{d}p' > 0. 
  \end{equation*}
\end{assump}

\noindent If $\varphi(p,z;p',z')=0$ everywhere, then this simply means there are no strategic interactions, which is the uninteresting case. To interpret the other case, note that when the network is exogenous, so that $\mathcal{R}_i^c \indep \bm{A}$, then by \autoref{r90bf0wj2},
\begin{multline*}
  \kappa \bar{f} \int_{\R^d} \int_{\R^{d_z}} \varphi(p,z;p',z') \,\text{d}\Phi^*(z') \\ \geq \E[\mathcal{R}_i^c \mid Z_i=z, \rho_i=\omega_n p] \,\E\left[\sum_j A_{ij} \,\bigg|\, Z_i=z, \rho_i=\omega_n p \right].
\end{multline*}

\noindent Hence, a sufficient condition is for the infimum of the right-hand side to be strictly positive. Note that if interactions exist, then necessarily $\E[\mathcal{R}_i^c] > 0$, so the sufficient condition strengthens this slightly by requiring it to hold conditional on type, for all types. Additionally, it requires a non-degenerate network in the sense that the limiting expected degree of any agent, conditional on their type, is positive.

For the last regularity condition, let $\bm{\rho}^m = (\rho_i)_{i=1}^{m+2}$, $\{H_n\}_{n\in\mathbb{N}}$ be a sequence of subsets of $\R^d$, and $\psi_i^m(H_n)$ be $i$'s agent statistic under a modified version of model \eqref{staticmodseq} where instead of $\bm{\rho}$ as our set of positions, we instead have $(\rho_i)_{i=1}^{m+2} \cap H_n$.\footnote{To be clear, the modified model is the same as the original model, except we draw types as follows. First, we generate the set of positions $(\rho_i)_{i=1}^{m+2} \cap H_n$. Then independently across each element $\rho_j$ in this set, we draw $Z_j$ from $\Phi(\cdot \,|\, \rho_j)$. \autoref{psibd} in the appendix is the same as \autoref{psibdsimple} but with the formal data-generating process completely spelled out.} The original model is a special case where $m = n-2$ and $H_n = \R^d$ for all $n$.

\begin{assump}[Bounded Moments]\label{psibdsimple} 
  (a) There exists a finite constant $C$ such that $\E[\abs{\psi_1^m(H_n)}^8 \mid X_1=x, X_2=y] < C$ for any $n$ sufficiently large, $m \in [n/2,3n/2]$, $x,y \in \R^d$, and sequence of subsets $\{H_n\}_{n\in\mathbb{N}}$ of $\R^d$. (b) There exists $c\geq 0$ such that $\abs{\psi_1^m(H_n)} \leq c\, m^c$ a.s.\ for any $m,n\in\mathbb{N}$.
  % originally i had y \in \R^d \cup \{\infty\} which was imprecise but meant that agent 2 didn't matter for agent 1 since she's outside 1's radius of stabilization. But then that's just the same as replacing m with m-1
\end{assump}

\noindent Part (b) says that agent statistics are uniformly bounded by a polynomial in the network size. This is easily satisfied by the moments in \autoref{ssmm} and \autoref{ssubinf} if instrument functions are uniformly bounded. Part (a) essentially requires uniformly bounded 8th moments for agent statistics but is slightly stronger for technical reasons. Many of the moments in the applications are uniformly bounded, including \autoref{sdysmm} and \ref{snonparam}, in which case part (a) is automatic. For subnetwork moments (used in \autoref{ssubinf}), \autoref{psibd} holds under \autoref{sparsity} by Lemma I.17 of \cite{leung2019normal}.

%------------------------------------
\subsection{Main Result}\label{ssmclt}
%------------------------------------

Before stating the main result, we need some additional notation. Let $\psi_i^{n+1}$ be $i$'s agent statistic in a model with $n+1$ agents. Consider counterfactually removing agent $n+1$ from the model, and let $\bm{A}^-$ and $\bm{\tau}^-$ be the network and type vector with agent $n+1$ removed. Let $\bm{Y}^- = \lambda(\bm{A}^-, \bm{\tau}^-)$ be the equilibrium outcome vector in the resulting model with the remaining $n$ agents. Let $\psi_i^{-(n+1)}$ be $i$'s agent statistic constructed from $(\bm{Y}^-, \bm{A}^-, \bm{\tau}^-)$. Define the {\em add-one cost} $\Xi_n = \psi_{n+1}^{n+1} + \sum_{i=1}^n (\psi_i^{n+1} - \psi_i^{-(n+1)})$, which is the total change in the sum of all agent statistics from adding a single agent $n+1$. Let $\bm{\Sigma}_n = \var(n^{-1/2} \sum_{i=1}^n \psi_i)$ and $\bm{I}_m$ be the $m\times m$ identity matrix. Recall that $\lambda_\text{min}(\bm{\Sigma}_n)$ denotes the smallest eigenvalue of $\bm{\Sigma}_n$. Finally, recall that agent statistics are vectors of dimension $m$.

\begin{theorem}\label{staticclt}
  Suppose $c'\Xi_n$ is asymptotically non-degenerate for any $c \in \R^m\backslash\{\zero\}$. Also suppose the model satisfies Assumptions \ref{S} and \ref{select}, weak dependence conditions (Assumptions \ref{dfrag} and \ref{nocoord}), and regularity conditions (Assumptions \ref{sparsity}--\ref{psibdsimple}). If $\psi(\cdot)$ satisfies \autoref{klocal}, then as $n\rightarrow\infty$,
  \begin{equation*}
    \frac{1}{\sqrt{n}} \sum_{i=1}^n \bm{\Sigma}_n^{-1/2} \left( \psi_i - \E[\psi_i] \right) \dlimarrow \mathcal{N}(\zero, \bm{I}_m) \quad\text{and}\quad \liminf_{n\rightarrow\infty} \lambda_\text{min}(\bm{\Sigma}_n) > 0
  \end{equation*}

  \noindent under the sequence given in \autoref{as}.
\end{theorem}
\begin{proof}
   See \autoref{staticproof} for the formal proof and a proof sketch.
\end{proof}

\noindent \autoref{nocoord} is discussed in \autoref{ssma2} and \autoref{scltassumps}. Non-degeneracy of the add-one cost essentially just requires a non-trivial choice of agent statistics. When agent $n+1$ is added to the model, the ``direct effect'' on the total statistic is that it increases by the value of her agent statistic $\psi_{n+1}^{n+1}$, which should be non-degenerate. However this also has an ``indirect effect'' on the statistics of other agents $i$ given by $\psi_i^{n+1} - \psi_i^{-(n+1)}$, for example due to strategic interactions. Non-degeneracy would fail, for instance, if the direct and indirect effects exactly cancel, which is unlikely to happen for non-trivial statistics. 

%----------------------------------------------------------------------
\section{CLT for the Dynamic Model}\label{sdm}
%----------------------------------------------------------------------

We first formalize the model for the initial conditions $\bm{Y}^0 = (Y_i^0)_{i=1}^n$. This is simply given by the static model in \autoref{smodel}. That is, agents choose their period-0 actions best-responding to the contemporaneous actions of her peers in that period. We view this as an approximation of the long-run outcome of the dynamic model (cf.\ \autoref{ssma2}). Define $\bm{\tau}^0 = (\tau_i^0)_{i=1}^n$. Let $U_0(\cdot)$ be the payoff function in period 0, which may potentially differ from $U(\cdot)$, and let $S_0(\cdot)$ be the analog of $S(\cdot)$ (statistics that capture strategic interactions). Let $\mathcal{E}(\bm{A},\bm{\tau}^0)$ be the set of outcomes $\bm{Y}^0$ such that each component $Y_i^0$ satisfies
\begin{equation}
  Y_i^0 = \bm{1}\left\{ U_0(S_i^0, \tau_i^0) > 0 \right\} \quad\text{for}\quad S_i^0 \equiv S_0(Y_{-i}^0, \tau_i^0, \tau_{-i}^0, A_i, A_{-i}).
\end{equation}

\begin{assump}[Initial Condition]\label{dyinit}
  (a) $\abs{\mathcal{E}(\bm{A},\bm{\tau}^0)} \geq 1$. (b) There exists a function $\lambda(\bm{A},\bm{\tau}^0)$ with range $\mathcal{E}(\bm{A},\bm{\tau}^0)$ such that the realized initial condition satisfies $\bm{Y}^0 = \lambda(\bm{A},\bm{\tau}^0)$.
\end{assump}

\noindent This is just a restatement of \autoref{select} with period-0 types. 

Define the {\em add-one cost} analogously to \autoref{ssmclt}: $\Xi_n = \psi_{n+1}^{n+1} + \sum_{i=1}^n (\psi_i^{n+1} - \psi_i^{-(n+1)})$, which is the total change in the sum of all agent statistics from removing a single agent $n+1$. Also recall that $\bm{\Sigma}_n = \var(n^{-1/2} \sum_{i=1}^n \psi_i)$ and $\bm{I}_m$ is the $m\times m$ identity matrix.

\begin{theorem}\label{dynamicclt}
  Suppose $c'\Xi_n$ is asymptotically non-degenerate for any $c \in \R^m\backslash\{\zero\}$. Also suppose 
  \begin{itemize}
    \item the best-response model satisfies \autoref{S2},
    \item the initial conditions model defined in \autoref{dyinit} satisfies Assumptions \ref{S}, \ref{dfrag}, \ref{dreg}, and \ref{nocoord}, using $U_0(\cdot), S_0(\cdot)$ in place of $U(\cdot), S(\cdot)$, and
    \item regularity conditions (Assumptions \ref{sparsity} and \ref{psibdsimple}) hold.
  \end{itemize}
  
  \noindent If $\psi(\cdot)$ satisfies \autoref{klocal}, then as $n\rightarrow\infty$,
  \begin{equation*}
    \frac{1}{\sqrt{n}} \sum_{i=1}^n \bm{\Sigma}_n^{-1/2} \left( \psi_i - \E[\psi_i] \right) \dlimarrow \mathcal{N}(\zero, \bm{I}_m) \quad\text{and}\quad \liminf_{n\rightarrow\infty} \lambda_\text{min}(\bm{\Sigma}_n) > 0
  \end{equation*}

  \noindent under the sequence given in \autoref{as}.
\end{theorem}
\begin{proof}
  See \autoref{dynamicproof}.
\end{proof}

\noindent The assumptions imposed on the initial conditions model are the same as those imposed on the static model for a CLT. Primitive sufficient conditions for \autoref{psibdsimple} are discussed following its statement in \autoref{sreg}. As in the static case, non-degeneracy of the add-one cost essentially just requires that the choice of $\psi(\cdot)$ is nontrivial; see the discussion following the statement of \autoref{staticclt}.

%----------------------------------------------------------------------
\section{Numerical Illustrations}\label{seill}
%----------------------------------------------------------------------

This section illustrates the performance of the HAC estimators in an empirical application and simulation study.

%------------------------------------
\subsection{Empirical Application}
%------------------------------------

We revisit the study of \cite{conley2010learning} on the diffusion of an agricultural technology through a network of farmers. They show that as farmers learn to use a new technology (pineapple production), they respond to information from their social contacts. In particular, they find that when the returns from a network neighbor's crop yields are above or below expectation, they adjust their agricultural inputs accordingly. Their main specifications are logistic models of the probability that a farmer's agricultural inputs change in response to the shares of network neighbors with above- or below-expectation yields in the previous period.

We replicate Table 4, column B of their paper. The outcome is an indicator for whether farmer $i$ adjusts her inputs in period $t$. There are four main regressors of interest comprising $S_i^t$, which are the share of good / bad news events from network neighbors who choose the same / different inputs as farmer $i$ last period. We include the same controls in their specification, including a full set of village and planting round dummies. 

The authors use a spatial HAC estimator out of concern for potential spatial autocorrelation in the errors. Our theoretical results show that the estimator also accounts for forms of network autocorrelation. If spatial positions are unavailable, our theory shows that the network HAC can be used instead.

\autoref{appresults} reports point estimates and spatial and network HAC standard errors for a range of bandwidths using the Bartlett kernel. The four main regressors of interest correspond to the first four rows in the table. The remaining rows are control variables. \cite{conley2010learning} use a bandwidth of 1500 and say that the results are similar for 1000 and 2000. This is replicated in our table. The path distances used in the network HAC are obtained from the data on farmers' advice networks used to construct the main regressors. Across the three villages, average path distance ranges from about 0.5 to 2 and the maximum distance from 3 to 8, so we report bandwidth values in the range of 0 to 3. 

From the table, we see that the results are not generally sensitive to bandwidth choice for positive values of the bandwidth. The table shows that spatial and network standard errors are similar, which is consistent with our theory. For comparison, we report a bandwidth of zero, which corresponds to i.i.d.\ standard errors. Depending on the coefficient, these standard errors can be conservative or anti-conservative relative to the HAC standard errors.

\begin{table}[ht]
\centering
\caption{Application Results}
\begin{threeparttable}
\begin{tabular}{lrrrrrrrr}
\toprule
{} &                              {} & \multicolumn{3}{c}{Spatial HAC} & \multicolumn{4}{c}{Network HAC} \\
\cmidrule{3-9}
{} & \multicolumn{1}{c}{$\hat\beta$} & \multicolumn{1}{c}{1000} & \multicolumn{1}{c}{1500} & \multicolumn{1}{c}{2000} & \multicolumn{1}{c}{0} & \multicolumn{1}{c}{1} & \multicolumn{1}{c}{2} & \multicolumn{1}{c}{3} \\
\midrule
success     &                          -0.084 &                    1.033 &                    0.947 &                    0.841 &                 0.980 &                 1.000 &                 0.921 &                 0.822 \\
awaysuccess &                           1.637 &                    0.797 &                    0.781 &                    0.736 &                 0.968 &                 0.789 &                 0.782 &                 0.738 \\
failure     &                           4.322 &                    1.910 &                    1.934 &                    1.922 &                 1.758 &                 1.963 &                 1.995 &                 2.021 \\
awayfailure &                          -5.898 &                    2.695 &                    2.566 &                    2.476 &                 2.374 &                 2.583 &                 2.518 &                 2.361 \\
geoabsx     &                           0.146 &                    0.080 &                    0.070 &                    0.065 &                 0.091 &                 0.092 &                 0.078 &                 0.071 \\
new         &                           1.973 &                    0.938 &                    0.885 &                    0.822 &                 0.869 &                 0.978 &                 0.811 &                 0.778 \\
exten       &                          -1.345 &                    0.667 &                    0.672 &                    0.693 &                 0.695 &                 0.837 &                 0.882 &                 0.865 \\
wealth      &                           0.182 &                    0.131 &                    0.133 &                    0.137 &                 0.137 &                 0.136 &                 0.142 &                 0.142 \\
abu1        &                           1.592 &                    1.081 &                    1.098 &                    1.151 &                 0.933 &                 1.207 &                 1.242 &                 1.275 \\
abu2        &                           2.145 &                    1.298 &                    1.228 &                    1.257 &                 1.030 &                 1.406 &                 1.258 &                 1.211 \\
rpent       &                          -0.286 &                    0.836 &                    0.734 &                    0.687 &                 0.808 &                 1.042 &                 0.763 &                 0.565 \\
\bottomrule
\end{tabular}
\begin{tablenotes}[para,flushleft]
  \footnotesize $n=107$. Parameters are listed in the same order as Table 4 of \cite{conley2010learning}. The second column displays point estimates of the logit regression. The remaining columns to the right display standard errors. The ``Spatial'' (``Network'') columns report standard errors using the spatial (network) HAC estimator with the indicated bandwidths. Village and planting round dummies are included but not reported.
\end{tablenotes}
\end{threeparttable}
\label{appresults}
\end{table}

%------------------------------------
\subsection{Monte Carlo}\label{smc}
%------------------------------------

We simulate a dynamic probit model with two periods and period-1 payoffs given by
\begin{equation*}
  U(S_i^1, \tau_i^1) = \beta_1 + \beta_2 X_i^1 + \beta_3 \frac{\sum_j A_{ij} Y_j^0}{\sum_j A_{ij}} + \nu_i^1,
\end{equation*}

\noindent where $\nu_i^t \sim \mathcal{N}(0,1)$. To estimate the model, we use a standard probit regression. As discussed in \autoref{sdysmm}, to account for autocorrelation in the errors and regressors, we estimate the variance using the sandwich formula $H^{-1}SH^{-1}$, where $H$ is the Hessian and $S$ is a HAC estimate of the variance of the scores. 

We are faced with the task of generating network autocorrelation in the errors while still maintaining the assumption of standard normal marginals, so that the probit can be used for estimation. Toward this end, we use the following modification of \autoref{e1}:
\begin{equation*}
  \nu_i^t = \omega_i \left( \frac{\sum_j A_{ij} \varepsilon_j^t}{\sum_j A_{ij}} + \varepsilon_i^t \right),
\end{equation*}

\noindent where $\varepsilon_i^t \stackrel{iid}\sim \mathcal{N}(0,1)$ and $\omega_i$ is a constant. The first term is the average error of network neighbors, which captures exogenous peer effects in unobservables and generates autocorrelation. We choose the weight to ensure standard normal marginals for $\nu_i^t$: $\omega_i = (1 + (\sum_j A_{ij})^{-1})^{-1/2}$.

Let covariates follow an AR(1) model $X_i^1 = 0.5 X_i^0 + u_i^t$, where $u_i^t \stackrel{iid}\sim \mathcal{N}(0,1)$ and $X_i^0 \stackrel{iid}\sim \text{exp}(1)$. We set the true $\beta$ to $(0.5, -0.3, 1)$. Initial conditions $Y_i^0$ are given by the same outcome model as period 1, except we set $\beta_3$ to zero and replace period-1 types $(X_i^1,\varepsilon_i^1)$ with period-0 types $(X_i^0,\varepsilon_i^0)$.  To generate the network, we use the random geometric graph model 
\begin{equation*}
  A_{ij} = \bm{1}\{\norm{\rho_i - \rho_j} \leq 1\} \mathbf{1}\{i\neq j\}.
\end{equation*}

\noindent Recall from \autoref{lna} that positions are given by $\rho_i = \omega_n \tilde\rho_i$. We draw $\tilde\rho_i \stackrel{iid}\sim \mathcal{U}([0,1]^2)$ and set the scaling factor $\omega_n$ to $(n / (5\pi) )^{1/2}$. This choice of $\omega_n$ ensures that the limiting expected degree is five, which generates a giant component \citep{penrose2003}. Additionally, this model clearly induces geographic homophily, and as a consequence $\bm{A}$ typically has a high clustering coefficient. All of these are well-known to be realistic network features \citep{barabasi2015}.

We compute standard errors four different ways. The first three use HAC estimators, all with the Bartlett kernel. For the standard spatial HAC, we set $h_n = n^{1/(3d)}$ (here $d=2$), and for the network HAC, $h_n = \log n$. The third uses the generalized spatial HAC from \autoref{snostat}, for which we set $h_n = n^{1/(4d)}$. For $\hat\theta(x)$ in the generalized spatial HAC, we use a kernel estimator and choose $b_n = (\log n/n)^{1/(2p+d)}$ for smoothness $p=d/2+1$, which achieves a $n^{-1/4}$ uniform convergence rate in the i.i.d.\ case. The fourth way of computing standard errors, which we term ``oracle'' standard errors, is by simulating the true variance of the probit estimator using 7500 simulation draws. This is used to assess the quality of the HAC estimators and our CLT.

We report results for $n = 500, 1\text{k}, 2\text{k}$. The point estimates are all extremely close to the truth on average for all sample sizes. \autoref{tabestse} displays standard errors, averaged across 15,000 simulations. Column ``GS HAC'' corresponds to the generalized spatial HAC. The table shows that all HAC standard errors are quite close on average to the target oracle standard errors. They can be slightly anti-conservative in smaller samples, which is a common issue with HAC estimators.

\autoref{tabrr} reports rejection percentages from two-sided $t$-tests at the 5 percent level of the null hypothesis that $\beta$ equals its true value, using the previous standard errors. The oracle standard errors control size well across all sample sizes, which illustrates the accuracy of the normal approximation. For the spatial and network HAC, rejection rates are fairly close to the nominal level, although there is a tendency toward overrejection. The network HAC happens to perform the best, while the generalized spatial HAC has the greatest tendency to overreject. The spatial HAC has better performance than the latter, at least when first-order stationarity holds. However, the story will be different when stationarity fails, as we will see below. While the network and generalized spatial HAC are not guaranteed to be positive definite in finite samples, they are positive definite in all of our simulation draws. Finally, the third column, labeled ``Naive,'' reports $t$-tests computed using i.i.d.\ standard errors, and we see that rejection rates can be many times higher than the nominal level. 

\begin{table}[ht]
\centering
\small
\caption{Probit: Standard Errors}
\begin{threeparttable}
\begin{tabular}{lrrrrrrrrrrrrrrr}
\toprule
{} & \multicolumn{3}{c}{Spatial HAC} & \multicolumn{3}{c}{Network HAC} & \multicolumn{3}{c}{GS HAC} & \multicolumn{3}{c}{Oracle} \\
\cmidrule{2-13}
{} &         500 &   1k &   2k &         500 &   1k &   2k &    500 &   1k &   2k &    500 &   1k &   2k \\
\midrule
$\beta_1$ &        0.16 & 0.12 & 0.09 &        0.17 & 0.12 & 0.09 &   0.16 & 0.12 & 0.08 &   0.18 & 0.13 & 0.09 \\
$\beta_2$ &        0.06 & 0.04 & 0.03 &        0.06 & 0.04 & 0.03 &   0.06 & 0.04 & 0.03 &   0.06 & 0.05 & 0.03 \\
$\beta_3$ &        0.26 & 0.19 & 0.14 &        0.27 & 0.19 & 0.14 &   0.26 & 0.18 & 0.13 &   0.30 & 0.20 & 0.15 \\
\bottomrule
\end{tabular}
\begin{tablenotes}[para,flushleft]
  \footnotesize Cells are averages over 15k simulations. ``GS HAC'' corresponds to the generalized spatial HAC and ``Oracle'' to true standard errors (obtained by simulation).
\end{tablenotes}
\end{threeparttable}
\label{tabestse}
\end{table}

\begin{table}[ht]
\centering
\footnotesize
\caption{Probit: Rejection Percentages for Two-Sided $t$-Test}
\begin{threeparttable}
\begin{tabular}{lrrrrrrrrrrrrrrr}
\toprule
{} & \multicolumn{3}{c}{Spatial HAC} & \multicolumn{3}{c}{Network HAC} & \multicolumn{3}{c}{GS HAC} & \multicolumn{3}{c}{Oracle} & \multicolumn{3}{c}{Naive} \\
\cmidrule{2-16}
{} &         500 &  1k &  2k &         500 &  1k &  2k &    500 &  1k &  2k &    500 &  1k &  2k &   500 &   1k &   2k \\
\midrule
$\beta_1$ &         8.0 & 7.1 & 6.9 &         7.5 & 6.5 & 6.3 &    8.4 & 7.6 & 7.3 &    5.0 & 5.2 & 4.9 &  14.6 & 14.1 & 14.7 \\
$\beta_2$ &         6.0 & 5.5 & 5.2 &         6.3 & 5.6 & 5.3 &    6.0 & 5.5 & 5.1 &    5.0 & 4.7 & 4.9 &   5.3 &  5.2 &  5.2 \\
$\beta_3$ &         7.8 & 7.1 & 6.8 &         7.3 & 6.6 & 6.3 &    8.2 & 7.5 & 7.2 &    4.8 & 5.3 & 4.9 &  13.6 & 13.9 & 13.7 \\
\bottomrule
\end{tabular}
\begin{tablenotes}[para,flushleft]
  \footnotesize Rejection percentages for two-sided $t$-tests at the 5\% level, obtained from 15k simulations. ``GS HAC'' corresponds to the generalized spatial HAC, ``Oracle'' to true standard errors (obtained by simulation), and ``Naive'' to i.i.d.\ standard errors.
\end{tablenotes}
\end{threeparttable}
\label{tabrr}
\end{table}

The moments used in pseudo-maximum likelihood are centered, since they are conditionally mean zero. Hence, by Theorems \ref{spatialhac} and \ref{networkhac}, the HAC estimators are asymptotically exact, as confirmed by the probit results. We next report results for the weighted average period-1 outcome $n^{-1} \sum_i Y_i^1 X_i^1$, which is generated from the same data-generating process as the previous tables. This is clearly not centered and not first-order stationary. However, an asymptotic analog of \autoref{station} does hold by \autoref{genstat}. Thus, our results predict that the spatial HAC is conservative while the generalized spatial HAC is exact. We have no predictions for the network HAC.

\autoref{YrrSE} displays standard errors and rejection percentages for the level-5\% two-sided $t$-test of the null that $\E[Y_i^1]$ equals its true value. The true value is computed from 7500 simulation draws, and the values displayed in the table are averages over 15,000 separate simulation draws. ``Oracle'' standard errors are obtained by simulating the variance of $n^{-1} \sum_i Y_i^1X_i^1$ from 7500 separate simulation draws.

From the table, we see that both the spatial and network HAC SEs are quite conservative, being about twice as large as the oracle SEs. In contrast, the generalized spatial HAC SEs produce rejection percentages close to the nominal level. The network HAC is also asymptotically conservative and comparable to the spatial HAC. This is interesting because network first-order stationarity (\autoref{networkstation}) does not apparently hold, so this is a case not covered by our theory.

\begin{table}[ht]
\centering
\small
\caption{Weighted Outcome: Rejection Percentages and Standard Errors}
\begin{threeparttable}
\begin{tabular}{lrrrrrrrrrrrr}
\toprule
{} & \multicolumn{3}{c}{Spatial HAC} & \multicolumn{3}{c}{Network HAC} & \multicolumn{3}{c}{GS HAC} & \multicolumn{3}{c}{Oracle} \\
\cmidrule{2-13}
Reject \% &        1.34 & 0.57 & 0.27 &        1.08 & 0.37 & 0.17 &   6.31 & 5.81 & 5.35 &   4.81 & 4.95 & 4.49 \\
SE        &        0.07 & 0.05 & 0.04 &        0.07 & 0.05 & 0.04 &   0.04 & 0.03 & 0.02 &   0.05 & 0.03 & 0.02 \\
\bottomrule
\end{tabular}
\begin{tablenotes}[para,flushleft]
  \footnotesize First row displays rejection percentages for two-sided $t$-tests at the 5\% level. Second row displays standard errors. Results obtained from 15k simulations. ``GS HAC'' corresponds to the generalized spatial HAC and ``Oracle'' to true standard errors (obtained by simulation).
\end{tablenotes}
\end{threeparttable}
\label{YrrSE}
\end{table}

%----------------------------------------------------------------------
\section{Conclusion}\label{sconclude}
%----------------------------------------------------------------------

The increasing availability of network data has enabled researchers to better study heterogeneous interactions and subject to empirical inquiry a large body of theoretical work on the impact of network topology on the diffusion of behaviors, products, and information \citep{jackson2017economic}. The central challenge of drawing valid inference from network data is that it typically consists of autocorrelated observations from a single large network. This is an inherent problem in models with social interactions, since by definition, an agent's outcome is a function of her neighbors'. Consequently, standard procedures that assume i.i.d.\ data can produce highly misleading inference. 

In the context of discrete choice models, we show that spatial and network HAC estimators can be used to correct for network autocorrelation. In an empirical application to \cite{conley2010learning}, the spatial and network HAC estimators deliver similar results, which supports the theory. Our simulation study illustrates the potential for naive i.i.d.\ standard errors to be highly anti-conservative and demonstrates the accuracy of the normal approximation and HAC estimators.

Our main technical results are CLTs for static and dynamic models of social interactions and consistency results for the HAC estimators. In the existing literature, first-order stationarity conditions are always required for HAC estimators to approximate the variance. We also provide results for the spatial HAC in settings where stationarity fails to hold, which include models defined by moment inequalities and many nonparametric models. For the network HAC, its validity when first-order stationarity fails is an open question. The advantage of the network HAC is that it does not require additional information on spatial/social positions. On the other hand, for it to be valid in static models, we require the network to be exogenous (orthogonal to the errors in the outcome model), whereas no such assumption is needed for the spatial HAC. When exogeneity fails, the dependence structure of the data conditional on the network is difficult to characterize, and it is an open question whether the network HAC is valid in this setting.

%----------------------------------------------------------------------

\FloatBarrier
\phantomsection
\addcontentsline{toc}{section}{References}
\bibliography{net_interact}{} 
\bibliographystyle{aer}

%----------------------------------------------------------------------

\newpage
\part{Supplemental Appendix}
\appendix
\renewcommand{\thesection}{SA.\arabic{section}} % custom SA section numbering

%----------------------------------------------------------------------
\section{Decentralized Selection}\label{scltassumps}
%----------------------------------------------------------------------

This section states a restriction on the selection mechanism referenced in \autoref{ssma2} that is used to prove the CLTs. Recall from \autoref{Rci} the definition of the indicator $\mathcal{R}^c_j$. Define an artificial directed network $\bm{D}$ on $\mathcal{N}_n$ such that $D_{ij} = A_{ij} \mathcal{R}^c_j$. Let $C_i$ be the strongly connected component of $\bm{D}$ that contains agent $i$.\footnote{A strongly connected component $C$ is a set of agents such that there exists a directed path from $i$ to $j$ and vice versa for every $i,j \in C$.  A directed path of length $m$ is a sequence of unique agents $k_1, k_2, \dots, k_m$ such that $D_{k_q,k_{q+1}}=1$ for any $q = 1, \dots, m-1$.} Define $i$'s {\em strategic neighborhood} as
\begin{equation}
  C_i^+ \equiv C_i \medcup \left\{k \in \mathcal{N}_n\colon \max_{k \in C_i} A_{jk}(1-\mathcal{R}_i^c)=1 \right\}. \label{stratneigh}
\end{equation}

\noindent This adds to $C_i$ the set of agents $k$ linked to a member of $C_i$ such that either $\inf_s U(s, X_k, \varepsilon_k) > 0$ or $\sup_s U(s, X_k, \varepsilon_i) \leq 0$. These are agents that choose action 1 or 0, respectively, regardless of the choices of other agents in the network. Let $\mathcal{C}^+ = \{C_i^+\colon i \in \mathcal{N}_n\}$ be the collection of strategic neighborhoods. (In fact this is the set of weakly connected components of $\bm{D}$.)

If $C^+$ is a strategic neighborhood, under \autoref{S}, it has the unique property that
\begin{equation}
  Y_{C^+} \in \mathcal{E}(A_{C^+},\tau_{C^+}) \quad\forall\, \bm{Y} \in \mathcal{E}(\bm{A}, \bm{\tau}). \label{snprop}
\end{equation}

\noindent In other words, for any Nash equilibrium vector $\bm{Y}$, if we remove all agents not in $C^+$ from the game, then the subvector $Y_{C^+}$ is still a Nash equilibrium in the resulting game. A formal proof is given in Lemma B.2 of \cite{leung2019compute}; see \S 2.1 of that paper for a detailed discussion. This property is typically not satisfied for most subsets of agents due to strategic interactions; removing an agent outside of a given set will usually affect the optimality of the actions of agents within the set.

The implication of \autoref{snprop} is that strategic neighborhoods effectively function as distinct subgames, even though they can have overlapping agents.\footnote{The overlapping agents are those $j$ for which $\mathcal{R}_j^c=0$, which implies they choose the same action regardless of the actions of agents in any strategic neighborhood of which $j$ is a member.} Our next assumption requires that these neighborhoods do not coordinate in the sense illustrated in \autoref{ssma2}. For $C^+ \in \mathcal{C}^+$ and $\lambda(\cdot)$ defined in \autoref{select}, let $\lambda(\bm{A},\bm{\tau})\big|_{C^+}$ be the restriction of the range of $\lambda(\cdot)$ to $\mathcal{E}(A_{C^+},\tau_{C^+})$, the set of Nash equilibria on $C^+$. Thus if $\bm{Y}$ is the observed equilibrium outcome, so that $\bm{Y} = \lambda(\bm{A},\bm{\tau})$, then $\lambda(\bm{A},\bm{\tau})\big|_{C^+} = Y_{C^+}$.

\begin{assump}[Decentralized Selection]\label{nocoord}
  Let $\lambda(\cdot)$ be defined in \autoref{select}. There exists a selection mechanism $\lambda'(\cdot)$, such that for any network size $n$ and strategic neighborhood $C^+ \in \mathcal{C}^+$, $\lambda(\bm{A},\bm{\tau})\big|_{C^+} = \lambda'(A_{C^+},\tau_{C^+})$.
\end{assump}

\noindent This states that the process in which agents in $C^+$ coordinate on a Nash equilibrium subvector $Y_{C^+} \in \mathcal{E}(A_{C^+},\tau_{C^+})$ according to $\lambda(\cdot)$ only depends on the subnetwork $A_{C^+}$ and types $\tau_{C^+}$ of agents in $C^+$. Myopic best-response dynamics trivially satisfy this property because agents' best responses only depend on the types of their neighbors when the initial condition is chosen in a manner that only depends on $C^+$.

%----------------------------------------------------------------------
\section{Path and Spatial Distances}\label{rggpd}
%----------------------------------------------------------------------

We state primitive conditions for \autoref{pd} under the following random geometric graph (RGG) model: 
\begin{equation*}
  A_{ij} = \bm{1}\{\norm{\rho_i-\rho_j}\leq 1\} \bm{1}\{i\neq j\} \quad \forall\, i,j\in\mathcal{N}_n,
\end{equation*}

\noindent where positions are uniformly distributed, so that $f \sim \mathcal{U}([0,1]^d)$. To our knowledge, the only graph-theoretic results available for proving \autoref{pd} pertain to this model \citep[see][section II.B for related literature]{kartun2018counting}. However, we conjecture that similar results hold for the larger class of random graphs considered in our paper, since \autoref{sparsity}(c) implies connection probabilities are exponentially decaying with $\norm{\rho_i-\rho_j}$, whereas the RGG model requires a hard threshold at 1 (or any other constant).

Given that $\rho_i = \omega_n \tilde\rho_i$ (\autoref{lna}), the RGG model is equivalent to a rescaled model in which positions are instead drawn uniformly from $[0,n^{1/d}]^d$ and $A_{ij} = \bm{1}\{\norm{\rho_i-\rho_j}\leq \kappa^{1/d}\}$. The latter model is studied in \cite{friedrich2013diameter}, who provide bounds on path distances in terms of social distance. We draw on their results to prove \autoref{pd}. Some of their results utilize the ``Poissonized'' RGG model, which is the same as the RGG model, except it replaces the set of $n$ positions $\{\rho_i\}_{i=1}^n$ with $\{\rho_i\}_{i=1}^{N_n+2}$, where $N_n \sim \text{Poisson}(n)$ is independent of $\{\rho_i\}_{i\in\mathbb{N}} \sim f$. 

\begin{lemma}
  Under the RGG model, \autoref{pd}(a) holds. 
\end{lemma}
\begin{proof}
  Fix $c'>0$. By the law of total probability,
  \begin{multline*}
    \prob(\ell_{\bm{A}}(i,j) > c \norm{\rho_i-\rho_j} \mid \ell_{\bm{A}}(i,j) < \infty) \leq \prob(\norm{\rho_i-\rho_j} \leq c' \log n \mid \ell_{\bm{A}}(i,j) < \infty) \\
    + \prob(\ell_{\bm{A}}(i,j) > c \norm{\rho_i-\rho_j} \mid \ell_{\bm{A}}(i,j) < \infty, \norm{\rho_i-\rho_j} > c' \log n).
  \end{multline*}

  \noindent Consider the second probability on the right-hand side. If this were instead under the Poissonized model, then by \S 3.3 of \cite{friedrich2013diameter}, for $c'$ chosen large enough, it would be $O(n^{-3})$.
  % the argument in the paper could either be for the poissonized model where \mathcal{P}_{nf} is the set of agent positions or the model where we add two additional agents. It's hard to tell with their loose use of conditioning, so I'm just gonna assume they hold
  Then by applying Lemma 1 of that paper, we have that the term is $O(n^{-2.5})$ under the original (not Poissonized) model.
  % Seems like everything in that lemma should be conditional on i,j being connected. So we should actually be dealing with P(N_n=n | \ell_{\bm{A}}(i,j) < \infty). But this is \geq \prob(\ell_{\bm{A}}(i,j) < \infty | N_n=n) P(N_n=n). The first probability converges to a non-zero constant in the supercritical regime, as argued below.
  Next, the first probability on the right-hand side is $O( (r_n \log n)^d ) = O(n^{-(1-\epsilon)})$ for any $\epsilon>0$. Therefore, $\prob(\norm{\rho_i-\rho_j} \leq c' \log n \mid \ell_{\bm{A}}(i,j) < \infty) = O(n^{-(1-\epsilon)})$ for $\epsilon = 0.25$.
\end{proof}

\begin{lemma}
  Under the RGG model, \autoref{pd}(b) holds. 
\end{lemma}
\begin{proof}
  Our argument follows the last paragraph of the proof of Lemma 20 of \cite{friedrich2013diameter}. Let $Q$ be the cube centered at $\rho_i$ with side length $c$ and $Q'$ the cube centered at $\rho_i$ with side length $2c$ with the same orientation as $Q$. Then $Q'\backslash Q$ is an annulus, and we can apply a crossing components argument. For the case $d=2$, see the use of Lemma 3.3 in \cite{bradonjic2010efficient} for a detailed argument, which modifies Lemma 10.5 of \cite{penrose2003}; see in particular Figure 3(c) for intuition. For the general case, \cite{friedrich2013diameter} use Proposition 10.6 of \cite{penrose2003}. In words, the crossing component argument shows that for any agent $j$ connected to $i$ such that $\rho_j \in Q$, there exists a path in $\bm{A}$ from $i$ to $j$ completely contained in $Q'$ with probability approaching one. Then since the shortest path is at least as long, it is bounded by the size of this path. But this path should be bounded in length, since it is contained in $Q'$, which has bounded size. Formally, by Lemma 14 of \cite{friedrich2013diameter}, $\ell_{\bm{A}}(i,j)$ is bounded above by a constant with probability approaching one, which proves the claim.
\end{proof}

%----------------------------------------------------------------------
\section{Proofs for Variance Estimators}\label{shacpf}
%----------------------------------------------------------------------

%------------------------------------
\subsection{Spatial HAC}\label{sstationpf}
%------------------------------------

{\bf Proof Sketch.} Consider the generalized spatial HAC in \autoref{nonstathac}, which is consistent for the variance, as discussed in that section. We can replace occurrences of $\hat\theta(p)$ in the estimator with $\bar{\psi}$ without affecting consistency due to \autoref{station}. We then show that the resulting modification of $\hat\sigma^2 - \hat\alpha\hat\alpha'$ can be rewritten as $\hat{\bm{\Sigma}}_\rho - (\hat\alpha - \bar{\psi}) (\hat\alpha - \bar{\psi})'$. The second term is clearly asymptotically positive semidefinite and can be shown to converge to some $n$-dependent constant, which we call $\bm{\Sigma}_n - \tilde{\bm{\Sigma}}_n$. 

\bigskip

\begin{proof}[Proof of \autoref{spatialhac}]
  Consider the generalized spatial HAC in \autoref{snostat} with $\hat\theta(p)$ replaced with $\bar{\psi}$. Define $\hat\sigma^2 = \hat v + \hat c$ and $\hat\alpha = \bar{\psi} + \hat\delta$, where 
  \begin{align*}
    &\hat v = \frac{1}{n} \sum_{i=1}^n \psi_i \psi_i', \\
    &\hat c = \frac{1}{n} \sum_{i=1}^n \sum_{j \neq i} (\psi_i \psi_j' - \bar{\psi}\bar{\psi}') K\left( (\rho_i-\rho_j) / h_n \right), \\ 
    &\hat\delta = \frac{1}{n} \sum_{i=1}^n \sum_{j\neq i} (\psi_i - \bar{\psi}) K((\rho_i-\rho_j) / h_n),
  \end{align*}

  \noindent $K(\cdot)$ is a kernel function satisfying \autoref{kernel}, and the bandwidth $h_n$ is a diverging sequence satisfying $h_n = O(n^{1/(3d)})$. The estimator is given by
  \begin{equation*}
    \hat\sigma^2 - \hat\alpha\hat\alpha'. 
  \end{equation*}
  
  \noindent \autoref{station} allows us to use $\bar{\psi}$ as a $\sqrt{n}$-consistent estimator for $\E[\psi_1 \mid \rho_1]$ in place of $\hat\theta(\rho_1)$ used in \cite{leung2019normal}. We can then weaken the bandwidth rate from $h_n = O(n^{1/(4d)})$ to $h_n = O(n^{1/(3d)})$.
  
  Some algebra shows that
  \begin{equation}
    \hat\sigma^2 - \hat\alpha\hat\alpha' = \frac{1}{n} \sum_{i=1}^n \sum_{j=1}^n (\psi_i - \bar{\psi}) (\psi_j - \bar{\psi})' K((\rho_i - \rho_j) / h_n) - \hat\delta\hat\delta'. \label{hacrewrite}
  \end{equation}

  \noindent Following verbatim the argument in the proof of Theorem G.1 in \cite{leung2019normal}, we have that there exist $\sigma_n^2$, $\alpha_n$, $\delta_n$ such that $\abs{\hat\sigma^2-\sigma_n^2} \plimarrow 0$, $\abs{\hat\alpha - \alpha_n} \plimarrow 0$, and $\abs{\hat\delta - \delta_n} \plimarrow 0$, and $\abs{\bm{\Sigma}_n - (\sigma_n^2 - \alpha_n\alpha_n')} \rightarrow 0$.\footnote{$\sigma_n^2$ and $\alpha_n$ are defined just prior to equation (G.1) and $\delta_n$ in (G.3). In our setting, a slight modification is that we replace $\bm{p}_{x,r_n}$ everywhere with $\bm{c}_{x,r_n}$, the latter defined in our \autoref{sgsstatic}.} Therefore, using \autoref{hacrewrite},
  \begin{equation}
    \frac{1}{n} \sum_{i=1}^n \sum_{j=1}^n (\psi_i - \bar{\psi}) (\psi_j - \bar{\psi})' K((\rho_i - \rho_j) / h_n) = \bm{\Sigma}_n + \delta_n\delta_n' + o_p(1). \label{gwve0b94}
  \end{equation}
  
  \noindent Since $\delta_n\delta_n'$ is positive semidefinite and the left-hand side is exactly $\hat{\bm{\Sigma}}_\rho$, the first claim of the theorem follows.

  For the second claim, let $N \sim \text{Poisson}(n)$ be independent of all other primitives. As in \autoref{networkstation}, let $\psi_i^N$ be $i$'s agent statistic under the (static or dynamic) model where the number of agents is $N$. Recall the definition of the add-one cost $\Xi_n$ from \autoref{ssmclt}, and let $\Xi_N$ be its analog in the model with $N$, rather than $n$, agents. Then, as defined in the proof of Theorem G.1 of \cite{leung2019normal}, 
  \begin{equation*}
    \delta_n = \E[\Xi_N] - \E[\psi_{N+1}^{N+1}],
  \end{equation*}

  \noindent which is the indirect effect on the total agent statistic from removing agent $N+1$ from the model. Since $\E[\psi_1]=c$ for all $n$, it follows that $\E[\Xi_N \mid N] - \E[\psi_{N+1}^{N+1} \mid N] = \zero$ a.s., and therefore, $\delta_n = \zero$ for all $n$. Consistency of $\hat{\bm{\Sigma}}_\rho$ then follows from \autoref{gwve0b94}.
\end{proof}

%------------------------------------
\subsection{Network HAC}\label{snethacpf}
%------------------------------------

{\bf Proof Sketch.} We prove \autoref{networkhac} by decomposing $\hat{\bm{\Sigma}}_{\bm{A}}$ into two terms: a leading term with the summands multiplied by $\bm{1}\{\norm{\rho_i-\rho_j} \leq b_n\}$ and a remainder term with the summands multiplied by $\bm{1}\{\norm{\rho_i-\rho_j} > b_n\}$ for a certain polynomially diverging sequence $\{b_n\}_{n\in\mathbb{N}}$. The leading term is shown to converge to the target. The remainder term is shown to be $o_p(1)$ because agents that are spatially distant ($\norm{\rho_i-\rho_j} > b_n$) cannot also be close in terms of path distance (i.e.\ $\ell_{\bm{A}}(i,j) \leq h_n$, as required by the HAC kernel weight). The logarithmic rate of $h_n$ is important for the latter argument. A more detailed sketch of the argument for the remainder term is given in the proof below. We next expound on the leading term.

The leading term almost has the form of a spatial HAC estimator, since $\bm{1}\{\norm{\rho_i-\rho_j} \leq b_n\}$ satisfies the conditions for a kernel function, but unfortunately we also have the network-dependent term $K(\ell_{\bm{A}}(i,j) / h_n)$ floating around (since it is a network HAC). We need to show that the leading term is consistent for the desired target, namely $\E[n^{-1} \sum_{i,j} (\psi_i - \E[\psi_1]) (\psi_j - \E[\psi_1])']$, which is just the expectation of the leading term without the kernel or indicator weights. This is done by (1) showing that the leading term is close to its analog under a Poissonized model where the number of agent is $N \sim \text{Poisson}(n)$, and (2) doing mean-variance calculations under the Poissonized model. The interesting aspect of this argument is showing that the mean converges to the target. This is usually straightforward in proofs of the consistency of HAC estimators because under weak conditions, the kernel converges to 1 a.s.\ as the bandwidth diverges. In our context, this is easy to establish for the indicator part of the kernel $\bm{1}\{\norm{\rho_i-\rho_j} \leq b_n\}$, since $b_n \rightarrow\infty$. However, without additional conditions, $K( \ell_{\bm{A}}(i,j) / h_n )$ does not a.s.\ converge to one for any $i,j$. This is because with positive probability, $i$ and $j$ can be disconnected in the limit, meaning $\lim_{n\rightarrow\infty} \ell_{\bm{A}}(i,j) = \infty$. For such agents, $K( \ell_{\bm{A}}(i,j) / h_n )=0$ by \autoref{kernel}, which creates a bias. 
  
This is where Assumptions \ref{networkstation} and \ref{netcexog} are needed; they are used to establish that $\cov(\psi_i, \psi_j \mid \bm{X}, \bm{A}) = 0$ if $\ell_{\bm{A}}(i,j) = \infty$, so disconnected agents are dropped from the covariance, eliminating this bias. This sheds light on the difficulty of allowing for an endogenous network. Without \autoref{netcexog}, the unobservables $\bm{\varepsilon}$ can have a complicated dependence structure conditional on $\bm{A}$, so the conditional covariance is not typically zero.

The consistency proof in \cite{kojevnikov2019limit} is much more straightforward than ours due to the notion of weak dependence used, which essentially implies that $\cov(\psi_i, \psi_j \mid \bm{A}) \rightarrow 0$ at a fast rate as $\ell_{\bm{A}}(i,j) \rightarrow \infty$. We use a fundamentally different weak dependence property (stabilization) proposed by \cite{leung2019normal}, which is unconditional in nature. It lends itself more easily to establishing consistency of $\hat{\bm{\Sigma}}_\rho$ rather than $\hat{\bm{\Sigma}}_{\bm{A}}$, so for this reason, our strategy is to essentially to approximate the latter with a version of the former.

\bigskip

\begin{proof}[Proof of \autoref{networkhac}]
  The second claim of the theorem is shown in the exact same way as \autoref{spatialhac}, so we only focus on the first claim. Define $\hat\sigma^2, \hat\alpha, \hat\delta, \sigma_n^2, \alpha_n, \delta_n$ as in the proof of \autoref{spatialhac}, and recall from that proof that $\abs{\hat\delta - \delta_n} \plimarrow 0$. It suffices to show that
  \begin{equation*}
    \lvert \hat{\bm{\Sigma}}_{\bm{A}} - \big( \sigma_n^2 - \alpha_n\alpha_n' + \delta_n\delta_n' \big) \rvert \plimarrow 0,
  \end{equation*}

  \noindent since, as discussed in that proof, $\abs{\bm{\Sigma}_n - (\sigma_n^2 - \alpha_n\alpha_n')} \rightarrow 0$. Note that the notation here is the same for both the static and dynamic models. Let $\{b_n\}_{n\in\mathbb{N}}$ be a sequence of positive numbers such that
  \begin{equation}
    b_n = o(n^{1/(16d)}). \label{b_n}
  \end{equation}
  
  \noindent Decompose 
  \begin{multline}
    \hat{\bm{\Sigma}}_{\bm{A}} = \frac{1}{n} \sum_{i=1}^n \sum_{j\neq i} (\psi_i - \bar{\psi}) (\psi_j - \bar{\psi})' K\left( \ell_{\bm{A}}(i,j) / h_n \right) \bm{1}\{\norm{\rho_i-\rho_j} \leq b_n\} \\
    + \frac{1}{n} \sum_{i=1}^n \sum_{j\neq i} (\psi_i - \bar{\psi}) (\psi_j - \bar{\psi})' K\left( \ell_{\bm{A}}(i,j) / h_n \right) \bm{1}\{\norm{\rho_i-\rho_j} > b_n\}. \label{hacdecomp}
  \end{multline}

  \noindent \autoref{1stterm} shows that the first term on the right-hand side is consistent for $\bm{\Sigma}_n + \delta_n\delta_n'$, and \autoref{2ndterm} shows that the second term is $o_p(1)$, which completes the proof.
  
  {\bf Detailed proof sketch.} The remainder of the proof provides a high-level sketch of the lemmas. The first term on the right-hand side is almost the same as the first term on the right-hand side of \autoref{hacrewrite}. We can effectively think of 
  \begin{equation*}
    \tilde K(\rho_i-\rho_j)/b_n) \equiv K\left( \ell_{\bm{A}}(i,j) / h_n \right) \bm{1}\{\norm{\rho_i-\rho_j} \leq b_n\}
  \end{equation*}
  
  \noindent as the kernel of a spatial HAC estimator, except for the unusual feature that this kernel also depends on $\bm{A}$. Following the algebra for \autoref{hacrewrite}, we have
  \begin{equation*}
    \frac{1}{n} \sum_{i=1}^n \sum_{j\neq i} (\psi_i - \bar{\psi}) (\psi_j - \bar{\psi})' K\left( \ell_{\bm{A}}(i,j) / h_n \right) \bm{1}\{\norm{\rho_i-\rho_j} \leq b_n\} 
    = \tilde\sigma^2 - \tilde\alpha\tilde\alpha' + \tilde\delta\tilde\delta',
  \end{equation*}

  \noindent where $\tilde\sigma^2$, $\tilde\alpha$, and $\tilde\delta$ are analogs of $\hat\sigma^2$, $\hat\alpha$, and $\hat\delta$ with $K((\rho_i-\rho_j)/h_n)$ replaced with $\tilde K((\rho_i-\rho_j)/b_n)$.

  If it weren't for the presence of $K\left( \ell_{\bm{A}}(i,j) / h_n \right)$, we would directly apply Theorem G.1 of \cite{leung2019normal} to obtain 
  \begin{equation*}
    \lvert \tilde\sigma^2 - \tilde\alpha\tilde\alpha' + \tilde\delta\tilde\delta' - \big( \sigma_n^2 - \alpha_n\alpha_n' + \delta_n\delta_n' \big) \rvert \plimarrow 0,
  \end{equation*}

  \noindent which is the desired result. The presence of the kernel, however, creates a number of complications. To understand them, we first recall the broad outlines of the proof of Theorem G.1. It proceeds through three steps. Step (1) shows that these quantities are consistent for their ``Poissonized'' analogs. By Poissonized analogs we mean the same estimators except we change the underlying data-generating process by replacing the vector of positions $(\rho_i)_{i=1}^n$ with $(\rho_i)_{i=1}^{N_n}$, where $N_n \sim \text{Poisson}(n)$ is independent of all other primitives, $\rho_i = \omega_n\tilde\rho_i$ for all $i$, and  $\tilde\rho_1, \tilde\rho_2, \dots$ are i.i.d.\ with density $f$.  Then the Poissonized analogs are shown to have (2) vanishing variances and (3) expectations that tend to the target estimands, namely $\sigma_n^2, \alpha_n, \delta_n$. This is easier to show for the Poissonized analogs as opposed to the original estimators because $(\rho_i)_{i=1}^{N_n}$ has the distribution of a Poisson point process, which has a convenient spatial independence property.
  
  We need to modify these arguments to account for the presence of $K\left( \ell_{\bm{A}}(i,j) / h_n \right)$. Step (2) in the previous paragraph proceeds almost exactly the same. The only difference is that we use the upper bound
  \begin{equation}
    \big \lvert K\left( \ell_{\bm{A}}(i,j) / h_n \right) \bm{1}\{\norm{\rho_i-\rho_j} \leq b_n\} \big \rvert \leq K^* \bm{1}\{\norm{\rho_i-\rho_j} \leq b_n\} \label{upperbound}
  \end{equation}

  \noindent due to \autoref{kernel}. Since the proof of step (2) consists of a series of upper bounds, this inequality enables us to basically ignore the presence of $K\left( \ell_{\bm{A}}(i,j) / h_n \right)$. 

  Next we turn to step (1). This also uses \autoref{upperbound}, as well as \autoref{pd}. The argument here is mostly technical, so we will not summarize it. However, recall the discussion of the importance of \autoref{pd} following its statement in \autoref{svar}.

  Finally we turn to step (3). Recall that this step involves showing that the expectations of the Poissonized estimators converge to their target estimands. Let us decompose this into two smaller steps. (3a) We show that expectations of the Poissonized estimators converge to analogs that replace $\tilde K((\rho_i-\rho_j)/b_n)$ with 1. Call these analogs their {\em target means} (as opposed to target estimands). (3b) We show that the target means converge to the target estimands. Step (3b) follows from step 1 of the proof of Theorem G.1. The more interesting complication is step (3a).

  Step (3a) is usually straightforward in proofs of HAC estimator consistency because under weak conditions, the kernel converges to 1 a.s.\ as the bandwidth diverges. In our context, this is easy to establish for the indicator part of the kernel $\bm{1}\{\norm{\rho_i-\rho_j} \leq b_n\}$ \cite[see][]{leung2019normal}. However, without additional conditions, $K\left( \ell_{\bm{A}}(i,j) / h_n \right)$ does not a.s.\ converge to one for any $i,j$. This is because with positive probability, $i$ and $j$ can be disconnected in the limit, meaning $\lim_{n\rightarrow\infty} \ell_{\bm{A}}(i,j) = \infty$. For such agents, $\tilde K((\rho_i-\rho_j)/b_n)=0$ by \autoref{kernel}, which creates a bias in $\tilde\sigma^2$. This is where Assumptions \ref{networkstation} and \ref{netcexog} are needed; they ensure that $\cov(\psi_i, \psi_j \mid \bm{X}, \bm{A}) = 0$ if $\lim_{n\rightarrow\infty} \ell_{\bm{A}}(i,j) = \infty$, so disconnected agents are dropped from the covariance, eliminating this bias. 
\end{proof}

%----------------------
\begin{lemma}\label{1stterm}
   Under the assumptions of \autoref{networkhac} (see the definitions made in the proof of that theorem),
  \begin{equation*}
    \bigg| \frac{1}{n} \sum_{i=1}^n \sum_{j\neq i} (\psi_i - \bar{\psi}) (\psi_j - \bar{\psi})' K\left( \ell_{\bm{A}}(i,j) / h_n \right) \bm{1}\{\norm{\rho_i-\rho_j} \leq b_n\} - \left(\bm{\Sigma}_n + \delta_n\delta_n'\right) \bigg| \plimarrow 0.
  \end{equation*}
\end{lemma}
\begin{proof}
  We formalize the argument sketched in the proof of \autoref{networkhac}, in particular steps (1) and (3a).

  {\bf Step (1).} The parts of the proof of Theorem G.1 of \cite{leung2019normal} that need to be modified are steps 1 and 2 of the proofs of Lemmas G.3 and G.6. We first consider G.3. Note that both proofs for simplicity only consider the case where $\psi_i$ is scalar, as we do next. 

  The only modification of step 1 needed is to account for the fact that our mean estimator is $\bar{\psi}$ rather than $\hat\theta(\rho_1)$, so we need to show that the former uniformly converges to $\E[\psi_1 \mid \rho_1]$. We claim that
  \begin{equation*}
    \bar{\psi} \plimarrow \E[\psi_1], \quad\text{and}\quad \E[\psi_1] = \E[\psi_1 \mid \bm{X}, \bm{A}] = \E[\psi_1 \mid \rho_1], 
  \end{equation*}

  \noindent where convergence occurs at a $n^{-1/2}$ rate by our CLT (\autoref{staticclt} for the static model, \autoref{dynamicclt} for the dynamic model). The first equality is direct from \autoref{networkstation}. The second equality follows because by the first equality,
  \begin{equation}
    \E\left[ \E[\psi_1] \mid \rho_1 \right] = \E\left[ \E[\psi_1 \mid \bm{X}, \bm{A}] \mid \rho_1 \right] = \E\left[ \E[\psi_1 \mid \bm{X}, \bm{A}, \rho_1] \mid \rho_1 \right] = \E[\psi_1 \mid \rho_1]. \label{2390d}
  \end{equation}

  \noindent The second equality above holds because conditional on $\bm{X},\bm{A}$, the randomness in the agent statistic only results from $\bm{\varepsilon}$, and this is conditionally independent of $\bm{\rho}$ by \autoref{netcexog}.
  
  Now we turn to step 2. Let $\psi_i^m, \tilde K_m(\cdot)$ be analogs of $\psi_i, \tilde K(\cdot)$ except we change the underlying data-generating process by changing the number of agents to $m$ rather than $n$ but keep $\omega_n$ defined as before. Define $\theta_m(p) = \E[\psi_i^m \mid \rho_i=p]$. The term $\rho_m$ in the proof of step 2 in Lemma G.3 is given in our context by
  \begin{equation*}
    \rho_m \equiv \sum_{i=1}^n \sum_{j\neq i} (\psi_i^m \psi_j^m - \theta_m(\rho_i) \theta_m(\rho_j)) \tilde K_m(\norm{\rho_i-\rho_j} / b_n).
  \end{equation*}

  \noindent We need to show that
  \begin{equation*}
    n^{-1}\abs{\rho_{N_n}-\rho_n} \plimarrow 0,
  \end{equation*}

  \noindent where $N_n \sim \text{Poisson}(n)$ is independent of all other model primitives.

  Fix $\varepsilon>0$. We show that the following quantity is $o_p(1)$:
  \begin{multline*}
    \underbrace{\E\left[ \bm{1}\left\{ n^{-1} \abs{\rho_{N_n} - \rho_n} > \varepsilon \right\} \bm{1}\left\{ \abs{N_n - n} \leq n/2 \right\} \right]}_{[I]} \\ + \underbrace{\E\left[ \bm{1}\left\{ n^{-1} \abs{\rho_{N_n} - \rho_n} > \varepsilon \right\} \bm{1}\left\{ \abs{N_n - n} > n/2 \right\} \right]}_{[II]}. 
  \end{multline*}

  \noindent By Lemma 1.4 of \cite{penrose2003},
  \begin{equation*}
    \prob\left( \abs{N_n - n} > n/2 \right) \rightarrow 0, 
  \end{equation*}

  \noindent and therefore $[II] \rightarrow 0$. On the other hand,
  \begin{align*}
    [I] &= \sum_{m\colon \abs{m-n} \leq n/2} \prob\left( n^{-1} \abs{\rho_m - \rho_n} > \varepsilon \right) \prob(N_n = m) \\
	&= \sum_{m\colon \abs{m-n} \leq n/2} \prob\left( n^{-1} \bigg| \sum_{k = n}^{m-1} \left( \rho_{k+1} - \rho_k \right) \bigg| > \varepsilon \right) \prob(N_n = m) \\
    &\leq \prob\left( \bigg| \frac{N_n-1-n}{\sqrt{n}} \bigg| \sup_{m\colon \abs{m-n}\leq n/2} \sup_{k \in [n,m-1]} n^{-1/2} \abs{\rho_{k+1} - \rho_k \big} > \varepsilon \right). 
  \end{align*}

  \noindent The term in the absolute value on the last line is $O_p(1)$ by the CLT. It then suffices to show that the term multiplying it is $o_p(1)$. 
  
  Observe that for any sequence $\{k_n\}_{n\in\mathbb{N}}$ such that $k_n \in [n,m-1]$ and $m$ satisfies $\abs{m-n}\leq n/2$, we have $k_n/n \rightarrow c \in (0,\infty)$. Therefore it is enough to show that
  \begin{equation*}
    n^{-1/2} \big(\rho_{k+1} - \rho_k\big) = o_p(1) 
  \end{equation*}

  \noindent for any $k \equiv k_n$ satisfying $k/n \rightarrow c \in (0,\infty)$. 
  
  Decompose
  \begin{equation*}
    n^{-1/2} \left( \rho_{k+1} - \rho_k \right) = [A1] + [A2] + [B] + [C], 
  \end{equation*}
  
  \noindent where 
  \begin{align*}
    &[A1] = n^{-1/2} \sum_{i=1}^k \left( \psi_i^{k+1} \psi_{k+1}^{k+1} - \theta_{k+1}(\rho_i) \theta_{k+1}(\rho_{k+1}) \right) \tilde K_{k+1}\left( \norm{\rho_i-\rho_{k+1}}/b_n \right), \\
    &[A2] = n^{-1/2} \sum_{j=1}^k \left( \psi_{k+1}^{k+1} \psi_j^{k+1} - \theta_{k+1}(\rho_{k+1}) \theta_{k+1}(\rho_j) \right) \tilde K_{k+1}\left( \norm{\rho_{k+1}-\rho_j}/b_n \right), \\
    &\begin{aligned} [B] = n^{-1/2} \sum_{i=1}^k &\sum_{j\neq i, j=1}^k \left( \psi_i^{k+1} \psi_j^{k+1} \right. K(\ell_{\bm{A}^{k+1}}(i,j)/h_n)  \\  & \left.- \psi_i^k \psi_j^k K(\ell_{\bm{A}^k}(i,j)/h_n) \right) \bm{1}\{\norm{\rho_i-\rho_j} \leq b_n\}, \end{aligned} \\
    &\begin{aligned} [C] = n^{-1/2} \sum_{i=1}^k &\sum_{j\neq i, j=1}^k \left( \theta_{k+1}(\rho_i)\theta_{k+1}(\rho_j) K(\ell_{\bm{A}^{k+1}}(i,j)/h_n) \right.  \\  & \left.- \theta_k(\rho_i)\theta_k(\rho_j) K(\ell_{\bm{A}^k}(i,j)/h_n) \right) \bm{1}\{\norm{\rho_i-\rho_j} \leq b_n\}, \end{aligned}
  \end{align*}

  \noindent where $\ell_{\bm{A}^m}(i,j)$ is the path distance between agents $i$ and $j$ in the network consisting of agents $1, \dots, m$.  We show that these three terms converge in $L_1$ to zero. The argument for $[A1]$ and $[A2]$ is identical to that in Lemma G.3 of \cite{leung2019normal}, since using the upper bound \autoref{upperbound} allows us to ignore $K(\ell_{\bm{A}}(i,j)/h_n)$ in the kernel.

  We only consider $[B]$ here, as the modification of $[C]$ is very similar. Let $\RR_i^m$ be the radius of stabilization of agent $i$ in the model with $m$ agents. For the static model, this is defined in \autoref{RR_i} for the case $\X = \X_m$ and for the dynamic model in \autoref{RR_iT}. Define the event
  \begin{multline*}
    \mathcal{E}_n = \left\{ \RR_i^{k+1} \leq \norm{\rho_i-\rho_{k+1}}/2 \medcap \RR_j^{k+1} \leq \norm{\rho_j-\rho_{k+1}}/2 \right. \\ \left. \medcap \RR_i^k \leq \norm{\rho_i-\rho_{k+1}}/2 \medcap \RR_j^k \leq \norm{\rho_j-\rho_{k+1}}/2 \right\}. 
  \end{multline*}
  
  \noindent Then 
  \begin{multline}
    \E[\lvert[B]\rvert] \leq n^{-1/2}k^2 \E\left[ \lvert \psi_i^{k+1} \psi_j^{k+1} K(\ell_{\bm{A}^{k+1}}(i,j)/h_n) - \psi_i^k \psi_j^k K(\ell_{\bm{A}^k}(i,j)/h_n) \rvert \right. \\ \left. \times \bm{1}\{\norm{\rho_i-\rho_j} \leq b_n\} \bm{1}\{\mathcal{E}_n^c\} \right] \\ + n^{-1/2}k^2 \E\left[ \lvert \psi_i^{k+1} \psi_j^{k+1} K(\ell_{\bm{A}^{k+1}}(i,j)/h_n) - \psi_i^k \psi_j^k K(\ell_{\bm{A}^k}(i,j)/h_n) \rvert \right. \\ \left. \times \bm{1}\{\norm{\rho_i-\rho_j} \leq b_n\} \bm{1}\{\mathcal{E}_n\} \right]. \label{EB1}
  \end{multline}
  
  \noindent The first term on the right-hand side of \autoref{EB1} is $o_p(1)$ following the original argument in the proof of Lemma G.3 (the same lemma being modified here), where we use \autoref{upperbound} on occasion to ignore the network kernel. So consider the second term. Under the event $\mathcal{E}_n$, $\psi_i^{k+1} = \psi_i^k$ (and likewise for $j$), as argued in the proof of Lemma G.3. In brief, the event says that the position of agent $k+1$ lies outside the cube centered at $i$ with radius equal to $\max\{\RR_i^{k+1}, \RR_i^k\}$. The radius of stabilization has the property that the removal of agents with positions outside of this ball has no effect on the realization of $i$'s agent statistic, and hence, $\psi_i^{k+1} = \psi_i^k$. This property is established under the static model in \autoref{static_xistab} and under the dynamic model in the proof of \autoref{dynamicclt}.
  
  We have therefore established that the second term on the right-hand side of \autoref{EB1} equals
  \begin{align*} 
    n^{-1/2}k^2 \E\left[ \lvert \psi_i^{k+1} \psi_j^{k+1} \big( K(\ell_{\bm{A}^{k+1}}(i,j)/h_n) - \right.& K(\ell_{\bm{A}^k}(i,j)/h_n) \big) \rvert \\ & \left. \times \bm{1}\{\norm{\rho_i-\rho_j} \leq b_n\} \bm{1}\{\mathcal{E}_n\} \right].
  \end{align*}

  \noindent This is bounded above by
  \begin{align*} 
    n^{-1/2}k^2 \E\big[ \bm{1}\{\text{agent } k{+}1 \text{ lies on the} & \text{ shortest path connecting } i,j \text{ in } \bm{A}^{k+1}\} \\ & \times \bm{1}\{\ell_{\bm{A}^{k+1}}(i,j) < \infty\} \lvert \psi_i^{k+1} \psi_j^{k+1} \rvert \bm{1}\{\norm{\rho_i-\rho_j} \leq b_n\} \big].
  \end{align*}

  \noindent The reason is that (1) if $k+1$ does not lie on the shortest path, then $\ell_{\bm{A}^{k+1}}(i,j) = \ell_{\bm{A}^k}(i,j)$, and (2) if $\ell_{\bm{A}^{k+1}}(i,j)=\infty$, then $\ell_{\bm{A}^{k+1}}(i,j) = \ell_{\bm{A}^k}(i,j) = \infty$.  Then for $C_n = n^{-1/2}k^2/(k-1)$, the previous equation equals
  \begin{equation}
    C_n \E\left[ \ell_{\bm{A}^{k+1}}(i,j)\, \abs{\psi_i^{k+1} \psi_j^{k+1}} \,\bm{1}\{\ell_{\bm{A}^{k+1}}(i,j) < \infty\} \bm{1}\{\norm{\rho_i-\rho_j} \leq b_n\} \right]. \label{time0}
  \end{equation}

  \noindent Split this into to expectations, one with the term in the expectation multiplied by $\bm{1}\{\ell_{\bm{A}^{k+1}}(i,j) \leq c\norm{\rho_i-\rho_j}\}$ and the other by $\bm{1}\{\ell_{\bm{A}^{k+1}}(i,j) > c\norm{\rho_i-\rho_j}\}$, where $c$ is the constant in \autoref{pd}. Call these two terms $[I]$ and $[II]$, respectively. We have
  \begin{align*}
    [I] &\leq c\,b_n C_n \E\left[\lvert \psi_i^{k+1} \psi_j^{k+1} \rvert \bm{1}\{\norm{\rho_i-\rho_j} \leq b_n\}\right] \\
	&\leq c\, b_n C_n \underbrace{\prob\left( \norm{\rho_i-\rho_j} \leq b_n \right)}_{O(n^{-1} b_n^d)} \underbrace{\sup_{p,p'} \E\left[ \lvert \psi_i^{k+1} \psi_j^{k+1} \rvert \mid \rho_i = \omega_n p, \rho_j = \omega_n p' \right]}_{O(1)}, 
  \end{align*}

  \noindent where the supremum term is $O(1)$ by \autoref{xibd}. Then by \autoref{b_n}, 
  \begin{equation}
    [I] = O(b_n n^{-1/2} n n^{-1} b_n^d) = o(n^{1/4} n^{-1/2} n^{1/16}) = o(1)
    \label{time1}
  \end{equation}

  \noindent On the other hand, by H\"{o}lder's inequality,
  \begin{multline*}
    [II] \equiv C_n \E\big[ \lvert \psi_i^{k+1} \psi_j^{k+1} \rvert \,\ell_{\bm{A}^{k+1}}(i,j) \bm{1}\{c\norm{\rho_i-\rho_j} < \ell_{\bm{A}^{k+1}}(i,j) < \infty, \norm{\rho_i-\rho_j} \leq b_n\} \big] \\
    \leq C_n \E\left[ \ell_{\bm{A}^{k+1}}(i,j)^{4/3} \bm{1}\{c\norm{\rho_i-\rho_j} < \ell_{\bm{A}^{k+1}}(i,j) < \infty, \norm{\rho_i-\rho_j} \leq b_n\} \right]^{3/4} \\ \times \E\left[ \lvert \psi_i^{k+1} \psi_j^{k+1} \rvert^4 \bm{1}\{\norm{\rho_i-\rho_j} \leq b_n\} \right]^{1/4}.
  \end{multline*}

  \noindent Since $\ell_{\bm{A}^{k+1}}(i,j) \leq k+1$, this is bounded above by
  \begin{multline*}
    \underbrace{C_n(k+1)}_{O(n^{3/2})} \underbrace{\prob\left( \norm{\rho_i-\rho_j} \leq b_n \mid c\norm{\rho_i-\rho_j} < \ell_{\bm{A}^{k+1}}(i,j) < \infty \right)^{3/4}}_{O( (n^{-1}b_n^d)^{3/4} )} \\
    \times \prob\left( c\norm{\rho_i-\rho_j} < \ell_{\bm{A}^{k+1}}(i,j) < \infty \right)^{3/4} \underbrace{\prob\left( \norm{\rho_i-\rho_j} \leq b_n \right)^{1/4}}_{O( (n^{-1}b_n^d)^{1/4} )} \\
    \times \underbrace{\sup_{p,p'} \E\left[ \abs{\psi_i^{k+1} \psi_j^{k+1}}^4 \mid \rho_i = \omega_n p, \rho_j = \omega_n p' \right]^{1/4}}_{O(1)},
  \end{multline*}

  \noindent noting that $k = O(n)$. The supremum term is $O(1)$ by \autoref{xibd}. By \autoref{pd}(a), $\prob\left( c\norm{\rho_i-\rho_j} < \ell_{\bm{A}^{k+1}}(i,j) < \infty \right) = O(n^{-3/4})$. Hence, the previous equation is $O(n^{3/2} n^{-9/16} n^{-1}b_n^d) = o(1)$ by \autoref{b_n}. We have therefore established that $[II] = o(1)$, which together with \autoref{time1}, implies $\autoref{time0} = o(1)$. Consequently,  the second term on the right-hand side of \autoref{EB1} is $o(1)$, as desired.

  This completes the modification of step 2 of Lemma G.3. The argument for step 2 of Lemma G.6 is the same.

  {\bf Step (3a).} The arguments in the proof of Theorem G.1 that need to be modified are step 2 of the proof of Lemma G.4 and step 1 of the proof of Lemma G.7 \citep[both in][]{leung2019normal}. We first consider the former. We will need the following definitions.
  \begin{itemize}
    \item Let $\tilde\rho_{-1}, \tilde\rho_0, \tilde\rho_1, \dots$ be i.i.d.\ with density $f$. Define $\rho_i = \omega_n\tilde\rho_i$ for all $i$. Let $N_n \sim \text{Poisson}(n)$ be independent of $\tilde\rho_1, \tilde\rho_2, \dots$ and all other primitives.

  \item Let $\psi_i^N$ be the analog of $\psi_i$ under the {\em Poissonized model}. By this we mean that we compute $i$'s agent statistic under a modified data-generating process that replaces the original the set of positions $(\rho_i)_{i=1}^n$ with $(\rho_i)_{i=1}^{N_n}$. Let $\bm{A}^N$ be the network realized under this Poissonized model.

    \item Let $\psi_i^{N+}$ be the analog of $\psi_i$ under the {\em $i$-specific Poissonized model}, which is almost the same as the Poissonized model in the previous paragraph, except the set of positions is given by $\{\rho_i, \rho_1, \rho_2, \dots, \rho_{i-1}, \rho_{i+1}, \dots, \rho_{N_n}\}$. Let $\theta_N(\rho_i) = \E[\psi_i^{N+} \mid \rho_i]$. For $i > 1$, this is the same as the Poissonized model, but for $i < 1$, it is not (although they produce the same distribution of outcomes).

    \item Let $\psi_i^{N+j}$ be the analog of $\psi_i$ under the {\em $(i,j)$-specific Poissonized model}, which is almost the same as the Poissonized model, except the set of positions is given by $\{\rho_i, \rho_j, \rho_1, \rho_2, \dots, \rho_{i-1}, \rho_{i+1}, \dots, \dots, \rho_{j-1}, \rho_{j+1}, \dots, \rho_{N_n}\}$.
  \end{itemize}
  
  Our goal is to show $\E[n^{-1}\rho_{N_n}] - c_n = o(1)$, where
  \begin{equation*}
    c_n \equiv \E\left[ \frac{1}{n} \sum_{i=1}^{N_n} \sum_{j\neq i, j=1}^{N_n} (\psi_i^N \psi_j^N - \theta_N(\rho_i) \theta_N(\rho_j)) \right],
  \end{equation*}
  
  \noindent and
  \begin{equation*}
    \rho_{N_n} \equiv \sum_{i=1}^{N_n} \sum_{j\neq i, j=1}^{N_n} (\psi_i^N \psi_j^N - \theta_N(\rho_i) \theta_N(\rho_j)) K\left( \ell_{\bm{A}^N}(i,j) / h_n \right) \bm{1}\{\norm{\rho_i-\rho_j} \leq b_n\}.
  \end{equation*}

  \noindent In the terminology of the proof sketch in the proof of \autoref{networkhac}, $c_n$ is the ``target mean.'' Let $\bm{A}^{ij}$ be the network under the $(i,j)$-specific Poissonized model. Decompose 
  \begin{multline*}
    c_n = \E\left[ \frac{1}{n} \sum_{i=1}^{N_n} \sum_{j\neq i, j=1}^{N_n} (\psi_i^N \psi_j^N - \theta_N(\rho_i) \theta_N(\rho_j)) \bm{1}\{\ell_{\bm{A}^N}(i,j) < \infty\} \right] \\ + \E\left[ \frac{1}{n} \sum_{i=1}^{N_n} \sum_{j\neq i, j=1}^{N_n} (\psi_i^N \psi_j^N - \theta_N(\rho_i) \theta_N(\rho_j)) \bm{1}\{\ell_{\bm{A}^N}(i,j) = \infty\} \right].
  \end{multline*}

  \noindent We first show that the second term is zero. Indeed this term equals
  \begin{equation}
    n\, \E\left[ \E\left[ \psi_{-1}^{N+0} \psi_0^{N+(-1)} \mid \{(X_k,\rho_k)\}_{k=-1}^{N_n}, \{\zeta_{k\ell}\}_{k,\ell=-1}^{N_n}, N_n \right] \bm{1}\{\ell_{\bm{A}^{-1,0}}(-1,0)=\infty\} \right] \label{32df0}
  \end{equation}

  \noindent by the Slivnyak-Mecke formula \citep[e.g][Lemma H.2]{leung2019normal} and the fact that agent types are identically distributed. (Recall the definition of the Poissonized models at the start of step (3a) above which introduced agents $-1$ and $0$.) By \autoref{unccom} below,
  \begin{multline}
    \autoref{32df0} = \E\left[ \E\left[ \psi_{-1}^{N+} \mid \{(X_k,\rho_k)\}_{k=-1}^{N_n}, \{\zeta_{k\ell}\}_{k,\ell=-1}^{N_n}, N_n \right] \right. \\ \left. \times \E\left[ \psi_0^{N+} \mid \{(X_k,\rho_k)\}_{k=-1}^{N_n}, \{\zeta_{k\ell}\}_{k,\ell=-1}^{N_n}, N_n \right] \bm{1}\{\ell_{\bm{A}^{-1,0}}(-1,0)=\infty\} \right]. \label{r1d03}
  \end{multline}

  \noindent Let $\bm{A}^0$ denote the network under the $0$-specific Poissonized model. By the network formation model and \autoref{netcexog},
  \begin{align*}
    \E\left[ \psi_0^{N+} \mid \{(X_k,\rho_k)\}_{k=-1}^{N_n}, \{\zeta_{k\ell}\}_{k,\ell=-1}^{N_n}, N_n \right] &= \E\left[ \psi_0^{N+} \mid \{X_k\}_{k=0}^{N_n}, \bm{A}^0, N_n \right] \\
    &= \E\left[ \psi_0^{N+} \mid \{X_k\}_{k=0}^{N_n}, \bm{A}^0, N_n, \rho_0 \right].
  \end{align*}

  \noindent By \autoref{networkstation}, $\E[\psi_0^{N+}]$ equals the second term and hence also the third term. Taking expectations conditional on $\rho_0$ then yields
  \begin{equation}
    \E\left[ \psi_0^{N+} \mid \{(X_k,\rho_k)\}_{k=-1}^{N_n}, \{\zeta_{k\ell}\}_{k,\ell=-1}^{N_n}, N_n \right] = \E\left[ \psi_0^{N+} \right] = \theta_N(\rho_0), \label{netstatpois}
  \end{equation}

  \noindent since $N_n \indep \rho_0$. Using \autoref{netstatpois}, 
  \begin{equation*}
    \autoref{r1d03} = \E\left[ \theta_N(\rho_{-1}) \theta_N(\rho_0) \bm{1}\{\ell_{\bm{A}^{-1,0}}(-1,0)=\infty\} \right].
  \end{equation*}

  \noindent Since agent types are identically distributed, we can replace $-1$ with $i$ and $0$ with $j$ in this expression. We have therefore shown that
  \begin{equation*}
    \autoref{32df0} = \E\left[ \theta_N(\rho_i) \theta_N(\rho_j) \bm{1}\{\ell_{\bm{A}^{ij}}(i,j)=\infty\} \right].
  \end{equation*}

  \noindent Consequently,
  \begin{equation*}
    c_n = \E\left[ \frac{1}{n} \sum_{i=1}^{N_n} \sum_{j\neq i, j=1}^{N_n} (\psi_i^N \psi_j^N - \theta_N(\rho_i) \theta_N(\rho_j)) \bm{1}\{\ell_{\bm{A}^N}(i,j) < \infty\} \right].
  \end{equation*}

  Next, notice that under the event $\ell_{\bm{A}^N}(i,j)=\infty$, we have $K\left( \ell_{\bm{A}^N}(i,j) / h_n \right)=0$ by \autoref{kernel}. Hence, 
  \begin{multline*}
    \E[n^{-1}\rho_{N_n}] = \E\bigg[ \frac{1}{n} \sum_{i=1}^{N_n} \sum_{j\neq i, j=1}^{N_n} (\psi_i^N \psi_j^N - \theta_N(\rho_i) \theta_N(\rho_j)) \\ \times K\left( \ell_{\bm{A}^N}(i,j) / h_n \right) \bm{1}\{\norm{\rho_i-\rho_j} \leq b_n\} \bm{1}\{\ell_{\bm{A}^N}(i,j) < \infty\} \bigg],
  \end{multline*}

  \noindent The previous two equations yield
  \begin{align*}
    \E[n^{-1}&\rho_{N_n}] - c_n \\
    =&\,\, \E\bigg[ \frac{1}{n} \sum_{i=1}^{N_n} \sum_{j\neq i, j=1}^{N_n} (\psi_i^N \psi_j^N - \theta_N(\rho_i) \theta_N(\rho_j)) \\ 
     &\times \big( K\left( \ell_{\bm{A}^N}(i,j) / h_n \right) \bm{1}\{\norm{\rho_i-\rho_j} \leq b_n\} - 1 \big) \bm{1}\{\ell_{\bm{A}^N}(i,j) < \infty\} \bigg] \\
      =&\, n \int_{\R^d} \int_{\R^d} \E\bigg[ \left(\psi_i^{N+j} \psi_j^{N+j} - \theta_N(\rho_i) \theta_N(\rho_j)\right) \big( K\left( \ell_{\bm{A}^{ij}}(i,j) / h_n \right) \bm{1}\{\norm{\rho_i-\rho_j} \leq b_n\} \\
       & - 1 \big) \bm{1}\{\ell_{\bm{A}^{ij}}(i,j) < \infty\} \,\bigg|\, \rho_i=\omega_n p, \rho_j=\omega_n p' \bigg] f(p) f(p') \,\text{d}p \,\text{d}p'
  \end{align*}

  \noindent where the last line follows from the Slivnyak-Mecke formula \citep[e.g][Lemma H.2]{leung2019normal}. By a change of variables $p'' = p+\omega_n(p'-p)$, the last line equals 
  \begin{multline}
    \kappa \int_{\R^d} \int_{\R^d} \E\bigg[ (\psi_i^N \psi_j^N - \theta_N(\omega_n \tilde\rho_i) \theta_N(\omega_n \tilde\rho_j)) \big( K\left( \ell_{\bm{A}^{ij}}(i,j) / h_n \right) \bm{1}\{\norm{p-p''} \leq b_n\} - 1 \big) \\
    \times \bm{1}\{\ell_{\bm{A}^{ij}}(i,j) < \infty\} \,\bigg|\, \tilde\rho_i=p, \tilde\rho_j=p+\omega_n^{-1}(p''-p) \bigg] f(p) f(p+\omega_n^{-1}(p''-p)) \,\text{d}p \,\text{d}p''. \label{vbdj2io}
  \end{multline}

  We next show that the term in the conditional expectation is $o_p(1)$ and then apply dominated convergence. First, by \autoref{xibd}, $(\psi_i^N \psi_j^N - \theta_N(\omega_n \tilde\rho_i) \theta_N(\omega_n \tilde\rho_j)) = O(1)$. Second, $\bm{1}\{\norm{p-p''} \leq b_n\} \rightarrow 1$ by definition of $b_n$ in \autoref{b_n}. Third, we claim $K\left( \ell_{\bm{A}^{ij}}(i,j) / h_n \right) \plimarrow 1$.
  
To prove the third claim, by \autoref{kernel}, it suffices to show that $\ell_{\bm{A}^{ij}}(i,j) / h_n = o_p(1)$ conditional on $\tilde\rho_i=p, \tilde\rho_j=p+\omega_n^{-1}(p''-p)$. Under this conditioning event, $\norm{\rho_i-\rho_j} = \norm{p-p''}$, which is a finite constant. We would then like to apply \autoref{pd}(b) to obtain the desired conclusion; the only hindrance is that the assumption is stated for the model with $n$ rather than $N_n+2$ agents. It thus suffices to establish the assumption under the $(i,j)$-specific Poissonized model.  Recall that $\bm{A}^{ij}$ denotes the network this model. Let $\bm{A}^m$ denote the network under the model with $m$ agents $\{\rho_i\}_{i=1}^m$. Then
  \begin{multline*}
    \prob(\ell_{\bm{A}^{ij}}(i,j) > \epsilon \mid \norm{\rho_i-\rho_j}) \leq \prob(\abs{N_n-n} > n/2 \mid \norm{\rho_i-\rho_j}) \\ + \sum_{m\colon \abs{m-n}\leq n/2} \prob(\ell_{\bm{A}^{m+2}}(i,j) > \epsilon \mid \norm{\rho_i-\rho_j}, N_n=m) \prob(N_n=m).
  \end{multline*}

  \noindent By independence of $N_n$, the first probability equals $\prob(\abs{N_n-n} > n/2) \leq 4/n^2 \var(N_n) = o(1)$. The second probability is also $o(1)$ because, for any sequence $m=m_n$ with $m/n \rightarrow \alpha \in (0,1)$,
  \begin{equation*}
    \prob(\ell_{\bm{A}^{m+2}}(i,j) > \epsilon \mid \norm{\rho_i-\rho_j}, N_n=m) = o(1)
  \end{equation*}

  \noindent given $\norm{\rho_i-\rho_j} = \norm{p-p''}$ by \autoref{pd}(b). Then since $h_n \rightarrow \infty$, $\ell_{\bm{A}}(i,j)/h_n = o_p(1)$, as desired.

  Finally, to apply dominated convergence, we need to show that the integrand in \autoref{vbdj2io} is uniformly bounded by an integrable function of $p,p''$ that does not depend on $n$. By \autoref{kernel}, the integrand is bounded above by
  \begin{equation*}
    K^* \E\left[ \abs{\psi_i^N \psi_j^N - \theta_N(\omega_n\tilde\rho_i) \theta_N(\omega_n\tilde\rho_j)} \mid \tilde\rho_i=p, \tilde\rho_j=p+\omega_n^{-1}(p''-p) \right] f(p) \sup_x f(x).
  \end{equation*}

  \noindent As argued in step 2 of the proof of Lemma G.4 of \cite{leung2019normal} (the lemma being modified here), the right-hand side is uniformly bounded by $c\, \text{exp}\{ c'\norm{p-p''}\} f(p)$ for some $c,c'>0$, which is integrable.

  This completes the modification of step 2 of Lemma G.4. The argument for step 1 of Lemma G.7 is the same.
\end{proof}

%----------------------
\begin{lemma}\label{unccom}
  Under \autoref{netcexog} and \autoref{nocoord}, \autoref{r1d03} holds.
\end{lemma}
\begin{proof}
  We prove this in two steps. First we argue that agent $i$'s statistic only depends on agents within her own component and likewise for $j$. Second, since $i$ and $j$ are not connected, their agent statistics are conditionally independent under \autoref{netcexog}.

  {\bf Step 1.} Let $i,j \in \{-1,0\}$ with $i\neq j$. Let $\RR_i^n$ be the radius of stabilization of agent $i$. For the static model, this is defined in \autoref{RR_i} (here the case is $\X = \{\rho_{-1}, \dots, \rho_{N_n}\}$) and for the dynamic model in \autoref{RR_iT}. Let $\mathcal{N}(i) \subseteq \mathcal{N}_n$ be the subset of agents $k$ whose positions $\tilde \rho_k$ lie in $Q(\tilde\rho_i,\RR_i^n \omega_n^{-1})$, which is the cube centered at $\tilde\rho_i$ with radius $\RR_i^n \omega_n^{-1}$. The radius of stabilization has the property that the realization $\psi_i^{N+j}$ is the same whether under the original $n$-agent model or under the restricted model where the set of agents is only $\mathcal{N}(i)$. This property is formally established under the static model in \autoref{static_xistab} and under the dynamic model in the proof of \autoref{dynamicclt}; both use \autoref{nocoord}. By construction of $\RR_i^n$, we have $\mathcal{N}(i) \subseteq C_i(\bm{A}^{-1,0})$, the component of $i$ in the network $\bm{A}^{-1,0}$. Then under the event $\ell_{\bm{A}^{-1,0}}(-1,0)=\infty$,
  \begin{equation}
    \psi_i^{N+j} = \psi_i^N, \label{v3209d}
  \end{equation}

  \noindent recalling that the latter is $i$'s agent statistic under the $i$-specific Poissonized model (rather than the $(i,j)$-specific model as in the former). 
  
  {\bf Step 2.} As previously established, $\mathcal{N}(i) \subseteq C_i(\bm{A}^{-1,0})$. This implies that, when conditioning on $\mathcal{B} = \{(X_k,\rho_k)\}_{k=-1}^{N_n}, \{\zeta_{k\ell}\}_{k,\ell=-1}^{N_n}, N_n$, the only randomness in $\psi_{-1}^{N+0}$ and $\psi_0^{N+(-1)}$ are due to $\{\varepsilon_k\colon k \in C_i(\bm{A}^{-1,0})\}$. Thus, under the event $\ell_{\bm{A}^{-1,0}}(-1,0)=\infty$, \autoref{netcexog} implies
  \begin{equation*}
    \{\varepsilon_k\colon k \in C_{-1}(\bm{A}^{-1,0})\} \indep \{\varepsilon_\ell\colon \ell \in C_0(\bm{A}^{-1,0})\} \mid \{(X_k,\rho_k)\}_{k=-1}^{N_n}, \{\zeta_{k\ell}\}_{k,\ell=-1}^{N_n}, N_n,
  \end{equation*}

  \noindent from which conditional independence of $\psi_{-1}^N$ and $\psi_0^N$ follows. Then \autoref{v3209d} establishes the result.
\end{proof}

%----------------------
\begin{lemma}\label{2ndterm}
  Under the assumptions of \autoref{networkhac} (see the definitions made in the proof of that theorem),
  \begin{equation*}
    \frac{1}{n} \sum_{i=1}^n \sum_{j\neq i} (\psi_i - \bar{\psi}) (\psi_j - \bar{\psi})' K\left( \ell_{\bm{A}}(i,j) / h_n \right) \bm{1}\{\norm{\rho_i-\rho_j} > b_n\} = o_p(1).
  \end{equation*}
\end{lemma}
\begin{proof}
  The mean of the absolute value of left-hand side is bounded above by
  \begin{multline*}
    n \E[\lvert (\psi_i \psi_j' - \bar{\psi}^2) K\left( \ell_{\bm{A}}(i,j) / h_n \right) \bm{1}\{\norm{\rho_i-\rho_j} > b_n\} \rvert ] \\
    \leq \E[(\psi_i \psi_j' - \bar{\psi}^2)^2]^{1/2} K^* n \prob(\ell_{\bm{A}}(i,j) \leq h_n \medcap \norm{\rho_i-\rho_j} > b_n)^{1/2} \\
    \leq C n \prob\left( \max_{j \in \mathcal{N}_{\bm{A}}(i,h_n)} \norm{\rho_i-\rho_j} > b_n \right)^{1/2} \\
    \leq C n \left( 3h_n(c_1\kappa \bar{f})^{h_n} (b_n/h_n)^{h_n d} e^{-c_2 b_n} \right)^{1/2}.
  \end{multline*}

  \noindent The constant $C>0$ exists by \autoref{psibd}. The last line follows from \autoref{MKN}. Since $b_n$ grows at a polynomial rate by \autoref{b_n} and $h_n$ grows at a sub-polynomial rate, this tends to zero, as desired.
  %\footnote{Note that if $d$ is known, then we can pick $h_n = o(b_n)$, which allows for a polynomial rate of growth in the bandwidth.}
\end{proof}

%----------------------------------------------------------------------
\section{Proof of \autoref{genstat}}\label{pogs}
%----------------------------------------------------------------------

To prove the result, we need to introduce ``Poissonized'' models where the number of agents is random and possibly countably infinite. This is in order to derive the large-$n$ limit of the mean of agent statistics, which we can show does not depend on $p$. 

\bigskip

\noindent {\bf Step 1: Notation} Recall from \autoref{lna} and \autoref{ssmclt} that the (static or dynamic) model is given by the tuple
\begin{equation*}
  (U, \lambda, V, \bm{\rho}, \bm{Z}, \bm{\zeta}).
\end{equation*}

\noindent It will be helpful to think of simulating the data-generating process in the following order. First, draw positions. In this case, they are given by $\omega_n \tilde{\bm{\rho}}$, where $\tilde\rho_1, \dots, \tilde\rho_n$ are i.i.d.\ draws from $f$. Then conditional on the set of positions, draw $\bm{Z}$ and $\bm{\zeta}$. Given these primitives, draw the network $\bm{A}$ according to \autoref{snf}, and finally draw outcomes $\bm{Y}$. We call this model the {\em binomial model}. We write the agent statistic of agent 1 as
\begin{equation*}
  \psi(\omega_n\tilde\rho_1, \omega_n\tilde{\bm{\rho}}, \bm{Z}, \bm{\zeta}, \bm{Y}, \bm{A}),
\end{equation*}

\noindent what we had previously called simply $\psi_1$.

We next define a {\em Poissonized model} where the number of agents is given by $N_n+1$, rather than $n$, where $N_n \sim \text{Poisson}(n)$ is independent of all other primitives. By Lemma 1.5 of \cite{penrose2003}, $\tilde\rho_1, \dots, \tilde\rho_{N_n+1}$ has the same distribution as $\mathcal{P}_{nf} \cup \{\tilde\rho_1\}$, where $\mathcal{P}_{nf}$ is an inhomogeneous Poisson point process on $\R^d$ with intensity $nf(\cdot)$ and independent of $\tilde\rho_1 \sim f$. Thus, define the {\em Poissonized model}
\begin{equation*}
  (U, \lambda, V, \omega_n (\mathcal{P}_{nf} \cup \{\tilde\rho_1\}), \bm{Z}, \bm{\zeta}).
\end{equation*}

\noindent As in the binomial model, we first draw positions $\omega_n (\mathcal{P}_{nf} \cup \{\tilde\rho_1\})$, and then for each element $i$ of this set, we draw $Z_i$'s independently across elements to obtain $\bm{Z}$ and likewise with $\bm{\zeta}$. In fact, under \autoref{rhounif}(a), $Z_i$'s are drawn i.i.d.; note that $\zeta_{ij}$'s are always drawn i.i.d.\ by definition of the network formation model. Then the network and outcomes are realized in the same way, and the agent statistic of the agent positioned at $\omega_n\tilde\rho_1$ is
\begin{equation*}
  \psi(\omega_n\tilde\rho_1, \omega_n (\mathcal{P}_{nf} \cup \{\tilde\rho_1\}), \bm{Z}, \bm{\zeta}, \bm{Y}, \bm{A}).
\end{equation*}

Finally, we define a {\em restricted Poissonized model} that simply replaces $\omega_n (\mathcal{P}_{nf} \cup \{\tilde\rho_1\})$ with the subset $\omega_n (\mathcal{P}_{nf} \cap Q(\tilde\rho_1, R\omega_n^{-1}) \cup \{\tilde\rho_1\})$, where $Q(i,R) \subseteq \R^d$ is the cube centered at $i \in \R^d$ with side length $R$.

We can then define the conditional expectations
\begin{align*}
  \theta_n(p) &= \E\left[ \psi(\omega_n\tilde\rho_1, \omega_n\tilde{\bm{\rho}}, \bm{Z}, \bm{\zeta}, \bm{Y}, \bm{A}) \mid \tilde\rho_1 = p \right], \\
  \theta_N(p) &= \E\left[ \psi(\omega_n p, \omega_n(\mathcal{P}_{nf} \cup \{p\}), \bm{Z}, \bm{\zeta}, \bm{Y}, \bm{A}) \right], \\
  \theta_R(p) &= \E\left[ \psi(\omega_n p, \omega_n (\mathcal{P}_{nf} \cap Q(p, R\omega_n^{-1}) \cup \{p\}), \bm{Z}, \bm{\zeta}, \bm{Y}, \bm{A}) \right].
\end{align*}

\noindent {\bf Step 2: Limit Approximation} By Lemma G.8 of \cite{leung2019normal},
\begin{equation*}
  \sup_p \abs{\theta_n(p) - \theta_N(p)} = o(n^{-1/3}).
\end{equation*}

\noindent \autoref{poisrest} below shows that
\begin{equation}
  \sup_p \abs{\theta_N(p) - \theta_R(p)} = o(n^{-1/3}) \label{poisrestr}
\end{equation}

\noindent for $R = C (\log n)^{1/\epsilon}$ for large enough $C,\epsilon>0$. Thus it suffices to show that $\theta_R(p)$ does not vary with $p$.

\bigskip

\noindent {\bf Step 3: Asymptotic Stationarity} Let $\X \equiv \omega_n (\mathcal{P}_{nf} \cap Q(p, R\omega_n^{-1}) \cup \{p\})$. Observe that positions $\X$ only enter the model directly through differences $\norm{p'-p''}$ for $p',p'' \in \X$ (in particular through the first argument of $V(\cdot)$) and indirectly through $\bm{Z},\bm{\zeta}$. \autoref{rhounif}(a) shuts off the indirect channel, so
\begin{multline*}
  \psi(\omega_n p, \omega_n (\mathcal{P}_{nf} \cap Q(p, R\omega_n^{-1}) \cup \{p\}), \bm{Z}, \bm{\zeta}, \bm{Y}, \bm{A}) \\
  \stackrel{d}= \psi\left(\zero, \omega_n(\mathcal{P}_{nf} - p) \cap Q(\zero, R) \cup \{\zero\}), \bm{Z}, \bm{\zeta}, \bm{Y}, \bm{A}\right).
\end{multline*}

\noindent Since the origin is an interior point of $\text{supp}(f)$, for $n$ large enough, $Q(\zero, R)$ is contained in the set $\omega_n (\text{supp}(f) - p)$. % need to assume $p$ is in the interior of $\text{supp}(f)$, though. but can we ignore boundary points since it's measure zero?
This claim uses the fact that $\omega_n$ diverges faster than $R = C(\log n)^{1/\epsilon}$. Furthermore, because $f$ is uniform, $\omega_n(\mathcal{P}_{nf} - p) \cap Q(\zero, R)$ has the same distribution as a homogeneous Poisson point process on $B(\zero, R)$ with intensity $n$. In particular, this process does not depend on $p$. Therefore, the expectation of the right-hand side of the previous equation does not depend on $p$, as desired. \qed

\begin{lemma}\label{poisrest}
  Under the assumptions of \autoref{genstat}, \autoref{poisrestr} holds for $R = C (\log n)^{1/\epsilon}$ for large enough $C,\epsilon>0$. 
\end{lemma}
\begin{proof}
  Let $\RR_n(\tilde\rho_1)$ be the radius of stabilization (\autoref{rosxi}) of the agent positioned at $\omega_n\tilde\rho_1$ in the Poissonized model. For the static model, this is defined in \autoref{RR_i} for the case $\X = \mathcal{P}_{nf} \cup \{\tilde\rho_1\}$ and for the dynamic model in \autoref{RR_iT}.

  Fix any $R>0$. Trivially,
  \begin{multline*}
    \sup_x n^{1/3} \abs{\theta_N(p)-\theta_R(p)} \leq \sup_p n^{1/3} \E\big[ \lvert \psi(\omega_n p, \omega_n( \mathcal{P}_{nf} \cup \{p\} ), \bm{Z}, \bm{\zeta}, \bm{Y}, \bm{A}) \\
    - \psi(\omega_n p, \omega_n (\mathcal{P}_{nf} \cap Q(p, R\omega_n^{-1}) \cup \{p\}), \bm{Z}, \bm{\zeta}, \bm{Y}, \bm{A}) \rvert \mid \tilde\rho_1 = p\big]
  \end{multline*}

  \noindent By definition of the radius of stabilization, if $R > \RR_n(p)$, then the term in the expectation on the right-hand is identically zero. Hence, the right-hand side equals
  \begin{multline*}
    \sup_p n^{1/3} \E\left[ \lvert \psi(\omega_n p, \omega_n( \mathcal{P}_{nf} \cup \{p\} ), \bm{Z}, \bm{\zeta}, \bm{Y}, \bm{A}) \right. \\
    \left. - \psi(\omega_n p, \omega_n (\mathcal{P}_{nf} \cap Q(p, R\omega_n^{-1}) \cup \{p\}), \bm{Z}, \bm{\zeta}, \bm{Y}, \bm{A}) \rvert \mathbf{1}\{R \leq \RR_n(p)\} \right] \\
    \leq C' n^{1/3} \sup_p \prob(\RR_n(p) \geq R)^{1/2},
  \end{multline*}

  \noindent by the Cauchy-Schwarz inequality, where $C'>0$ does not depend on $n$ and exists by \autoref{psibd}. Under the given assumptions, exponential stabilization (\autoref{bpestab}) holds (established in \autoref{static_bpestab} and the proof of \autoref{dynamicclt} for the static and dynamic models, respectively). Choosing $R = C (\log n)^{1/\epsilon}$ for large enough $C,\epsilon>0$, exponential stabilization implies the last line is $o(1)$, as desired.
\end{proof}

%----------------------------------------------------------------------
\section{Proof of Static CLT}\label{staticproof}
%----------------------------------------------------------------------

%------------------------------------
\subsection{Main Idea}
%------------------------------------

As discussed at the start of \autoref{ssm}, the proof consists of verifying high-level conditions for a general CLT stated in \autoref{master-clt}. The main ones are ``stabilization'' conditions, and their basic idea is discussed at length in \S 4 of \cite{leung2019normal}. Here we give a brief summary.

An agent statistic $\psi_i$ is {\em stabilizing} essentially if we can establish the following two conditions. (i) We can find some $i$-specific subset of agents $J_i$ such that the value of $\psi_i$ does not change if we were to counterfactually remove all agents $\mathcal{N}_n\backslash J_i$ from the game and rerun the model using only agents in $J_i$. (ii) We can show that $|J_i|$ has exponential tails. Note that the index $i$ is arbitrary because agent types are identically distributed.

Consider the simple case where $\psi_i = Y_i$, so the network moment is just the empirical choice probability. Recall the definition of $i$'s strategic neighborhood $C_i^+$ from \autoref{stratneigh}. For this case, we can verify the stabilization requirements by taking $J_i = C_i^+$. For requirement (i), we use property \autoref{snprop} of strategic neighborhoods, which holds under \autoref{S}. It states that, to find the set of $i$'s possible equilibrium outcomes, we only need to look at $\mathcal{E}(A_{C_i^+}, \tau_{C_i^+})$, the set of Nash equilibria in the counterfactual game consisting solely of agents in $C_i^+$, rather than the entire equilibrium set $\mathcal{E}(\bm{A}, \bm{\tau})$. Furthermore, under \autoref{nocoord}, the selection mechanism $\lambda'$ that picks from this set of candidate equilibria is entirely determined by inputs $A_{C_i^+}, \tau_{C_i^+}$. This means that, if we were to remove $\mathcal{N}_n\backslash C_i^+$ from the game, $\lambda'$ still picks the same equilibrium subvector of outcomes for $C_i^+$. Therefore, the realization of $Y_i$ is only a function of agents in $C_i^+$, which establishes requirement (i) of stabilization.

Under \autoref{dfrag}, it can be established that $|J_i| = O_p(1)$. Intuitively, when strategic interactions are properly restricted, the impact of one agent's choice will not affect the choices of too many other agents. Under regularity conditions given in \autoref{scltassumps}, in particular the exponential decay on linking probabilities (\autoref{sparsity}), this can be strengthened to $|J_i|$ having exponential tails, which establishes requirement (ii) of stabilization. The technique is to stochastically bound $|J_i|$ by a certain multi-type branching process (see \autoref{introbp}) and use standard arguments in branching process theory to show that the total offspring count of the process has exponential tails.

%------------------------------------
\subsection{Generalized Setup}\label{sgsstatic}
%------------------------------------

To prove \autoref{staticclt}, we verify the assumptions of \autoref{master-clt}. To conform with the setup of that theorem, we need to change notation in order to generalize the model to accommodate a random number of agents, e.g.\ as in \autoref{networkstation}.

To work toward the required notation, let us first restate the original model \autoref{finitemodel} as the tuple 
\begin{equation}
  (U, \lambda, V, r_n^{-1}\tilde{\bm{\rho}}, \W), \label{finitemodel2}
\end{equation}

\noindent where $U(\cdot)$ is the payoff function, $\lambda(\cdot)$ the selection mechanism, $V(\cdot)$ the latent index in the network formation model, $r_n = \omega_n^{-1}$, $\tilde{\bm{\rho}} = (\tilde\rho_i)_{i=1}^n$ for $\{\tilde\rho_i\}_{i\in\mathbb{N}} \stackrel{iid}\sim f$, and $\W = (\bm{Z}, \bm{\zeta})$. We generalize this by replacing the fourth element of the tuple with any random set whose size is finite a.s. Note that in the sequel we will be using $r_n$, the inverse of $\omega_n$, in order to maintain consistency with the setup of \cite{leung2019normal}.

Let $\X = \{\tilde\rho_i\}_{i=1}^N \cup G$ for some a.s.\ finite random variable $N$ that may depend on $n$ and finite set $G \subseteq \R^d$. We will be interested in cases where $N$ is some deterministic function of $n$ or $N \sim \text{Poisson}(n)$. Fix an arbitrary $x \in \R^d$.  For $r \in \R_+$, let $\tau_{x,r}\colon \R^d \rightarrow \R^d$ be the map $y \mapsto x+r^{-1}(y-x)$. We write $\tau_{x,r}\X$ to mean $\{\tau_{x,r}y\colon y \in \X\}$. In what follows, we will take $\tau_{x,r_n}\X$ as the set of positions. We think of each element of $\tau_{x,r_n}\X$ not only as a position $p$ but also as an agent associated with that position. That is, following the convention in spatial econometrics \citep[e.g.][]{jenish2009central}, we also view $\tau_{x,r_n}\X$ as the set of agent labels, which is well-defined since the elements are a.s.\ unique due to $\tilde\rho_i$ be continuously distributed. 

Next, we define the attributes of agents, previously given by $\bm{Z}$ and $\bm{\zeta}$. These are originally defined as arrays of $n$ and $n(n-1)/2$ elements respectively, but we can think of them as mappings from agent labels $i,j\in\mathcal{N}_n$ to attributes $(Z_i,Z_j,\zeta_{ij})$. Motivated by this second interpretation, we replace these sets with stochastic processes that map positions to attributes. Formally, let $\bm{Z}(\cdot)$ be a stochastic process on $\R^d$ with range $\R^{d_z}$ and $\bm{\zeta}(\cdot,\cdot)$ a stochastic process on $\R^d \times \R^d$ with range $\R^{d_\zeta}$. Suppose these processes each have i.i.d.\ marginals, are mutually independent, and are generated independently of positions $\X$. 

Define the ``attribute process'' $\W\colon (p,p') \mapsto (\bm{Z}(p), \bm{Z}(p'), \bm{\zeta}(p,p'))$. For $x\in\R^d, r \in \R_+$, let $\cc_{x,r}\colon \R^d \rightarrow \R^d$ be the map $y \mapsto r^{-1}(x+r(y-x))$. This inverts the operation of $\tau_{x,r}$ and then scales up the result by $r^{-1}$. Let $\W\cc_{x,r}$ denote the composition of the random function $\W$ and $\cc_{x,r}$, i.e.
\begin{equation*}
  \W\cc_{x,r}(y,y') = \big( \bm{Z}(\cc_{x,r}y), \bm{Z}(\cc_{x,r}y'), \bm{\zeta}(\cc_{x,r}y,\cc_{x,r}y') \big)
\end{equation*}

\noindent for any $y,y' \in \R^d$.

We can finally state the general static model, which is given by the tuple
\begin{equation}
  (U, \lambda, V, \tau_{x,r_n}\X, \W\cc_{x,r_n}). \label{model_static_generic}
\end{equation}

\noindent Here we think of $\W\cc_{x,r_n}$ as taking $\tau_{x,r_n}\X$ as its input to generate attributes conditional on agents' positions. Thus, in the case where $N=n$ and $x = \zero$, this model is equivalent to \autoref{finitemodel2}. Note that when $x=\zero$, $\cc_{x,r_n}$ is simply the identity map; for technical reasons, it is convenient in the proof of \autoref{master-clt} to consider $x \neq \zero$. The purpose of $\cc_{x,r_n}$ is to undo the transformation $\tau_{x,r_n}$ and scale the result up by $r_n^{-1}$. Thus, $\cc_{x,r_n}(\tau_{x,r_n}(\tilde{\bm{\rho}}))$ equals $r_n^{-1}\tilde{\bm{\rho}}$ for any $x$. This is important because under model \autoref{finitemodel2}, we interpret $r_n^{-1}\tilde{\bm{\rho}}$ as the underlying set of positions, so attributes $\W$ need to take this set as its input and not a modification of the set like $\tau_{x,r_n}\tilde{\bm{\rho}}$.

The remainder of this subsection clarifies how the network and equilibrium outcomes are generated under model \autoref{model_static_generic}.  Rather than an $n\times n$ matrix as in model \autoref{finitemodel2}, we view the network $\bm{A} \equiv \bm{A}(\tau_{x,r_n}\X, \W\cc_{x,r_n})$ as a function or random mapping from pairs of agents $p,p' \in \tau_{x,r_n}\X$ to $\{0,1\}$. Denote by $A_{p,p'}$ the potential link between $p, p' \in \tau_{x,r_n}\X$ in $\bm{A}(\tau_{x,r_n}\X, \W\cc_{x,r_n})$. We assume
\begin{equation*}
  A_{p,p'} = \bm{1}\{V(\norm{p-p'}, \W\cc_{x,r_n}(p,p')) > 0\},
\end{equation*}

\noindent which is the analog of \autoref{modelnf}.

Next we define the outcomes $\bm{Y}$ under \autoref{model_static_generic}. Rather than an $n\times 1$ vector of outcomes, we view $\bm{Y} \equiv \bm{Y}(\tau_{x,r_n}\X, \W\cc_{x,r_n})$ as a mapping from $\tau_{x,r_n}\X$ to $\{0,1\}$. Let $Y_p$ denote the outcome of agent $p \in \tau_{x,r_n}\X$ according to $\bm{Y}$ and $Z_p = (p, \bm{Z}(\cc_{x,r_n}p))$, which corresponds in the finite model to $p$'s type $\tau_i$. We assume
\begin{equation*}
  Y_p = \bm{1}\{U(S_p, Z_p) > 0\},
\end{equation*}

\noindent where 
\begin{equation*}
  S_p \equiv S(Y_{-p}, Z_p, Z_{-p}, A_p, A_{-p}).
\end{equation*}

\noindent The first, third, and fourth arguments of the right-hand side are random functions with domain $\tau_{x,r_n}\X\backslash\{p\}$ such that $Y_{-p}\colon p' \mapsto Y_{p'}$, $Z_{-p}\colon p' \mapsto Z_{p'}$, and $A_p\colon p' \mapsto A_{p,p'}$. The last argument is a random function with domain $\tau_{x,r_n}\X\backslash\{p\} \times \tau_{x,r_n}\X\backslash\{p\}$ such that $A_{-p}\colon (p', p'') \mapsto A_{p',p''}$. Thus, $S_p$ simply generalizes the original definition of $S_i$ in \autoref{model}. 

Finally, we define the selection mechanism $\lambda(\cdot)$ in \autoref{model_static_generic}. Let 
\begin{equation*}
  \mathcal{E}(\bm{A}, \tau_{x,r_n}\X, \W\cc_{x,r_n}) \equiv \mathcal{E}(\bm{A}(\tau_{x,r_n}\X, \W\cc_{x,r_n}), \tau_{x,r_n}\X, \W\cc_{x,r_n})
\end{equation*}

\noindent be the set of Nash equilibria, i.e.\ the set of outcome mappings $\bm{Y}$ such that for each $\bm{Y} \in \mathcal{E}(\bm{A}, \tau_{x,r_n}\X, \W\cc_{x,r_n})$, $Y_p$ satisfies \autoref{model} for each $p \in \tau_{x,r_n}\X$. The selection mechanism $\lambda(\cdot)$ is a mapping $(\bm{A}, \tau_{x,r_n}\X, \W\cc_{x,r_n}) \mapsto \mathcal{E}(\bm{A}, \tau_{x,r_n}\X, \W\cc_{x,r_n})$, and the realized outcome satisfies $\bm{Y} = \lambda(\bm{A}, \tau_{x,r_n}\X, \W\cc_{x,r_n})$. This is just a restatement of \autoref{select}(b) for model \autoref{model_static_generic}.

\bigskip

\noindent {\bf Network Moments.} Our goal is to prove a CLT for network moments of the form
\begin{align*}
  \Lambda(\tau_{x,r_n}\X, \W\cc_{x,r_n}) 
  &\equiv \sum_{p\in\tau_{x,r_n}\X} \psi(p, \tau_{x,r_n}\X, \W\cc_{x,r_n}, \bm{Y}(\tau_{x,r_n}\X, \W\cc_{x,r_n}), \bm{A}(\tau_{x,r_n}\X, \W\cc_{x,r_n}))  \\
  &\equiv \sum_{p\in\tau_{x,r_n}\X} \psi(p, \tau_{x,r_n}\X, \W\cc_{x,r_n}, \bm{Y}, \bm{A}), 
\end{align*}

\noindent where the last line abbreviates by suppressing the dependence of $\bm{Y}$ and $\bm{A}$ on their arguments. Note that, by construction, $\bm{Y}$ and $\bm{A}$ are deterministic functions of the first two arguments of $\psi(\cdot)$. Thus we can absorb them into the agent statistic, writing
\begin{equation}
  \xi(p, \tau_{x,r_n}\X, \W\cc_{x,r_n}) \equiv \psi(p, \tau_{x,r_n}\X, \W\cc_{x,r_n}, \bm{Y}, \bm{A}). \label{xistatic}
\end{equation}

\noindent This falls within the setup of \autoref{smaster}.

\bigskip

\noindent {\bf Strategic Neighborhoods.} We need to generalize some notation in \autoref{scltassumps} to the setup of \autoref{model_static_generic}. For $p \in \tau_{x,r_n}\X$, define 
\begin{equation*}
  \mathcal{R}^c_p = \bm{1}\big\{ \inf_s U(s, Z_p) \leq 0 \medcap \sup_s U(s, Z_p) > 0 \big\}. 
\end{equation*}

\noindent For $p,p' \in \tau_{x,r_n}\X$, define
\begin{equation*}
  D_{p,p'} = A_{p,p'} \mathcal{R}^c_{p'},
\end{equation*}

\noindent where $A_{p,p'}$ is the potential link for agents $p,p'$ in the network $\bm{A}(\tau_{x,r_n}\X, \W\cc_{x,r_n})$. Let $\bm{D}$ be the mapping $(p,p') \mapsto D_{p,p'}$, which is a network function on $\tau_{x,r_n}\X$. For any $p \in \tau_{x,r_n}\X$, define $C(p, \tau_{x,r_n}\X, \W\cc_{x,r_n}, \bm{D}) \subseteq \tau_{x,r_n}\X$ as the set of agents (i.e.\ agent positions) in the component of $\bm{D}$ that contains $p$, or for short, the {\em $\bm{D}$-component} of $p$. Define the {\em strategic neighborhood} of $p$ as 
\begin{multline}
  C_p^+ = C(p, \tau_{x,r_n}\X, \W\cc_{x,r_n}, \bm{D}) \medcup \big\{y \in \tau_{x,r_n}\X\colon \exists y' \in C(y, \tau_{x,r_n}\X, \W\cc_{x,r_n}, \bm{D}) \\ \text{ s.t. } A_{y,y'} = 1 \medcap \inf_s U(s, Z_{y'}) > 0 \medcap \sup_s U(s, Z_{y'}) \leq 0\big\}. \label{C^+}
\end{multline}

\noindent Let $\mathcal{C}^+(\tau_{x,r_n}\X,\W\cc_{x,r_n})$ be the set of all strategic neighborhoods.

Finally, we restate \autoref{psibdsimple} under the new setup.

\begin{assump}[Bounded Moments]\label{psibd} \hfill
  \begin{enumerate}[(a)]
    \item There exists $C<\infty$ such that
      \begin{equation*}
	\E\big[ \abs{\psi\left(x, \tau_{x,r_n}(\X_m \cup \{y\}) \cap H_n, \W\cc_{x,r_n} , \bm{Y}, \bm{A}\right)}^8 \big] < C 
      \end{equation*}

      \noindent for any $n$ sufficiently large, $m \in [n/2, 3n/2]$, sequence of subsets of $\R^d$ $\{H_n\}_{n\in\mathbb{N}}$, $x \in\R^d$, $y \in \R^d\cup\{\infty\}$.
      
    \item For any a.s.\ finite set $\X \subseteq \R^d$, $x \in \R^d$, $n \in \mathbb{N}$, 
      \begin{equation*}
	\abs{\psi(x, \tau_{x,r_n}\X, \W\cc_{x,r_n}, \bm{Y}, \bm{A})} \leq c \abs{\X}^c 
      \end{equation*}
      
      \noindent a.s., for some positive constant $c$.
  \end{enumerate}
\end{assump}

%------------------------------------
\subsection{Proofs}\label{sstaticass}
%------------------------------------

\begin{proof}[Proof of \autoref{staticclt}]
  We apply \autoref{master-clt} by verifying Assumptions \ref{xistab}--\ref{xipoly}. Assumptions \ref{xibd} and \ref{xipoly} hold under \autoref{psibdsimple} (or equivalently \autoref{psibd}). The challenging conditions to verify are Assumptions \ref{xistab} and \ref{bpestab}. These hold by Lemmas \ref{static_xistab} and \ref{static_bpestab} that follow next. 
\end{proof}

For an agent $p \in \tau_{x,r_n}\X$, let $\mathcal{N}_{\bm{A}}(p,K)$ be the set of agents in the $K$-neighborhood of $p$ under network $\bm{A}$. Recall that agents are labeled according to their positions, so $\mathcal{N}_{\bm{A}}(p,K) \subseteq \tau_{x,r_n}\X$. Also recall the definition of $C_p^+$ from the previous subsection. Define the {\em relevant set} of agent $p$ as
\begin{equation}
  J_p \equiv J_p(\tau_{x,r_n}\X, \W\cc_{x,r_n}) = \bigcup\left\{ C_y^+\colon y \in \mathcal{N}_{\bm{A}}(p,K) \right\}, \label{J_i}
\end{equation}

\noindent which is also a subset of $\tau_{x,r_n}\X$. The next lemma shows that any $K$-local agent statistic is entirely determined by the positions and attributes of agents in $J_p$. 

\begin{lemma}\label{CJlem}
  Under Assumptions \ref{klocal}, \ref{S}, and \ref{nocoord}, for any $p \in \tau_{x,r_n}\X$,
  \begin{equation}
    \psi(p, \tau_{x,r_n}\X, \W\cc_{x,r_n}, \bm{Y}, \bm{A}) = \psi(p, J_p, \W\cc_{x,r_n}, \bm{Y}, \bm{A}). \label{constructJ}
  \end{equation}
\end{lemma}
\begin{proof}
  Recall from the previous subsection that $\mathcal{E}(\bm{A}, \tau_{x,r_n}\X, \W\cc_{x,r_n})$ is the set of Nash equilibria under model \autoref{model_static_generic}. For any $\bm{Y}$ in that set, recall that $Y_p$ is the outcome for agent $p \in \tau_{x,r_n}\X$ according to $\bm{Y}$, and for $C \subseteq \tau_{x,r_n}\X$, define $Y_C = \{Y_p\colon p \in C\}$. By Lemma B.2 of \cite{leung2019compute}, under \autoref{S}, for any strategic neighborhood $C^+ \in \mathcal{C}^+(\tau_{x,r_n}\X,\W\cc_{x,r_n})$,
  \begin{equation}
    \mathcal{E}(\bm{A}, C^+, \W\cc_{x,r_n}) = \{Y_{C^+}\colon \bm{Y} \in \mathcal{E}(\bm{A}, \tau_{x,r_n}\X,\W\cc_{x,r_n})\}. \label{AC+}
  \end{equation}

  \noindent That is, the set of Nash equilibria in the restricted model with agents given by $C^+$ is equivalent to the set of equilibrium outcomes for agents in $C^+$ in the full model with agents given by $\tau_{x,r_n}\X$. By \autoref{klocal}, $\psi(p, \tau_{x,r_n}\X, \W\cc_{x,r_n}, \bm{Y}, \bm{A})$ functionally depends on its arguments only through the positions, attributes, random-utility shocks, potential links, and equilibrium outcomes of agents in $\mathcal{N}_{\bm{A}}(p,K)$. Consider an arbitrary agent $w \in \mathcal{N}_{\bm{A}}(p,K)$. By \autoref{AC+} and \autoref{nocoord}, the realization of $Y_w$ is determined by the positions, attributes, and random-utility shocks of agents in $C_w^+$. Hence, the value of $\psi(p, \tau_{x,r_n}\X, \W\cc_{x,r_n}, \bm{Y}, \bm{A})$ would remain the same if we were to counterfactually remove all agents not in $\mathcal{S}\equiv \cup_{w \in\mathcal{N}_{\bm{A}}(p,K)} C_w^+$ from the model and draw new equilibrium outcomes for the remaining subgame. That is, \eqref{constructJ} holds for $\mathcal{S}$ in place of $J_p$. But by definition $J_p = \mathcal{S}$, so the result follows.
\end{proof}

%----------------------
For $p \in \tau_{x,r_n}\X$, let
\begin{equation}
  \RR_n(p) \equiv \RR(p, \tau_{x,r_n}\X, \W) = \max_{p'\in J_p}\, \norm{p-p'}. \label{RR_i}
\end{equation}

\noindent be its {\em radius of stabilization}.

\begin{lemma}\label{static_xistab}
  Under the assumptions of \autoref{staticclt}, \autoref{xistab}(a) and (c) hold for each component $\xi$ defined in \eqref{xistatic} for $\RR_\xi$ given in \eqref{RR_i}. 
\end{lemma}
\begin{proof}
  The proof follows the line of argument in Lemma D.2 of \cite{leung2019normal}, which verifies the assumptions for a static model of network formation with strategic interactions. We first prove \autoref{xistab}(c). For any $H \subseteq \R^d$ and $p \in \tau_{x,r_n}\X \cap H$,
  \begin{equation}
    J_p(\tau_{x,r_n}\X \cap H, \W\cc_{x,r_n}) \subseteq J_p(\tau_{x,r_n}\X, \W\cc_{x,r_n}), \label{Jsub} 
  \end{equation}
  
  \noindent since removing agents only shrinks the size of the relevant set. Then \autoref{xistab}(c) is immediate for radius of stabilization \autoref{RR_i}. 

  The remainder of the proof concerns \autoref{xistab}(a). We show that \autoref{RR_i} is bounded in probability under model \autoref{model_static_generic} with 
  \begin{equation*}
    \X = \tilde\X_n \cup \{x',y_n\},
  \end{equation*}
  
  \noindent where $x' \in \R^d$, $y_n = \tau_{x,r_n^{-1}}(y)$ for $y\in\R^d$, $\tilde\X_n = \{\tilde\rho_i\}_{i=1}^N$, and either $N=\nu(n)$ or $N \sim \text{Poisson}(n)$ independent of all other primitives.
  
  {\em Step 1.} We first prove that for any $z \in \tau_{x,r_n}(\tilde\X_n \cup \{x',y_n\})$,
  \begin{equation}
    \abs{J_z(\tau_{x,r_n}(\tilde\X_n \cup\{x',y_n\}), \W\cc_{x,r_n})} = O_p(1).
    \label{J0goal}
  \end{equation}
  
  \noindent By \autoref{remove} below, the relevant set of $z$ is contained in the union of $z$'s relevant set when another agent $z'$ is removed from the model and the relevant set of $z'$ when $z$ is removed from the model. Therefore, for any $z \in \tau_{x,r_n}\tilde\X_n$,
  \begin{multline}
    J_z(\tau_{x,r_n}(\tilde\X_n \cup\{x',y_n\}), \W\cc_{x,r_n}) \subseteq J_z(\tau_{x,r_n}(\tilde\X_n \cup \{y_n\}), \W\cc_{x,r_n}) \\ \medcup J_{\tau_{x,r_n}x'}(\tau_{x,r_n}(\tilde\X_n \cup\{x',y_n\})\backslash\{z\}, \W\cc_{x,r_n}) \\ 
    \subseteq J_{\tau_{x,r_n}y_n}(\tau_{x,r_n}(\tilde\X_n \cup \{y_n\})\backslash\{z\}, \W\cc_{x,r_n}) \medcup J_z(\tau_{x,r_n}\tilde\X_n, \W\cc_{x,r_n}) \\ \medcup J_{\tau_{x,r_n}x'}(\tau_{x,r_n}(\tilde\X_n\cup\{x'\})\backslash\{z\}, \W\cc_{x,r_n}).
    \label{JJ}
  \end{multline}

  \noindent It suffices to show that each of the sets on the right-hand side has asymptotically bounded size. Observe that for any $x'\in \R^d$,
  \begin{equation*}
    J_{\tau_{x,r_n}x'}(\tau_{x,r_n}(\tilde\X_n\cup\{x'\}), \W\cc_{x,r_n}) = J_{x'}(\tau_{x',r_n}(\tilde\X_n\cup\{x'\}), \W\cc_{x',r_n}).
  \end{equation*}
  
  \noindent This follows because positions enter the model either directly through differences, e.g.\ $||\tau_{x,r_n}X-\tau_{x,r_n}X'|| = ||\tau_{y,r_n}X-\tau_{y,r_n}X'||$ for $X,X' \in \tilde\X_n$, or indirectly through the attribute process. Thus, the rescaled model \autoref{model_static_generic} generates the same network outcome regardless of the centering value $x$. Therefore, to show \autoref{J0goal}, it is enough show that
  \begin{equation}
    \abs{J_y(\tau_{x,r_n}\tilde\X_n\cup\{y\}, \W\cc_{x,r_n})} = O_p(1)
    \label{J1goal}
  \end{equation}

  \noindent for any $x,y \in \R^d$. 
  
  {\em Step 2.} The next two steps make use of the branching process notation in \autoref{introbp}. By \autoref{Jsd}, for $n$ sufficiently large, the left-hand side of \autoref{J1goal} is stochastically dominated by the branching process size $\hat{\mathfrak{X}}_{r_n}^K(y,\bm{Z}(\cc_{x,r_n}y))$ defined in \autoref{bpsd}. Hence, we only need to show that the size of this process is $O_p(1)$.
  
  Recall from that section that this process is generated through the following three steps: first generate the $K$-depth branching process $\bm{B}_{r_n}^K(y,\bm{Z}(\cc_{x,r_n}y))$ defined in \autoref{Bxr}; second, using the notation defined prior to \autoref{Jsd} in \autoref{bp}, generate independent branching processes starting at each particle of the fixed-depth process: $\{\mathfrak{X}_{r_n}(x',z')\colon (x',z') \in \bm{B}_{r_n}^K(y,\bm{Z}(\cc_{x,r_n}y))\}$, constructed independently conditional on $\bm{B}_{r_n}^K(y,\bm{Z}(\cc_{x,r_n}y))$; and third, conditional on the set of particles $\Psi_n$ generated from the second step, draw independent $1$-depth branching processes for each such particle $\{\tilde{\mathfrak{X}}_{r_n}^1(x'',z'')\colon (x'',z'') \in \Psi_n\}$, where $\tilde{\mathfrak{X}}_{r_n}^1(x'',z'')$ is again defined in \autoref{introbp}.

  We next show that $\hat{\mathfrak{X}}_{r_n}^K(y,\bm{Z}(\cc_{x,r_n}y))$ is asymptotically bounded. For economy of notation, we make some new definitions that lump together some of the processes introduced in the previous paragraph. Take any $\mathfrak{X}_{r_n}(x',z')$ generated in the second step and add all particles generated by this process to a set $\mathcal{T}_{r_n}(x',z')$. (So $\abs{\mathcal{T}_{r_n}(x',z')} = \mathfrak{X}_{r_n}(x',z')$.) Define
  \begin{equation*}
    \mathfrak{X}^{23}_{r_n}(x',z') = \mathfrak{X}_{r_n}(x',z') + \sum_{(x'',z'') \in \mathcal{T}_{r_n}(x',z')} \tilde{\mathfrak{X}}_{r_n}^1(x'',z''). 
  \end{equation*}

  \noindent That is, we add to $\mathfrak{X}_{r_n}(x',z')$ the size of each 1-depth branching process generated in step 3 that is initialized at a particle in $\mathcal{T}_{r_n}(x',z')$. Then
  \begin{multline*}
    \prob\left(\hat{\mathfrak{X}}_{r_n}^K(y,\bm{Z}(\cc_{x,r_n}y)) > B\right) = \prob\left( \sum_{(x',z') \in \bm{B}_{r_n}^K(y,\bm{Z}(\cc_{x,r_n}y))} \mathfrak{X}^{23}_{r_n}(x',z') > B \right) \\ 
    \leq \prob\left( \sum_{(x',z') \in \bm{B}_{r_n}^K(y,\bm{Z}(\cc_{x,r_n}y))} \mathfrak{X}^{23}_{r_n}(x',z') > B \medcap \abs{\bm{B}_{r_n}^K(y,\bm{Z}(\cc_{x,r_n}y))} \leq B' \right) \\ + \prob\left( \abs{\bm{B}_{r_n}^K(y,\bm{Z}(\cc_{x,r_n}y))} > B' \right).
  \end{multline*}

  \noindent Since $\abs{\bm{B}_{r_n}^K(y,\bm{Z}(\cc_{x,r_n}y))} \stackrel{d}= \tilde{\mathfrak{X}}_{r_n}^K(y,\bm{Z}(\cc_{x,r_n}y))$, it is $O_p(1)$ by \autoref{MKexptail}. Then for any $\varepsilon > 0$, we can choose $B'$ large enough such that
  \begin{equation*}
    \limsup_{n\rightarrow\infty} \prob\left( \abs{\bm{B}_{r_n}^K(y,\bm{Z}(\cc_{x,r_n}y))} > B' \right) < \varepsilon/2. 
  \end{equation*} 

  \noindent On the other hand, by the union bound,
  \begin{multline}
    \prob\bigg( \sum_{(x',z') \in \bm{B}_{r_n}^K(y,\bm{Z}(\cc_{x,r_n}y))} \mathfrak{X}^{23}_{r_n}(x',z') > B \medcap \abs{\bm{B}_{r_n}^K(y,\bm{Z}(\cc_{x,r_n}y))} \leq B' \bigg) \\
    = \E\bigg[ \prob\bigg( \sum_{(x',z') \in \bm{B}_{r_n}^K(y,\bm{Z}(\cc_{x,r_n}y))} \mathfrak{X}^{23}_{r_n}(x',z') > B \,\bigg|\, \bm{B}_{r_n}^K(y,\bm{Z}(\cc_{x,r_n}y)) \bigg) \\ \times \bm{1}\{\abs{\bm{B}_{r_n}^K(y,\bm{Z}(\cc_{x,r_n}y))} \leq B'\} \bigg] \\
    \leq \E\bigg[ \sum_{(x',z') \in \bm{B}_{r_n}^K(y,\bm{Z}(\cc_{x,r_n}y))} \prob\left( \mathfrak{X}^{23}_{r_n}(x',z') > B/B' \mid \bm{B}_{r_n}^K(y,\bm{Z}(\cc_{x,r_n}y)) \right) \\ \times \bm{1}\{\abs{\bm{B}_{r_n}^K(y,\bm{Z}(\cc_{x,r_n}y))} \leq B'\} \bigg]. \label{r22390b}
  \end{multline}

  \noindent Given the way in which $\mathfrak{X}^{23}_{r_n}(x',z')$ is generated above, we have 
  \begin{equation*}
    \mathfrak{X}^{23}_{r_n}(x',z') \indep \bm{B}_{r_n}^K(y,\bm{Z}(\cc_{x,r_n}y)) \mid (x',z').
  \end{equation*}
  
  \noindent Hence \autoref{r22390b} is bounded above by
  \begin{equation}
    B' \E\left[ \max_{(x',z') \in \bm{B}_{r_n}^K(y,\bm{Z}(\cc_{x,r_n}y))} \prob\left( \mathfrak{X}^{23}_{r_n}(x',z') > B/B' \mid x',z' \right) \bm{1}\{\abs{\bm{B}_{r_n}^K(y,\bm{Z}(\cc_{x,r_n}y))} \leq B'\} \right]. \label{beweo}
  \end{equation}

  For $n$ sufficiently large, $r_n \leq r'$ for some $r' \in (0,\kappa]$. Thus, $\abs{\bm{B}_{r_n}^K(y,\bm{Z}(\cc_{x,r_n}y))} \stackrel{d}= \tilde{\mathfrak{X}}_{r_n}^K(y,\bm{Z}(\cc_{x,r_n}y))$ is stochastically dominated by $\tilde{\mathfrak{X}}_{r'}^K(y,\bm{Z}(\cc_{x,r_n}y))$, by inspection of the intensity measure \autoref{Mintens}, since a higher value $r'$ only changes the process by increasing the expected number of generated offspring. Hence, \autoref{beweo} is bounded above by
  \begin{equation}
    \E\left[ \max_{(x',z') \in \bm{B}_{r'}^K(y,\bm{Z}(\cc_{x,r_n}y))} \prob\left( \mathfrak{X}^{23}_{r_n}(x',z') > B/B' \mid x',z' \right) \bm{1}\{\abs{\bm{B}_{r'}^K(y,\bm{Z}(\cc_{x,r_n}y))} \leq B'\} \right]. \label{beweo2}
  \end{equation}
  
  \noindent Let $\mathcal{S} = \{(x,z)\colon x \in \R^d, z \in \text{supp}(\Phi(\cdot \mid x))\}$. Assume for the moment that
  \begin{equation}
    \lim_{B\rightarrow\infty} \limsup_{n\rightarrow\infty} \sup_{(y,z)\in\mathcal{S}} \prob\left( \mathfrak{X}^{23}_{r_n}(y,z) > B \right) = 0  \label{X23}
  \end{equation}

  \noindent for any $x' \in \R^d$ and $r'$ sufficiently small. (We will prove this claim in step 3 below.) Then 
  \begin{multline*}
    \lim_{B\rightarrow\infty} \limsup_{n\rightarrow\infty} \max_{(x',z') \in \bm{B}_{r'}^K(y,\bm{Z}(\cc_{x,r_n}y))} \prob\left( \mathfrak{X}^{23}_{r_n}(x',z') > B/B' \mid x',z' \right)  \\ \times \bm{1}\{\abs{\bm{B}_{r'}^K(y,\bm{Z}(\cc_{x,r_n}y))} \leq B'\} \stackrel{a.s.}= 0.
  \end{multline*}

  \noindent Hence, by the bounded convergence theorem, for any $\varepsilon, B' > 0$, we can choose $B$ large enough such that
  \begin{equation*}
    \limsup_{n\rightarrow\infty} \autoref{beweo} \leq \limsup_{n\rightarrow\infty} \autoref{beweo2} < \varepsilon/2. 
  \end{equation*}

  \noindent We have therefore shown that for any $\varepsilon > 0$, we can choose $B$ large enough such that
  \begin{equation}
    \limsup_{n\rightarrow\infty} \prob\left( \abs{J_y(\tau_{x,r_n}\tilde\X_n\cup\{y\}, \W)} > B \right) < \varepsilon, \label{s1g}
  \end{equation}

  \noindent which is \autoref{J1goal}, as desired.
  
  {\em Step 3.} We prove \autoref{X23}. As in step 1,
  \begin{multline*}
    \prob\left(\mathfrak{X}^{23}_{r_n}(y,z)) > B\right) = \prob\left( \sum_{(x',z') \in \mathcal{T}_{r_n}(y,z)} \tilde{\mathfrak{X}}_{r_n}^1(x',z') > B \right) \\  
    \leq \prob\left( \sum_{(x',z') \in \mathcal{T}_{r_n}(y,z)} \tilde{\mathfrak{X}}_{r_n}^1(x',z') > B \medcap \mathfrak{X}_{r_n}(y,z) \leq B' \right) \\ + \prob\left( \mathfrak{X}_{r_n}(y,z)) > B' \right).
  \end{multline*}

  \noindent By \autoref{Dexptail}, for any $\varepsilon > 0$, we can choose $B'$ large enough such that
  \begin{equation}
    \limsup_{n\rightarrow\infty} \sup_{(y,z) \in \mathcal{S}} \prob\left( \mathfrak{X}_{r_n}(y,z) > B' \right) < \varepsilon/2. \label{TB'}
  \end{equation}

  \noindent On the other hand, by the union bound,
  \begin{multline*}
    \prob\bigg( \sum_{(x',z') \in \mathcal{T}_{r_n}(y,z)} \tilde{\mathfrak{X}}_{r_n}(x',z';1) > B \medcap \mathfrak{X}_{r_n}(y,z) \leq B' \bigg) \\
    = \E\bigg[ \prob\bigg( \sum_{(x',z') \in \mathcal{T}_{r_n}(y,z)} \tilde{\mathfrak{X}}_{r_n}(x',z';1) > B \,\bigg|\, \mathcal{T}_{r_n}(y,z) \bigg) \\ \times \bm{1}\{\abs{\mathcal{T}_{r_n}(y,z)} \leq B'\} \bigg] \\
    \leq \E\bigg[ \sum_{(x',z') \in \mathcal{T}_{r_n}(y,z)} \prob\left( \tilde{\mathfrak{X}}_{r_n}(x',z';1) > B/B' \mid \mathcal{T}_{r_n}(y,z) \right) \\ \times \bm{1}\{\abs{\mathcal{T}_{r_n}(y,z)} \leq B'\} \bigg].
  \end{multline*}

  \noindent Since $\tilde{\mathfrak{X}}_{r_n}^1(x',z') \indep \mathcal{T}_{r_n}(y,z) \mid (x',z')$, the right-hand side is bounded above by
  \begin{equation}
    B' \E\left[ \max_{(x',z') \in \mathcal{T}_{r_n}(y,z)} \prob\left( \tilde{\mathfrak{X}}_{r_n}^1(x',z') > B/B' \mid x',z' \right) \bm{1}\{\abs{\mathcal{T}_{r_n}(y,z)} \leq B'\} \right]. \label{beweo3}
  \end{equation}

  \noindent Note that for $n$ sufficiently large, $\abs{\mathcal{T}_{r_n}(y,z)} \stackrel{d}= \mathfrak{X}_{r_n}(y,z)$ is stochastically dominated by $\mathfrak{X}_{r'}(y,z)$ for some $r' \in (0,\kappa]$ by inspection of the intensity measure \autoref{Dintens}. Hence,
  \begin{equation}
    \autoref{beweo3} \leq B' \E\left[ \max_{(x',z') \in \mathcal{T}_{r'}(y,z)} \prob\left( \tilde{\mathfrak{X}}_{r_n}^1(x',z') > B/B' \mid x',z' \right) \bm{1}\{\abs{\mathcal{T}_{r'}(y,z)} \leq B'\} \right]. \label{beweo4}
  \end{equation}
  
  \noindent By \autoref{MKexptail}, 
  \begin{equation*}
    \lim_{B\rightarrow\infty} \limsup_{n\rightarrow\infty} \sup_{(x',z') \in \mathcal{S}} \prob\left( \tilde{\mathfrak{X}}_{r_n}^1(x',z') > B/B' \right) = 0. 
  \end{equation*}

  \noindent Hence, by the bounded convergence theorem, for any $\varepsilon, B' > 0$, we can choose $B$ large enough such that
  \begin{equation*}
    \limsup_{n\rightarrow\infty} \sup_{(y,z) \in \mathcal{S}} \autoref{beweo3} \leq \limsup_{n\rightarrow\infty} \sup_{(y,z) \in \mathcal{S}} \autoref{beweo4} < \varepsilon/2. 
  \end{equation*}

  \noindent Combined with \autoref{TB'}, this proves \autoref{X23}.
  
  {\em Step 4.} Having shown asymptotic boundedness of the relevant set \autoref{J1goal}, we can finally establish that the induced radius of stabilization for $p \in \tau_{x,r_n}(\tilde{\X}_n \cup \{x',y_n\})$ satisfies
  \begin{equation}
    \RR(p, \tau_{x,r_n}(\tilde\X_n\cup\{x',y\}), \W\cc_{x,r_n}) = O_p(1).
    \label{RRi_goal}
  \end{equation}

  \noindent In the remainder of the proof, we abbreviate the relevant set of $p$ as
  \begin{equation*}
    J_p = J_p(\tau_{x,r_n}(\tilde\X_n\cup\{x',y_n\}), \W\cc_{x,r_n}). 
  \end{equation*}

  By the law of total probability, 
  \begin{equation}
    \prob\left( \max_{y \in J_p} \norm{p-y} > C \right) \leq \prob\left( \abs{J_p} > B \right) + \prob\left( \max_{y\in J_p} \norm{p-y} > C \medcap \abs{J_p} \leq B\right). \label{ltpJi}
  \end{equation}

  \noindent By steps 1--3, \autoref{J0goal} holds, so for $\varepsilon > 0$, we can choose $B$ such that
  \begin{equation}
    \limsup_{n\rightarrow\infty} \prob\left( \abs{J_p} > B \right) < \varepsilon/2. \label{s1.5g}
  \end{equation}
  
  Consider the second probability on the right-hand side of \autoref{ltpJi}. Under the event $\abs{J_p}\leq B$, it follows that $J_p \subseteq \mathcal{N}_{\bm{A}}(p, B)$. This is because $J_p$ contains agents indirectly connected to $p$ through either $\bm{D}$ or $\bm{A}$. Since $\bm{D}$ is a subnetwork of $\bm{A}$, if $J_p$ has at most $B$ agents, then every agent in the set must be no more than path distance $B$ away from $p$ in $\bm{A}$. We therefore have
  \begin{equation*}
    \prob\left( \max_{y\in J_p} \norm{p-y} > C \medcap \abs{J_p} \leq B\right) \leq \prob\left( \max_{y\in \mathcal{N}_{\bm{A}}(p, B)} \norm{p-y} > C \right). 
  \end{equation*} 

  \noindent As noted in the previous steps, the size of $\mathcal{N}_{\bm{A}}(p, B)$ is stochastically dominated by that of the fixed-depth branching process $\bm{B}_{r_n}^B(p,\bm{Z}(\cc_{x,r_n}p))$ (see \autoref{Jsd}). Thus, the right-hand side of the previous equation is bounded above by
  \begin{equation*}
    \prob\left( \max_{(y,z) \in \bm{B}_{r_n}^B(p,\bm{Z}(\cc_{x,r_n}p))} \norm{p-y} > C \right)
    \leq \prob\left( \max_{(y,z) \in \bm{B}_{\kappa}^B(p,\bm{Z}(\cc_{x,r_n}p))} \norm{p-y} > C \right). 
  \end{equation*} 

  \noindent By \autoref{Mdist}, the last line converges to zero as $C \rightarrow\infty$. Therefore, for any $\varepsilon > 0$ and $B > 0$, we can choose $C$ large enough such that
  \begin{equation}
    \limsup_{n\rightarrow\infty} \prob\left( \max_{y\in J_p} \norm{p-y} > C \medcap \abs{J_p} \leq B\right) < \varepsilon/2. \label{s2g}
  \end{equation}

  Combining \autoref{ltpJi}, \autoref{s1.5g}, and \autoref{s2g}, we have 
  \begin{equation*}
    \lim_{C\rightarrow\infty} \limsup_{n\rightarrow\infty} \prob\left( \max_{y\in J_p} \norm{p-y} > C \right) = 0, 
  \end{equation*}

  \noindent which establishes \autoref{RRi_goal}.
\end{proof}

%----------------------
\begin{lemma}\label{static_bpestab}
  Under the assumptions of \autoref{staticclt}, Assumptions \ref{xistab}(b) and \ref{bpestab} hold for each component $\xi$ defined in \eqref{xistatic} for $\RR_\xi$ given in \eqref{RR_i}.
\end{lemma}
\begin{proof}
  We first prove \autoref{bpestab}. Consider model \autoref{model_static_generic} with $\X = \{\tilde\rho_i\}_{i=1}^N \cup \{x,\tau_{x,r_n^{-1}}y\}$, where either $N=m-1$ with $m \in [n/2,3n/2]$ for all $n$ or $N \sim \text{Poisson}(n)$ independent of all other primitives. We need to show that under this model, $\RR_n(x)$, defined in \autoref{RR_i}, has exponential tails, uniformly over $n>\tilde{n}$, $m \in [n/2,3n/2]$, and $x,y\in\R^d$. By \autoref{DC}, this holds if $\abs{J_x(\tau_{x,r_n}\X, \W\cc_{x,r_n})}$ has uniformly exponential tails. This property follows from \autoref{JCN}.

  We next prove \autoref{xistab}(b). Consider model \autoref{model_static_generic} with $\X = \{\tilde\rho_i\}_{i=1}^N \cup \{x\}$. In \autoref{RR_i}, we construct a radius of stabilization using the relevant set $J_x$. Recall from \eqref{J_i} that this set is defined as the union of the strategic neighborhoods of agents in $x$'s $K$-neighborhood in the network $\bm{A}$. It delineates the set of agents such that their removal from the model could potentially change $x$'s agent statistic. The radius of external stabilization $\RR_n^e(x)$ in \autoref{exstab} is sort of the opposite; it delineates the set of agents whose agent statistics could change due to the removal of $x$ from the model. We therefore build the relevant set in ``reverse'' to construct $\RR_n^e(x)$. Then we show it is asymptotically bounded. This follows the line of argument used in Lemma D.3 of \cite{leung2019normal} to verify the assumptions for a static model of network formation with strategic interactions.
  
  We define a relevant set $J_x^e$ for external stabilization as follows. First initialize $J_x^e$ as $\mathcal{N}_{\bm{A}}(x, 1)$. For each $x'$ in this set, add to $J_x^e$ the set $C(x', \tau_{x,r_n}\X, \W\cc_{x,r_n}, \bm{D})$. Finally, for each $x'' \in C(x', \tau_{x,r_n}\X, \W\cc_{x,r_n}, \bm{D})$, add to $J_x^e$ the agents in $\mathcal{N}_{\bm{A}}(x'', K)$. By construction,
  \begin{equation}
    y \in J_x^e \quad\text{implies}\quad x \in J_y. \label{J>J}
  \end{equation}
  
  Now, recall that $\RR_n(y)$ is the radius of stabilization for $y$ given in \autoref{RR_i}. Let $Q(x,r)$ be the cube in $\R^d$ centered at $x$ with side length $r$. Define the radius of external stabilization
  \begin{equation*}
    \RR_n^e(x) = \max\left\{ \norm{x-y}\colon y \in \bigcup_{y'\in J_x^e} Q(y',\RR_n(y')) \right\}.  
  \end{equation*}

  \noindent By \autoref{J>J}, $Q(x,\RR_n^e(x))$ contains $Q(y,\RR_n(y))$ for all $y\in\tau_{x,r_n}\X$ such that $x \in J_y$. Hence, if
  \begin{multline*}
    \psi\big(y, \tau_{x,r_n}(\tilde\X_n\cup\{x\}) \cap Q(y,\RR_n(y)), \W\cc_{x,r_n}, \bm{Y}, \bm{A} \big) \\ \neq \psi\big(y, \tau_{x,r_n}\tilde\X_n \cap Q(y,\RR_n(y)), \W\cc_{x,r_n}, \bm{Y}, \bm{A} \big) \quad\text{a.s.},
  \end{multline*}
  
  \noindent (which implies $x \in J_y$), this in turn implies that $Q(x,\RR_n^e(x))$ contains $Q(y,\RR_n(y))$. Then the only requirements of \autoref{exstab} left to show are (a) and (c). Requirement (c), that the radius of external stabilization is increasing, follows trivially from the fact that removing agents from the network can only shrink the sizes of $J_x^e$ and $J_{y'}$ for any $y'$.
  
  The remainder of the proof establishes requirement (a) of \autoref{exstab}, i.e.\ $\RR_n^e(x) = O_p(1)$ uniformly over $x\in\R^d$. Under the event $\abs{J_x^e}\leq B$, note that $J_x^e \subseteq \mathcal{N}_{\bm{A}}(x, B)$. Hence,
  \begin{equation}
    \prob(\RR_n^e(x) > C) \leq \prob(\abs{J_x^e} > B) + \prob( \mathcal{E}_n ), \label{RRn>C}
  \end{equation}

  \noindent where
  \begin{equation*}
    \mathcal{E}_n = \left\{ \abs{J_x^e} \leq B \medcap \max\bigg\{ \norm{x-y}; y \in \medcup_{y'\in \mathcal{N}_{\bm{A}}(x, B)} Q\big(y', \max_{y'' \in J_{y'}} \norm{y'-y''}\big) > C \bigg\} \right\}
  \end{equation*}

  \noindent Consider $\prob(\mathcal{E}_n)$. By the triangle inequality,
  \begin{multline*}
    \max\left\{ \norm{x-y}; y \in \medcup_{y'\in \mathcal{N}_{\bm{A}}(x, B)} Q\big(y', \max_{y'' \in J_{y'}} \norm{y'-y''}\big) \right\} \\ 
    \leq \max\left\{ 2\norm{x-y'} + \norm{x-y}; y' \in \mathcal{N}_{\bm{A}}(x, B), y\in J_{y'} \right\} \\
    \leq 2 \max\big\{ \norm{x-y'}; y'\in \mathcal{N}_{\bm{A}}(x, B) \big\} + \max\big\{ \norm{x-y}; y\in J_{y'}, y' \in \mathcal{N}_{\bm{A}}(x, B) \big\}.
  \end{multline*}

  \noindent Define $\bm{J}_x = \medcup_{y'\in \mathcal{N}_{\bm{A}}(x, B)} J_{y'}$. Under the event that $\bm{J}_x$ has size at most $B'>B$, we have trivially $\mathcal{N}_{\bm{A}}(x, B) \subseteq \mathcal{N}_{\bm{A}}(x, B')$.  Furthermore, $\bm{J}_x \subseteq \mathcal{N}_{\bm{A}}(x, B')$. Therefore, $\prob(\mathcal{E}_n)$ is bounded above by
  \begin{equation}
    \prob(\bm{J}_x > B') + 2\prob\left( \max_{y \in \mathcal{N}_{\bm{A}}(x, B')} 2\norm{x-y} > C \right). \label{2sg}
  \end{equation}

  \noindent By \autoref{MKN}, for any $\varepsilon,B'>0$, there exists $C>0$ such that
  \begin{equation}
    \limsup_{n\rightarrow\infty} \sup_x \prob\left( \max_{y \in \mathcal{N}_{\bm{A}}(x, B')} 2\norm{x-y} > C \right) < \varepsilon/4. \label{3sg}
  \end{equation}

  Turning to the first part of \autoref{2sg}, note that
  \begin{equation*}
    \bm{J}_x = \bigcup_{y'\in \mathcal{N}_{\bm{A}}(x, B+K)} C_{y'}^+. 
  \end{equation*}

  \noindent This is just the relevant set $J_x$ with $K$ replaced with $B+K$. Then by the proof of \autoref{static_xistab}, for any $\varepsilon,B > 0$ and $B'$ large enough,
  \begin{equation}
    \limsup_{n\rightarrow\infty} \sup_x \prob(\bm{J}_x > B') < \varepsilon/4. \label{4sg}
  \end{equation}

  Combining \autoref{2sg}, \autoref{3sg}, and \autoref{4sg}, we have
  \begin{equation}
    \limsup_{n\rightarrow\infty} \sup_x \prob(\mathcal{E}_n) < \varepsilon/2.
    \label{5sg}
  \end{equation}
  
  To complete the proof, in light of \autoref{RRn>C} and \autoref{5sg}, it suffices to show that for $B$ large enough,
  \begin{equation}
    \limsup_{n\rightarrow\infty} \sup_x \prob(\abs{J_x^e} > B) < \varepsilon/2.
    \label{1sg}
  \end{equation}

  \noindent This step of the argument is essentially the same as the proof of \autoref{static_xistab}, so we only sketch the argument. First, we construct a branching process whose size stochastically dominates $\abs{J_x^e}$. This is done by replacing each $\mathcal{N}_{\bm{A}}(x',K)$ with $\bm{B}_{r_n}^K(x',\bm{Z}(\cc_{x,r_n}x'))$ defined in \autoref{Bxr} and replacing each $C(x',\tau_{x,r_n}\X, \W\cc_{x,r_n}, \bm{D})$ with the set of particles $\mathcal{T}_{r_n}(x',\bm{Z}(\cc_{x,r_n}x'))$ defined in step 2 of the proof of \autoref{static_xistab}. The sizes of the original sets are dominated by the sizes of their replacements by the proof of \autoref{Jsd}. Second, the sizes of the replacements have exponential tails uniformly over $x'$ by Lemmas \ref{Dexptail} and \ref{MKexptail}. Then \autoref{1sg} follows.
\end{proof}

%----------------------
The following lemma is used in \autoref{static_xistab} to verify internal stabilization. It follows the notation of part (a) of that proof.

\begin{lemma}\label{remove}
  For any $z,z' \in \tau_{x,r_n}\tilde\X_n$,
  \begin{multline*}
    \underbrace{J_z(\tau_{x,r_n}(\tilde\X_n \cup\{x',y_n\}), \W\cc_{x,r_n})}_{J_z(\X)} \subseteq \underbrace{J_z(\tau_{x,r_n}(\tilde\X_n \cup \{x',y_n\})\backslash\{z'\}, \W\cc_{x,r_n})}_{J_z(\X-z')} \\ \medcup \underbrace{J_{z'}(\tau_{x,r_n}(\tilde\X_n \cup\{x',y_n\})\backslash\{z\}, \W\cc_{x,r_n})}_{J_{z'}(\X-z)}.
  \end{multline*}
\end{lemma}
\begin{proof}
  Let $\bm{A}^*$ be defined as the network $\bm{A}$ with agent $z$ removed.  Suppose $y \in J_z(\X)$. If $y = z'$, then clearly $y \in J_{z'}(\X-z)$, so suppose $y \neq z'$. Then either (a) $y \in \mathcal{N}_{\bm{A}}(z,K)$ or (b) $y \in C_w^+$ for some $w \in \mathcal{N}_{\bm{A}}(z,K)$. Consider case (a). If $y \not\in J_z(\X-z')$, then it means that $z'$ lies on the path of length $K$ in $\bm{A}$ connecting $y$ and $z$. This implies $y \in \mathcal{N}_{\bm{A}^*}(z',K)$. Therefore $y \in J_{z'}(\X-z)$. Consider case (b). If $y \not\in J_z(\X-z')$, then there are two possibilities. The first is that $z'$ lies on some directed path in $\bm{D}$ connecting $y$ and some agent $w$, and $z$ is not on this path.\footnote{By directed path we mean a sequence of directed links $A_{y,a_1}, A_{a_1,a_2}, \dots, A_{a_{k-1},a_k}, A_{a_k,w}$.} The second is that $A_{z'y} \mathcal{R}_y^c = 1$. In the second case, clearly $y \in \mathcal{N}_{\bm{A}^*}(z',K)$. In the first case, there is a directed path from $z'$ to $y$ in $\bm{D}^*$. Thus, in either case, $y \in J_{z'}(\X-z)$.
\end{proof}

%----------------------------------------------------------------------
\section{Proof of Dynamic CLT}\label{dynamicproof}
%----------------------------------------------------------------------

The proof technique is the same as that discussed in \autoref{staticproof}. We only need to define the relevant set $J_i$ differently.

%------------------------------------
\subsection{Generalized Setup}\label{sgsdynamic}
%------------------------------------

As in \autoref{sgsstatic}, in order to apply \autoref{master-clt}, we need to generalize our setup to allow for a random number of agents. The generalization is essentially the same as the static model \autoref{model_static_generic}, except we need to define attribute processes for each time period.

For $x \in \R^d$, define $\X, r_n, \tau_{x,r}, \cc_{x,r}$ as in \autoref{sgsstatic}. In place of $\bm{Z}(\cdot)$ and $\bm{\zeta}(\cdot,\cdot)$ given in that section, we define $\bm{Z}(\cdot) = (\bm{Z}^t(\cdot))_{t=0}^T$, where $\bm{Z}^t(\cdot)$ is a stochastic process on $\R^d$ with range $\R^{d_z}$ and $\bm{\zeta}(\cdot,\cdot) = (\bm{\zeta}^t(\cdot,\cdot))_{t=0}^T$, where $\bm{\zeta}^t(\cdot,\cdot)$ is a stochastic process on $\R^d \times \R^d$ with range $\R^{d_\zeta}$. Suppose the marginals of $\bm{Z}(\cdot)$ and $\bm{\zeta}(\cdot,\cdot)$ are i.i.d., $\{\bm{\zeta}^t(\cdot,\cdot)\}_{t=0}^T$ are independent, and $\bm{Z}, \indep \bm{\zeta} \indep \X$. Define the ``attribute process'' $\W(\cdot,\cdot) = (\W^t(\cdot,\cdot))_{t=0}^T$ for $\W^t\colon (p,p') \mapsto (\bm{Z}^t(p), \bm{Z}^t(p'), \bm{\zeta}^t(p,p'))$. Let $\W^t\cc_{x,r}$ denote the composition of the random function $\W^t$ and $\cc_{x,r}$, i.e.
\begin{equation*}
  \W^t\cc_{x,r}(y,y') = \big( \bm{Z^}t(\cc_{x,r}y), \bm{Z}^t(\cc_{x,r}y'), \bm{\zeta}^t(\cc_{x,r}y,\cc_{x,r}y') \big)
\end{equation*}

\noindent for any $y,y' \in \R^d$. Let $\W\cc_{x,r} = (\W^t\cc_{x,r})_{t=0}^T$.

We can finally state the generalized dynamic model, which is given by the tuple
\begin{equation}
  (U, U_0, \lambda, V, \tau_{x,r_n}\X, \W\cc_{x,r_n}), \label{model_dynamic_generic}
\end{equation}

\noindent Compared to \autoref{model_static_generic}, the only difference is the addition of $U_0(\cdot)$, which determines the initial condition, and the redefinition of $\W$.

The remainder of this subsection clarifies how the network and equilibrium outcomes are generated under this model. Define the network $\bm{A} \equiv \bm{A}(\tau_{x,r_n}\X, \W\cc_{x,r_n})$ in the same way as \autoref{sgsstatic}; it is a random mapping $(p,p') \mapsto A_{p,p'}$, where $A_{p,p'}$ denotes the potential link between $p,p' \in \tau_{x,r_n}\X$. For period-$t$ outcomes $\bm{Y}^t$ under \autoref{model_dynamic_generic}, rather than an $n\times 1$ vector of outcomes, we view $\bm{Y}^t$ as a mapping from $\tau_{x,r_n}\X$ to $\{0,1\}$. Let $Y_p^t$ denote the outcome of agent $p \in \tau_{x,r_n}\X$ at time $t$ according to $\bm{Y}^t$. Let $Z_p^t = (p, \bm{Z}^t(\cc_{x,r_n}p))$, which is just $p$'s type at period $t$. We assume
\begin{equation*}
  Y_p^t = \bm{1}\{U(S_p^t, Z_p^t) > 0\},
\end{equation*}

\noindent where 
\begin{equation*}
  S_p^t \equiv S(Y_{-p}^{t-1}, Z_p^{t-1}, Z^{t-1}_{-p}, A_p, A_{-p}).
\end{equation*}

\noindent The first, third, and fourth arguments of the right-hand side are random functions with domain $\tau_{x,r_n}\X\backslash\{p\}$ such that $Y_{-p}^{t-1}\colon p' \mapsto Y_{p'}^{t-1}$, $Z_{-p}^{t-1}\colon p' \mapsto Z_p^{t-1}$, and $A_p\colon p' \mapsto A_{p,p'}$. The last argument is a random function with domain $\tau_{x,r_n}\X\backslash\{p\} \times \tau_{x,r_n}\X\backslash\{p\}$ such that $A_{-p}\colon (p', p'') \mapsto A_{p',p''}$. Thus, this just generalizes the definition of $S(\cdot)$ in model \autoref{model} to model \autoref{model_dynamic_generic}. To emphasize its dependence on the primitives, we write $\bm{Y}^t(\tau_{x,r_n}\X, \W^t\cc_{x,r_n})$ for the period-$t$ outcome mapping, suppressing the implicit dependence on the network function $\bm{A}(\tau_{x,r_n}\X, \W\cc_{x,r_n})$.

This defines the evolution of the outcome time series from period 1 onwards. The initial conditions model for $\bm{Y}^0$ in model \autoref{model_dynamic_generic} is the static model of \autoref{sgsstatic} with period-0 attributes. Recalling the definition of $\lambda(\cdot)$ from that section, we assume the initial condition satisfies $\bm{Y}^0 = \lambda(\bm{A}, \tau_{x,r_n}\X, \W^0\cc_{x,r_n})$. This is just a restatement of \autoref{dyinit}(b).

\bigskip

\noindent {\bf Network Moments.} Define 
\begin{equation}
  \bm{Y}(\tau_{x,r_n}\X, \W\cc_{x,r_n}) = (\bm{Y}^t(\tau_{x,r_n}\X, \W^t\cc_{x,r_n}))_{t=0}^T.
\end{equation}

\noindent Our goal is to prove a CLT for network moments of the form
\begin{align*}
  \Lambda(\tau_{x,r}\X, \W\cc_{x,r_n}) 
  &\equiv \sum_{p\in\tau_{x,r}\X} \psi(p, \tau_{x,r}\X, \W\cc_{x,r_n}, \bm{Y}(\tau_{x,r}\X, \W\cc_{x,r_n}), \bm{A}(\tau_{x,r}\X, \W\cc_{x,r_n}))  \\
  &\equiv \sum_{p\in\tau_{x,r}\X} \psi(p, \tau_{x,r}\X, \W\cc_{x,r_n}, \bm{Y}, \bm{A}), 
\end{align*}

\noindent where the last line abbreviates by suppressing the dependence of $\bm{Y}$ and $\bm{A}$ on their arguments. Note that, by construction, $\bm{Y}$ and $\bm{A}$ are deterministic functions of the first two arguments of $\psi(\cdot)$. Thus we can absorb them into the agent statistic, writing
\begin{equation*}
  \xi(p, \tau_{x,r}\X, \W\cc_{x,r_n}) \equiv \psi(p, \tau_{x,r}\X, \W\cc_{x,r_n}, \bm{Y}, \bm{A}).
\end{equation*}

\noindent This falls within the setup of \autoref{smaster}.

%------------------------------------
\subsection{Proofs}\label{pfdynamicclt}
%------------------------------------

In the proof of \autoref{staticclt}, we construct a ``relevant set'' $J_x$ such that agent $x$'s agent statistic only depends on its arguments through $J_x$ in the sense of \autoref{constructJ}. To prove \autoref{dynamicclt}, we need to define an analogous set. We first need some new notation. For $p,p' \in \tau_{x,r_n}\X$ define 
\begin{multline*}
  D_{p,p'}^t = \bm{1}\big\{ \sup_s V(\norm{p-p'}, s, \W^t\cc_{x,r_n}(p,p')) > 0 \\ \medcap \inf_s V(\norm{p-p'}, s, \W^t\cc_{x,r_n}(p,p')) \leq 0 \big\}.
\end{multline*}

\noindent Let $\bm{D}^t$ network function on $\tau_{x,r_n}$ mapping $(p,p') \mapsto D_{p,p'}^t$. Define agent $x$'s strategic neighborhood in $\bm{D}^t$ as in \autoref{sstaticass}, now with the notation $C_{tx}^+$.

For $M$ given in \autoref{S2} and $K$ given in \autoref{klocal}, define the relevant set as
\begin{equation}
  J_{xT} \equiv J_{xT}(\tau_{x,r_n}\X, \W\cc_{x,r_n}) = \bigcup\left\{ C_{0y}^+\colon y \in \mathcal{N}_{\bm{A}}(x',K+M(T-1)) \right\}. \label{J_iT}
\end{equation}

\begin{lemma}\label{constructJT}
  Suppose $\psi(\cdot)$ satisfies \autoref{klocal} and the initial conditions model satisfies Assumptions \ref{S} and \ref{nocoord}. For any $x' \in \tau_{x,r_n}\X$,
  \begin{equation*}
    \psi(x', \tau_{x,r_n}\X, \W\cc_{x,r_n}, \bm{Y}, \bm{A}) = \psi(x', J_{x'T}, \W\cc_{x,r_n}, \bm{Y}, \bm{A}). 
  \end{equation*}
\end{lemma}
\begin{proof}
  The agent statistic of $x'$ is a function of its arguments through $\{Y_p^t\colon p \in \mathcal{N}_{\bm{A}}(x',K), t = 0, \dots, T\}$ by \autoref{klocal}. By \autoref{S2}, $Y_p^T$ is a function of outcomes and attributes only through the period $T-1$ outcomes and attributes of agents in $\mathcal{N}_{\bm{A}}(p,M)$. Hence, the agent statistic of $x'$ is a function of agents in $\mathcal{N}_{\bm{A}}(x', K+M)$.  Now, these period $T-1$ outcomes of agents $p' \in \mathcal{N}_{\bm{A}}(p,M)$, in turn, are functions of period $T-2$ outcomes of agents in $\mathcal{N}_{\bm{A}}(p',M)$, and so on. Repeat this argument until we hit period 0. At this point, we have found that the agent statistic of $x'$ is a function of agents in $\mathcal{N}_{\bm{A}}(x', K+M(T-1))$, and we are considering the period-0 outcome of some agent $\ell$, $Y_\ell^0$. By the initial conditions model and \autoref{CJlem}, $Y_\ell^0$ is fully determined by the attributes of agents in $C_{0\ell}^+$. The result follows.
\end{proof}

Define
\begin{equation}
  \RR_{nT}(x') \equiv \RR_T(x',\tau_{x,r_n}\X, \W\cc_{x,r_n}) = \max_{y\in J_{x'T}} \norm{x'-y}, \label{RR_iT}
\end{equation}

\noindent the {\em radius of stabilization} of $x' \in \tau_{x,r_n}\X$. 

\bigskip

\begin{proof}[Proof of \autoref{dynamicclt}]
  We apply \autoref{master-clt} by verifying Assumptions \ref{xistab}--\ref{xipoly}. Assumptions \ref{xibd} and \ref{xipoly} hold under \autoref{psibdsimple} (or equivalently \autoref{psibd}). The argument for verifying Assumptions \ref{xistab} and \ref{bpestab} is the same as the static case (Lemmas \ref{static_xistab} and \ref{static_bpestab}). Observe that the radius of stabilization in the static case \autoref{RR_i} is virtually the same as \autoref{RR_iT}. The only differences are that we use the $K+M(T-1)$ rather than the $K$-neighborhood (but this does not matter because both $K$ and $T$ are fixed constants) and the period-0 strategic neighborhood $C_{0y}^+$ instead of the strategic neighborhood $C_y^+$. The only difference between the latter two objects is that $C_y^+$ is defined using the attribute process $\W\cc_{x,r_n}$ (since there are no time periods), whereas $C_{0y}^+$ is defined using the period-0 attribute process $\W^0\cc_{x,r_n}$. Thus, in following the proofs of Lemmas \ref{static_xistab} and \ref{static_bpestab}, we only need to modify the definition of the branching process $\hat{\mathfrak{X}}^K_{r_n}(y, \bm{Z}(\cc_{x,r_n}y))$ in two ways. First, replace $K$ with $K+M(T-1)$. Second, change the intensity measure of $\mathfrak{X}_{r_n}(x',z')$ given in \autoref{Dintens} by using period-0 attributes $\W^0\cc_{x,r_n}$ instead. (Recall from the proof of Lemma \ref{static_xistab} that this branching process is used to stochastically bound $\abs{C_y^+}$.)
\end{proof}

%----------------------------------------------------------------------
\section{CLT for Stabilizing Functionals}\label{smaster}
%----------------------------------------------------------------------

The proofs of Theorems \ref{staticclt} and \ref{dynamicclt} consist of verifying high-level conditions for a general CLT stated next. This is a minor modification of Theorem C.2 of \cite{leung2019normal}. There are two main differences relative to their setup. The first is that positions are given by $\tilde\rho_1, \tilde\rho_2, \dots$, which are i.i.d.\ draws from $f$, and attributes are drawn conditional on these positions. In our setup, positions are given instead by $\omega_n\tilde\rho_1, \omega_n\tilde\rho_2, \dots$, which is directly an increasing domain asymptotics setup. For this reason, we use the projection $\cc_{x,r_n}$ (defined in \autoref{sgsstatic}) in place of the projection $\bm{p}_{x,r_n}$ used by \cite{leung2019normal}; \autoref{sgsstatic} illustrates the purpose of $\cc_{x,r_n}$. The second difference is that we do not derive a closed-form expression for the limiting variance, which allows us to simplify the conditions and avoid a high-level continuity condition.

%------------------------------------
\subsection{Definitions}\label{shldefs}
%------------------------------------

As in the main text, define $\{\tilde\rho_i\}_{i=1}^n \stackrel{iid}\sim f$, a density on $\R^d$ bounded away from zero and infinity. Let $\X = \{\tilde\rho_i\}_{i=1}^N \cup G$ for some a.s.\ finite random variable $N$ that may depend on $n$ and finite set $G \subseteq \R^d$. Recall from \autoref{sgsstatic} the definitions of $\W(\cdot,\cdot)$, $\bm{Z}(\cdot)$, $\bm{\zeta}(\cdot,\cdot)$, $r_n$, $\tau_{x,r}$, and $\cc_{x,r}$. The first three are assumed to be independent of $N$. Define $Q(x,r)$ as the cube in $\R^d$ centered at $x$ with side length $r$.

For $x,y \in \R^d$ and $r \in \R_+$, we prove a central limit theorem for functionals of the form
\begin{equation*}
  \Lambda(\tau_{y,r}\X, \W\cc_{y,r}) = \sum_{x \in \tau_{y,r}\X} \xi(x, \tau_{y,r}\X, \W\cc_{y,r}), 
\end{equation*}

\noindent where $\xi(\cdot)$ has range $\R^m$ and satisfies
\begin{equation*}
  \xi(\tau_{y,r}x, \tau_{y,r}\X, \W\cc_{y,r}) = \xi(\tau_{y',r}x, \tau_{y',r}\X, \W\cc_{y',r})
\end{equation*}

\noindent for any $y'\in\R^d$. This last property holds trivially in the static and dynamic models.

\begin{definition}\label{rosxi}
  $\RR_\xi(x,\X,\W\cc_{x,r})$\footnote{If $x \not\in \X$, we abbreviate $\RR_\xi(x,\X,\W\cc_{x,r}) \equiv \RR_\xi(x,\X \cup \{x\},\W \cc_{x,r})$.} is a {\em radius of stabilization} of $\xi$ if 
  \begin{equation}
    \xi(x, \X, \W ) = \xi(x, \X \cap H, \W \cc_{x,r}) \label{1RR}
  \end{equation}
  
  \noindent for any $r \in \R_+$, and $H \supseteq Q(x, \RR_\xi(x,\X,\W \cc_{x,r}))$, and 
  \begin{equation}
    \RR_\xi(x, \tau_{x,r}\X, \W \cc_{x,r}) = \RR_\xi(\tau_{y,r}x, \tau_{y,r}\X, \W\cc_{x,r}) \label{Rdiff}
  \end{equation}

  \noindent for any $y \in \R^d$.
\end{definition}

\begin{definition}\label{incrad}
  A radius of stabilization $\RR_\xi(\cdot)$ is {\em increasing} on $\{\X\}_{n\in\mathbb{N}}$ if for all $x\in\R^d$, $n$ sufficiently large, and $H \subseteq \R^d$,
  \begin{equation*}
    \RR_\xi(x, \tau_{x,r_n}\X \cap H, \W\cc_{x,r}) \leq \RR_\xi(x, \tau_{x,r_n}\X, \W\cc_{x,r}). 
  \end{equation*}
\end{definition}

\begin{definition}\label{sixi}
  Let $\RR_\xi(\cdot)$ be a radius of stabilization and $G \equiv G(\cdot) \equiv \{G_n(\cdot)\}_{n\in\mathbb{N}}$ be a set-valued functions, where for any $x \in \R^d$, $G_n(x) \subseteq \R^d$. We say $\xi$ is {\em $(\RR_\xi,G)$-stabilizing on $\{\X\}_{n\in\mathbb{N}}$} if for any $x\in\R^d$ and $x' \in \tau_{x,r_n}(\X \cup G_n(x))$,
  \begin{equation*}
    \RR_\xi(x', \tau_{x,r_n}(\X \cup G_n(x)), \W\cc_{x,r}) = O_p(1). 
  \end{equation*}
\end{definition}

\begin{definition}\label{exstab}
  $\xi$ is {\em $\RR_\xi$-externally stabilizing on} $\{\X\}_{n\in\mathbb{N}}$ if for all $n$ and $x\in\R^d$, there exists $\RR_n^e(x) \equiv \RR_n^e(x, \tau_{x,r_n}\X, \W\cc_{x,r}) \geq 0$ such that the following properties hold.
  \begin{enumerate}[(a)]
    \item $\RR_n^e(x) = O_p(1)$ uniformly in $x$, i.e.
      \begin{equation*}
	\lim_{R\rightarrow\infty} \lim_{n\rightarrow\infty} \sup_{x\in\R^d} \prob(\RR_n^e(x) > R) = 0.
      \end{equation*}
      
    \item Define $\RR_X \equiv \RR_\xi(X,\tau_{X,r_n}(\X \cup \{x\}),\W \cc_{x,r})$ and $Q_X \equiv Q(X,\RR_X r_n)$. For any $X \in \X$, if
      % this is used in proof of lemma \ref{xipstab}
      \begin{equation*}
	\xi\big(X, \tau_{X,r_n}\big( (\X\cup\{x\}) \cap Q_X \big), \W \cc_{x,r}\big) \\ \neq \xi\big(X, \tau_{X,r_n}\big(\X \cap Q_X\big), \W \cc_{x,r}\big) \quad\text{a.s.,}
      \end{equation*}

      \noindent then $Q(\tau_{x,r_n}X,\RR_X) \subseteq Q(x,\RR_n^e(x))$ for any $n$ sufficiently large.

    \item For all $x\in\R^d$, $n$ sufficiently large, and $H \subseteq \R^d$,
      \begin{equation*}
	\RR_n^e(x, \tau_{x,r_n}\X \cap H, \W\cc_{x,r}) \leq \RR_n^e(x, \tau_{x,r_n}\X, \W\cc_{x,r}). 
      \end{equation*}
  \end{enumerate}
\end{definition}

\begin{definition}\label{binexpstab}
  $\xi$ is {\em $\RR_\xi$-binomial exponentially stabilizing} if for some $\tilde{n},c,\epsilon>0$,
  \begin{equation*}
    \sup_{n>\tilde{n}} \sup_{m \in [n/2, 3n/2]} \sup_{H \subseteq \R^d} \sup_{x,y\in\R^d} \prob\left( \RR_\xi(x, \tau_{x,r_n}( \X_{m-1} \cap H) \cup \{y\}, \W\cc_{x,r}) \geq r \right) \leq c\,\text{exp}\left\{ -cr^\epsilon \right\}. 
  \end{equation*}
  
  \noindent It is {\em $\RR_\xi$-Poisson exponentially stabilizing} if for $N \sim \text{Poisson}(n)$ and some $\tilde{n},c,\epsilon>0$,
  \begin{equation*}
    \sup_{n>\tilde{n}} \sup_{x,y\in\R^d} \prob\left( \RR_\xi(x, \tau_{x,r_n} \X \cup \{y\}, \W\cc_{x,r}) \geq r \right) \leq c\,\text{exp}\left\{ -cr^\epsilon \right\}. 
  \end{equation*}
\end{definition}

%------------------------------------
\subsection{Assumptions}\label{gencltassumps}
%------------------------------------

We state high-level conditions required for a CLT. 

\begin{assump}[Stabilization]\label{xistab}
  There exists a radius of stabilization $\RR_\xi$ such that the following statements hold.
  \begin{enumerate}[(a)]
    \item For $\nu(n) \in [n/2,3n/2]$ for all $n$, $\xi$ is $(\RR_\xi,G)$-stabilizing on $\{\X_N\}_{n\in\mathbb{N}}$ for $N = \nu(n)$ and $N \sim \text{Poisson}(n)$ for any $G = \{G_n(x)\}_{n\in\mathbb{N}}$ such that $G_n(x) \subseteq \{x',y\}$ for some $x',y \in \R^d$ and all $n\in\mathbb{N}$, $x\in\R^d$.

    \item For $\nu(n) < n$ and $\nu(n)/n \rightarrow 1$, $\xi$ is $\RR_\xi$-externally stabilizing on $\{\X_N\}_{n\in\mathbb{N}}$ for $N = \nu(n)$ and $N \sim \text{Poisson}(n)$.

    \item $\RR_\xi$ is increasing on $\{\X_N\}_{n\in\mathbb{N}}$ for $N = \nu(n)$ and $N \sim \text{Poisson}(n)$ with $\nu(n)$ defined as in either (a) or (b).
  \end{enumerate}
\end{assump}

\begin{assump}[Exponential Stabilization]\label{bpestab}
  There exists a radius of stabilization $\RR_\xi$ such that $\xi$ is $\RR_\xi$-binomial exponentially stabilizing and $\RR_\xi$-Poisson exponentially stabilizing.
\end{assump}

\begin{assump}[Bounded Moments]\label{xibd}
  There exists $C<\infty$ such that
  \begin{equation*}
    \E\big[ \abs{\xi\left(x, \tau_{x,r_n}(\X_m \cup G) \cap H_n, \W\cc_{x,r_n} \right)}^8 \big] < C 
  \end{equation*}

  \noindent for all $n\in\mathbb{N}, m \in [n/2, 3n/2], \{H_n\}_{n\in\mathbb{N}}$ with $H_n \subseteq \R^d$, $G \in \{ \{y\}, \emptyset \}$, and $x,y \in\R^d$. 
\end{assump}

\begin{assump}[Polynomial Bound]\label{xipoly}
  For any a.s.\ finite set $\X \subseteq \R^d$, $x \in \R^d$, $r \in \R_+$, 
  \begin{equation}
    \abs{\xi(x, \X, \W\cc_{x,r})} \leq c \abs{\X}^c
  \end{equation}

  \noindent a.s., for some positive constant $c$.
\end{assump}

%------------------------------------
\subsection{Main Result}\label{masterthm}
%------------------------------------

Let $\X_n = \{\tilde\rho_i\}_{i=1}^n$ and $x \in \R^d$. Define the {\em add-one cost}
\begin{equation}
  \Xi_x(r_n^{-1}\X_n, \W) = \Lambda\left( r_n^{-1}\X_n \cup \{x\}), \W \right) - \Lambda\left( r_n^{-1}\X_n, \W \right) \label{add1}
\end{equation}

\noindent and variance $\bm{\Sigma}_n = \var(n^{-1/2} \Lambda(r_n^{-1}\X_n, \W))$. Recall that $\lambda_\text{min}(\bm{\Sigma}_n)$ is the smallest eigenvalue of $\bm{\Sigma}_n$ and $\bm{I}_m$ the $m\times m$ identity matrix.

\begin{theorem}\label{master-clt}
  Suppose that, for each component of the vector $\xi(r_n^{-1}X, r_n^{-1}\X_n, \W)$, Assumptions \ref{xistab} and \ref{bpestab} hold for the same radius of stabilization $\RR_\xi$ (the radius may be component-specific). Further suppose that $c'\Xi_{r_n^{-1}X}(r_n^{-1}\X_n, \W)$ is asymptotically non-degenerate for all $c \in \R^m\backslash\{\zero\}$. Then under Assumptions \ref{xibd} and \ref{xipoly}, $\liminf_{n\rightarrow\infty} \lambda_\text{min}(\bm{\Sigma}_n) > 0$, and
  \begin{equation*}
    n^{-1/2} \bm{\Sigma}_n^{-1/2} \left( \Lambda(r_n^{-1}\X_n, \W) - \E\big[ \Lambda(r_n^{-1}\X_n, \W) \big] \right) \dlimarrow \mathcal{N}\left( \zero, \bm{I}_m \right). 
  \end{equation*}
\end{theorem}
\begin{proof}
  The proof is almost the same as Theorem C.2 of \cite{leung2019normal} with the following minor modifications. First, in place of the composition $\W\bm{p}_{x,r}$ used everywhere in their proof, we use $\W\bm{c}_{x,r}$. This replacement has no effect on the original arguments, since the assumptions are also modified in the same way. 
  
  Second, whereas several parts of the proof of Theorem C.2 are concerned with the derivation of a closed-form expression for $\lim_{n\rightarrow\infty} \bm{\Sigma}_n$, we do not derive it for our setting. This has two minor effects on the proof. First, the proof of the Poissonized CLT (Theorem H.1) is simpler and does not require Lemma H.4, which derives the limit variance under the Poissonized model. We state the simplified proof in \autoref{poi-clt} below. Second, the de-Poissonization argument needs some minor modification to avoid the use of certain limit quantities that are used to characterize $\lim_{n\rightarrow\infty} \bm{\Sigma}_n$. These modifications are given after the statements of Lemmas \ref{poi-clt}--\ref{2.14} below.

  Note that the lemmas below all consider the case where the dimension $m$ of $\xi$ is one, so they only establish a univariate CLT, However, with this result, we can extend to $m>1$ as follows. Let $t \in \R^m$. Under the assumptions of the theorem, $t' \Lambda(r_n^{-1}\X_n, \W)$ satisfies the assumptions of the CLT for the case $m=1$. Then the result follows from the Cram\'{e}r-Wold device.
\end{proof}

%---------------------------
The following lemmas make use of the add-one cost $\Xi_x$ defined in \autoref{add1}. For an $\R_+$-valued function $\lambda(\cdot)$, let $\mathcal{P}_\lambda$ be an inhomogeneous Poisson point process with intensity function $\lambda(\cdot)$. Note that for $N \sim \text{Poisson}(n)$, $\{\tilde\rho_i\}_{i=1}^N \stackrel{d}= \mathcal{P}_{nf}$.

\begin{lemma}[Poissonized CLT]\label{poi-clt}
  Let $m=1$. Suppose Assumptions \ref{xistab}--\ref{xipoly} hold, 
  \begin{equation}
    \sup_{n \in \mathbb{N}} \sup_{H \subseteq \R^d} \sup_{x\in\R^d} \E\left[ \Xi_x\big(\tau_{x,r_n}\mathcal{P}_{nf} \cap H, \W \cc_{x,r_n} \big)^4 \right] < \infty, \label{usi}
  \end{equation}
  
  \noindent and $\liminf_{n\rightarrow\infty} \sigma_n^2 > 0$, where $\sigma_n^2 = \var(n^{-1/2}\Lambda(r_n^{-1}\mathcal{P}_{nf}, W))$. Then
  \begin{equation*}
    n^{-1/2}\sigma_n^{-1}\left( \Lambda(r_n^{-1}\mathcal{P}_{nf}, \W) - \E\big[ \Lambda(r_n^{-1}\mathcal{P}_{nf}, \W) \big] \right) \dlimarrow \mathcal{N}(0,1). 
  \end{equation*}
\end{lemma}
\begin{proof}
  This is a minor modification of the proof of Theorem H.1 in \cite{leung2019normal}. We will rewrite $\Lambda(r_n^{-1}\mathcal{P}_{nf}, \W) - \E[ \Lambda(r_n^{-1}\mathcal{P}_{nf}, \W)]$ as the sum of a martingale difference sequence. Partition $\R^d$ into cubes with side length $r_n$. Label those that intersect the support of $f$ as $Q_1, \dots, Q_{k_n}$, with respective centers $x_1, \dots, x_{k_n}$ labeled in increasing lexicographic order. For each $\ell = 1, \dots, k_n$, define $\mathcal{F}_\ell$ as the $\sigma$-field generated by the points of 
  \begin{equation*}
    \mathcal{P}_1 \medcap \left\{ \bigcup_{1 \leq m\leq \ell} \left\{ Q_m \times [0,\infty) \right\} \right\}.
  \end{equation*}
  
  \noindent We can then define the martingale differences $\delta_\ell = \E\left[ \Delta_{x_\ell} \mid \mathcal{F}_\ell \right]$, where
  \begin{equation*}
    \Delta_{x_\ell} = \sigma_n^{-1}\Lambda\left( r_n^{-1}\mathcal{P}_{nf}, \W \right) - \sigma_n^{-1}\Lambda\left( r_n^{-1} \left( (\mathcal{P}_{nf} \backslash Q_\ell) \cup (\mathcal{P}_{nf}' \cap Q_\ell) \right), \W \right).
  \end{equation*}

  \noindent Notice $\sigma_n\delta_\ell = \E\left[ \Lambda\left( r_n^{-1}\mathcal{P}_{nf}, \W \right) \mid \mathcal{F}_\ell \right] - \E\left[ \Lambda\left( r_n^{-1}\mathcal{P}_{nf}, \W \right) \mid \mathcal{F}_{\ell-1} \right]$, and
  \begin{equation*}
    \Lambda(r_n^{-1}\mathcal{P}_{nf}, \W ) - \E\left[ \Lambda(r_n^{-1}\mathcal{P}_{nf}, \W) \right] = \sigma_n \sum_{\ell=1}^{k_n} \delta_\ell. 
  \end{equation*}
  
  \noindent Furthermore, $\{\delta_\ell\}_{\ell=1}^{k_n}$ is a martingale difference sequence with filtration $\left\{ \mathcal{F}_\ell \right\}_{\ell=0}^{k_n}$, where $\mathcal{F}_0$ is the trivial $\sigma$-algebra. By orthogonality of martingale differences,
  % https://math.stackexchange.com/questions/515613/a-square-integrable-martingale-has-orthogonal-increments
  \begin{equation}
    \var\left( \Lambda(r_n^{-1}\mathcal{P}_{nf}, \W) \right) = \sigma_n \sum_{\ell=1}^{k_n} \E[\delta_\ell^2]. \label{varo}
  \end{equation}

  \noindent We complete the proof by verifying the conditions for a CLT for martingale difference arrays \citep[][Theorem 3.2]{HallHeyde2014}. There are three such conditions, the first of which is
  \begin{equation}
    \sup_{n\in\mathbb{N}} \E\left[ \max_{1\leq \ell\leq k_n} \frac{1}{n} \delta_\ell^2 \right] < \infty. \label{mds-one}
  \end{equation}

  \noindent The left-hand side is bounded above by 
  \begin{equation}
    \sup_{n\in\mathbb{N}} \frac{1}{n} \sum_{\ell=1}^{k_n} \E[\delta_\ell^2] \leq \sup_{n\in\mathbb{N}} \frac{k_n}{n} \max_{1\leq \ell\leq k_n} \sigma_n^{-1} \E\left[\Delta_{x_\ell}^2\right]. \label{mds-one2}
  \end{equation}

  \noindent This is finite by \autoref{usi}, the assumption that $\sigma_n$ is asymptotically non-degenerate, and the fact that $k_n = O(n)$, thereby establishing \autoref{mds-one}.

  The second condition needed for a martingale difference CLT is
  \begin{equation}
    n^{-1/2} \max_{1\leq \ell\leq k_n} \abs{\delta_\ell} \plimarrow 0. 
    \label{mds-two}
  \end{equation}

  \noindent By Markov's inequality, 
  \begin{equation*}
    \prob\left( n^{-1/2} \max_{1\leq \ell\leq k_n} \abs{\delta_\ell} \geq \epsilon \right) \leq \sum_{\ell=1}^{k_n} \frac{1}{n^{2} \epsilon^4} \E[\delta_\ell^4], 
  \end{equation*}

  \noindent which tends to zero by the arguments for finiteness of \autoref{mds-one2}. This proves \autoref{mds-two}.

  The last condition required for a martingale difference CLT is $n^{-1} \sum_{\ell=1}^{k_n} \delta_\ell^2 \plimarrow 1$, which holds trivially by \autoref{varo}.
\end{proof}

%---------------------------
The next three lemmas are restatements of de-Poissonization lemmas given in \S H.4 of \cite{leung2019normal}. In what follows, let $X$ be a draw from $f$ independent of $\tilde\rho_1, \tilde\rho_2, \dots$.

\begin{lemma}\label{2.12} 
  Let $m=1$. Suppose that the conclusion of \autoref{poi-clt} holds. Further suppose there exist $\gamma > 0.5$ and a sequence of constants $\{\alpha_n\}_{n\in\mathbb{N}}$ such that
  \begin{align}
    &\lim_{n\rightarrow\infty} \left( \sup_{n-n^\gamma \leq m \leq n+n^\gamma} \big| \E\left[ \Xi_X(\tau_{X,r_n}\X_m, \W\cc_{X,r_n}) \right] - \alpha_n \big| \right) = 0, \label{2.38} \\
    &\lim_{n\rightarrow\infty} \left( \sup_{n-n^\gamma \leq m < m' \leq n+n^\gamma} \big| \E\left[ \Xi_X(\tau_{X,r_n}\X_m, \W\cc_{X,r_n}) \Xi_X(\tau_{X,r_n}\X_{m'}, \W\cc_{X,r_n}) \right] - \alpha_n^2 \big| \right) = 0, \label{2.39} \\
    &\lim_{n\rightarrow\infty} \left( \sup_{n-n^\gamma \leq m < m' \leq n+n^\gamma} n^{-1/2} \E\left[ \Xi_X(\tau_{X,r_n}\X_m, \W\cc_{X,r_n})^2 \right] \right) = 0.
    \label{2.40}
  \end{align}

  \noindent If \autoref{xipoly} holds, then for $\sigma_n$ defined in \autoref{poi-clt}, $\alpha_n^2 \leq \sigma_n^2$. If additionally $\liminf_{n\rightarrow\infty} (\sigma_n-\alpha_n) > 0$, then
  \begin{equation*}
    n^{-1/2} (\sigma_n^2-\alpha_n^2)^{-1/2} \left( \Lambda(r_n^{-1}\X_n, \W) - \E\big[ \Lambda(r_n^{-1}\X_n, \W) \big] \right) \dlimarrow \mathcal{N}(0,1).  
  \end{equation*}
\end{lemma}
\begin{proof}
  The proof is the same as that of Lemma H.8 of \cite{leung2019normal}.
\end{proof}

%---------------------------
\begin{lemma}\label{2.13}
  Let $m=1$. Suppose there exists a sequence of pairs of i.i.d.\ random variables $\{(\Delta_n,\Delta_n')\}_{n\in\mathbb{N}}$ such that for any sequence $\{\nu(n), \nu'(n)\}_{n\in\mathbb{N}}$ satisfying $\nu(n) < \nu'(n)$ and $\nu(n)/n, \nu'(n)/n \rightarrow 1$, we have
  \begin{equation}
    \max\left\{ \abs{\Xi_X(\tau_{X,r_n}\X_{\nu(n)}, \W\cc_{X,r_n}) - \Delta_n}, \abs{\Xi_X(\tau_{X,r_n}\X_{\nu'(n)}, \W\cc_{X,r_n}) - \Delta_n'} \right\} \plimarrow 0. \label{2.46}
  \end{equation}
  
  \noindent If \autoref{xibd} holds, then $\sup_n \E[\Delta_n] < \infty$, and \autoref{2.38}--\autoref{2.40} hold with $\alpha_n = \E[\Delta_n]$ and $\gamma = 0.75$.
\end{lemma}
\begin{proof}
  The proof is the same as that of Lemma H.9 of \cite{leung2019normal}, except that we do not need a limit for $\E[\Delta_n]$. 
\end{proof}

%---------------------------
For the next lemma, define
\begin{align*}
  \tilde{\Xi}_{i,n} &= \Lambda(r_n^{-1}\X_n, \W) - \Lambda(r_n^{-1}(\X_n\backslash\{\tilde\rho_i\}), \W), \\
  \tilde{\Xi}^-_{i,n} &= \Lambda(r_n^{-1}(\X_{n+1}\backslash\{\tilde\rho_i\}), \W) - \Lambda(r_n^{-1}(\X_n\backslash\{\tilde\rho_i\}), \W).
\end{align*}

\begin{lemma}\label{2.14}
  Let $m=1$. Suppose there exists a sequence of random variables $\{\Delta_n\}_{n\in\mathbb{N}}$ such that (a) if $\Delta_n'$ denotes an independent copy of $\Delta_n$, then for any sequence $\{\nu(n)\}_{n\in\mathbb{N}}$ satisfying $\nu(n) \leq n$ and $\nu(n)/n \rightarrow 1$,
  \begin{equation}
    \max\left\{ \abs{\Xi_X(\tau_{X,r_n}\X_{\nu(n)-1}, \W\cc_{X,r_n}) - \Delta_n}, \abs{\tilde{\Xi}^-_{\nu(n),n} - \Delta_n'} \right\} \plimarrow 0, \label{2.52}
  \end{equation}

  \noindent (b) $\Delta_n$ is asymptotically non-degenerate, and (c)
  \begin{equation}
    \prob\left( \Xi_X(\tau_{X,r_n}\X_{\nu(n)-1}, \W\cc_{X,r_n}) \neq \tilde{\Xi}_{\nu(n),n} \right) \rightarrow 0. \label{2.53}
  \end{equation}

  \noindent If \autoref{xibd} holds, then 
  \begin{equation*}
    \liminf_{n\rightarrow\infty} \var\left( n^{-1/2}\Lambda(r_n^{-1}\X_n, \W) \right) > 0. 
  \end{equation*}
\end{lemma}
\begin{proof}
  The proof is the same as that of Lemma H.10 of \cite{leung2019normal}
\end{proof}

%---------------------------
\begin{proof}[Proof of \autoref{master-clt}]
  Let $m=1$. The arguments that follow are minor modifications of the proof of Theorem H.2 of \cite{leung2019normal}. We verify the conditions of \autoref{2.13} and \ref{2.14}. Then \autoref{2.12} establishes the result. It suffices to verify assumption (b) of \autoref{2.14} and to prove that if $\{(\nu(n), \nu'(n))\}_{n\in\mathbb{N}}$ is a sequence satisfying $\nu(n) < \nu'(n)$, $\nu(n) < n$, and $\nu(n)/n, \nu'(n)/n \rightarrow 1$, that \autoref{2.46}, \autoref{2.52}, and \autoref{2.53} hold. 

  Following the coupling construction in the proof of Theorem H.2 of \cite{leung2019normal}, we can construct independent draws $X,Y$ from $f$ and inhomogeneous Poisson processes on $\R^d$ labeled $\mathcal{P}^a_{nf}$ and $\mathcal{P}^b_{nf}$ on the same probability space\footnote{Replace the set $F_X \times [0,nf(X)]$ in their coupling construction with $\{(x,t) \in F_X \times [0,\infty)\colon t \leq nf(x)\}$ and likewise with $F_Y \times [0,nf(Y)]$.} such that
  \begin{equation*}
    \W\cc_{X,r_n}(\mathcal{P}^a_{nf} \cup \{X\}) \indep \W\cc_{Y,r_n}(\mathcal{P}^b_{nf} \cup \{Y\}). 
  \end{equation*}
  
  \noindent Therefore, the following two random variables are independent:
  \begin{equation*}
    \Delta_n \equiv \Xi_X(\tau_{x,r_n}\mathcal{P}^a_{nf}, \W\cc_{X,r_n}) \quad\text{and}\quad \Delta'_n \equiv \Xi_Y(\tau_{x,r_n}\mathcal{P}^b_{nf}, \W\cc_{Y,r_n}).
  \end{equation*}

  \noindent The coupling also constructs $\Xi_X(\tau_{X,r_n}\X_{\nu(n)}, \W\cc_{X,r_n}), \Xi_Y(\tau_{Y,r_n}\X_{\nu'(n)}, \W\cc_{Y,r_n})$ on the same probability space.

  We now prove \autoref{2.46} and assumption (b) of \autoref{2.14}, which require some definitions. For any $R>0$, define the event
  \begin{equation*}
    E_X^n(R) = \left\{ \tau_{X,r_n}\X_{\nu(n)} \cap Q(X,R) = \tau_{X,r_n}\mathcal{P}_{nf}^a \cap Q(X,R) \right\}. 
  \end{equation*}
  
  \noindent Also define 
  \begin{equation*}
    \RR_\Lambda(X) = \max\{\RR_\Lambda(X, \tau_{X,r_n}\X_{\nu(n)}, \W\cc_{X,r_n}), \RR_\Lambda(X, \tau_{X,r_n}\mathcal{P}^a_{nf}, \W\cc_{X,r_n})\}, 
  \end{equation*}

  \noindent where $\RR_\Lambda(X, \tau_{X,r_n}\X_{\nu(n)}, \W\cc_{X,r_n})$ satisfies
  \begin{equation*}
    \Xi_X(\tau_{X,r_n}\X_{\nu(n)}, \W\cc_{X,r_n}) = \Xi_X(\tau_{X,r_n}\X_{\nu(n)} \cap Q(X,R), \W\cc_{X,r_n})
  \end{equation*}

  \noindent for any $R>\RR_\Lambda(X, \tau_{X,r_n}\X_{\nu(n)}, \W\cc_{X,r_n})$  and similarly for $\RR_\Lambda(X, \tau_{X,r_n}\mathcal{P}^a_{nf}, \W\cc_{X,r_n})$. 

  For states of the world in the event $E_X^n(R) \cap \{R > \RR_\Lambda(X)\}$, by definition of $\RR_\Lambda(X)$,
  \begin{multline*}
     \Xi_X(\tau_{X,r_n}\X_{\nu(n)}, \W\cc_{X,r_n}) = \Xi_X(\tau_{X,r_n}\X_{\nu(n)} \cap Q(X,R), \W\cc_{X,r_n}) \\ = \Xi_X(\tau_{X,r_n}\mathcal{P}^a_{nf} \cap Q(X,R), \W\cc_{X,r_n}) = \Xi_X(\tau_{X,r_n}\mathcal{P}^a_{nf}, \W\cc_{X,r_n}). 
  \end{multline*}

  \noindent Therefore, for any $\epsilon>0$,
  \begin{multline}
    \prob\left( \big|\Xi_X(\tau_{X,r_n}\X_{\nu(n)}, \W\cc_{X,r_n}) - \Xi_X(\tau_{X,r_n}\mathcal{P}^a_{nf}, \W\cc_{X,r_n})\big| > \epsilon \right) \\ \leq \prob(E_X^n(R)^c) + \prob(\RR_\Lambda(X) > R), \label{cm0}
  \end{multline}

  \noindent where $E_X^n(R)^c$ is the complement of $E_X^n(R)$. Under \autoref{xistab}, $\Lambda$ is $\RR_\Lambda$-stabilizing in the sense of Definition H.2 of \cite{leung2019normal} by Lemma H.11 of that paper. Hence, $\RR_\Lambda(X, \tau_{X,r_n}\X_{\nu(n)}, \W\cc_{X,r_n})$ and $\RR_\Lambda(X, \tau_{X,r_n}\mathcal{P}^a_{nf}, \W\cc_{X,r_n})$ can be constructed such that $\RR_\Lambda(X) = O_p(1)$. For such $\RR_\Lambda(X)$, we can choose $R>0$ large enough such that for all $n$ sufficiently large and any $\varepsilon > 0$, $\prob(\RR_\Lambda(X) > R) < \varepsilon/2$. By Lemma H.1 of \cite{leung2019normal}, for any such $R$, we can choose $n$ large enough such that $\prob(E_X^n(R)^c) < \varepsilon/2$. Combining these facts with \autoref{cm0}, we have $\abs{\Xi_X(\tau_{X,r_n}\X_{\nu(n)}, \W\cc_{X,r_n}) - \Delta_n} \plimarrow 0$, and by an identical argument, $\abs{\Xi_Y(\tau_{Y,r_n}\X_{\nu'(n)}, \W\cc_{Y,r_n}) - \Delta'_n} \plimarrow 0$. This establishes \autoref{2.46}. 
  
  To show assumption (b) of \autoref{2.14}, note that, following the argument above, $c'\Xi_{r_n^{-1}X}(r_n^{-1}\X_n, \W)$ converges to $\Delta_n$, so the latter is asymptotically non-degenerate. 
  
  Finally, the proof of \autoref{2.52} and \autoref{2.53} is similar to that of \autoref{2.46}.
\end{proof}

%----------------------------------------------------------------------
\section{Branching Process Lemmas}\label{introbp}
%----------------------------------------------------------------------

This section restates branching process results due to \cite{leung2019normal}, which are used to prove the CLTs. There is only some slight modification in notation relative to their case due to differences outlined at the start of \autoref{smaster}. We first define a branching process used to bound the sizes of components in $\bm{D}$ defined in \autoref{sgsstatic}. Let $d_z$ be the dimension of $Z_i$, $x \in \R^d$, $z \in \R^{d_z}$ and $r \geq 0$. Recall the definition of $\varphi(p,z;p',z')$ from \autoref{L_r} and $r_n = \omega_n^{-1}$. Let $\bar{f} = \sup_{p\in\R^d} f(p)$. Define $\mathfrak{X}_{r}(p,z)$ as the size of the branching process on $\R^{d+d_z}$ starting at a particle of type $(p,z)$, where the offspring of a type $(p',z')$ particle is given by a Poisson point process on $\R^d \times \R^{d_z}$ with intensity
\begin{equation}
  \text{d}\pi_{r}(p',z'; p'',z'') \equiv \kappa\bar{f} (1+r) \varphi(p',z';p'',z'') \,\text{d}\Phi^*(z'') \,\text{d}p'', \label{Dintens}
\end{equation}

\noindent where $\Phi^*$ is defined in \autoref{dfrag}. In brief, this branching process is generated as follows. We initialize the process at the ``first generation,'' which consists of a single particle $(p,z) \in \R^d \times \R^{d_z}$. The second generation consists of the ``offspring'' of $(p,z)$, which is the Poisson point process described above. We refer to each point of this process as a particle. The third generation consists of the offspring of the second-generation particles, which are realized according to a Poisson point process distributed as above. Being elements of $\R^d \times \R^{d_z}$, particles can be interpreted as types $\tau_i$. These processes are drawn independently conditional on the second-generation ``parent'' particles. This process is repeated indefinitely. The number of particles ultimately generated is $\mathfrak{X}_{r}(p,z)$.

This process is of interest because the expected number of offspring of a particle of type $(p,z)$ is
\begin{equation}
  \kappa\bar{f} (1+r) \int_{\R^d} \int_{\R^{d_z}} \varphi(p,z;p',z') \,\text{d}\Phi^*(z') \,\text{d}p', \label{poiDexp}
\end{equation}

\noindent which is an upper bound on the expected conditional degree of an agent in $\bm{D}$ by \eqref{r90bf0wj2}. \autoref{Jsd} below clarifies the relation between $\mathfrak{X}_{r_n}(p,z)$ and the size of the component in $\bm{D}$ of an agent with type $(p,z)$.

We next define a fixed-depth branching process used to bound the sizes of $K$-neighborhoods in $\bm{A}$. Recalling the definitions in \autoref{sparsity}, let
\begin{equation}
  \tilde{p}_1(p,z;p',z') = \tilde{\Phi}_\zeta\left( \tilde V^{-1}(\norm{p-p'}, 0) \right). \label{tilp_r}
\end{equation}

\noindent This is an upper bound on the conditional linking probability for a pair of agents in the network $\bm{A}$. Let $\tilde{\mathfrak{X}}^K_{r}(p',z')$ denote the size of the branching process on the type space $\R^d \times \R^{d_z}$ that terminates after $K+1$ generations, starting at a particle of type $(p',z')$, whose the offspring distribution is given by a Poisson point process on $\R^d \times \R^{d_z}$ with intensity
\begin{equation}
  \text{d}\tilde{\pi}_{r}(p',z'; p'',z'') = \kappa \bar{f} (1+r) \tilde{p}_1(p',z';p'',z'') \,\text{d}\Phi^*(z'') \,\text{d}p''. \label{Mintens}
\end{equation}

\noindent This is generated the same way as $\mathfrak{X}_r(p,z)$, except the intensity measure is different, and once the $(K+1)$-th generation is born, no further offspring are generated.  As with $\mathfrak{X}_{r}(p,z)$, this process is of interest because the expected number of offspring of a particle of type $(p,z)$ is an upper bound on the conditional expected degree of an agent $(p,z)$ in $\bm{A}$ by calculations similar to \autoref{r90bf0wj2}. \autoref{Jsd} below clarifies the relation between $\tilde{\mathfrak{X}}^K_{r_n}(p,z)$ and the size of the $K$-neighborhood in $\bm{A}$ of an agent with $(p,z)$.

%----------------------
\subsection{Stochastic Dominance}\label{bpsd}
%----------------------

We state a lemma used in \autoref{staticproof} to verify \autoref{xistab}. It shows that sizes of relevant sets \eqref{J_i} are stochastically dominated by the sizes of certain branching processes. It follows that strategic neighborhoods and $K$-neighborhoods in $\bm{A}$ are also stochastically dominated. Then using results from the next subsection, we can establish that the sizes of relevant sets are asymptotically bounded or have exponential tails.

Let $x,y \in \R^d$, and $\nu(n)$ satisfy $\nu(n)<n$, $\nu(n)/n\rightarrow 1$. We assume the network is realized according to model \autoref{model_static_generic}\footnote{For model \autoref{model_dynamic_generic}, everything is the same except we use period-0 attributes $\W^0\cc_{x,r_n}$.} with
\begin{equation}
  \X = \{\tilde\rho_i\}_{i=1}^N \cup \{\tau_{x,r_n^{-1}}p\} \label{Ms}
\end{equation}

\noindent for any $x,p \in \R^d$ (these quantities are defined in \autoref{shldefs}). Note that $\tau_{x,r_n}\tau_{x,r_n^{-1}}p = p$. Let $\bm{A}$ and $\bm{D}$ be defined under this model as in \autoref{sstaticass}. Also recall from that section the definition of the $\bm{D}$-component of an agent positioned at $p \in \tau_{x,r_n}\X$, denoted $C_p \equiv C(p, \tau_{x,r_n}\X, \W\cc_{x,r_n}, D)$ and strategic neighborhood $C_p^+$. Recall from \autoref{staticproof} the definition of the $K$-neighborhood $\mathcal{N}_{\bm{A}}(p,K)$, relevant set $J_p$, and radius of stabilization $\RR_n(p)$. \label{Cx'}

Let $J_p(\tau_{x,r_n}\X, \W\cc_{x,r_n})$ be the relevant set of the agent positioned at $p$ under model in part (a) of the proof of \autoref{static_xistab}. The next lemma shows that this set is stochastically bounded by the following branching process. For an initial particle $(p,\bm{Z}(\cc_{x,r_n}p))$, let
\begin{equation}
  \bm{B}_{r}^K(p,\bm{Z}(\cc_{x,r_n}p)) \label{Bxr}
\end{equation}

\noindent be the set of particles in a branching process after $K$ generations with intensity given by \autoref{Mintens}. This is our branching-process approximation of the $K$-neighborhood of an agent in $\bm{A}$. Then to approximate $\bm{D}$-component sizes, for each particle $(p,z) \in \bm{B}_{r}^K(p,\bm{Z}(\cc_{x,r_n}p))$, initiate independent branching processes with intensities given by \autoref{Dintens} whose sizes consequently have the same distribution as $\mathfrak{X}_{r}(p,z)$. Lastly, for each particle generated by the latter process, initiate a branching process that runs for only one generation with intensity given by \autoref{Mintens}. Let $\hat{\mathfrak{X}}_{r}^K(p,\bm{Z}(\cc_{x,r_n}p))$ denote the size of the overall process. \label{hatfrak}

\begin{lemma}\label{Jsd}
  For any $\epsilon > 0$ and $n$ sufficiently large,
  \begin{equation*}
    \prob( \abs{J_p(\tau_{x,r_n}\X, \W\cc_{x,r_n})} > \epsilon ) \leq \prob( \hat{\mathfrak{X}}_{r_n}^K(p,\bm{Z}(\cc_{x,r_n}p)) > \epsilon). 
  \end{equation*}
\end{lemma}
\begin{proof}
  Lemma C.1 of \cite{leung2019compute} shows that
  \begin{equation*}
    \prob( \abs{C_p^+} > \epsilon ) \leq \prob( \mathfrak{X}_{r_n}(p,\bm{Z}(\cc_{x,r_n}p)) > \epsilon) 
  \end{equation*}

  \noindent for $N=n$, which easily generalizes to $N \sim \text{Poisson}(n)$. A similar argument yields
  \begin{equation*}
    \prob( \abs{\mathcal{N}_{\bm{A}}(p,K)} > \epsilon ) \leq \prob( \tilde{\mathfrak{X}}_{r_n}^K(p,\bm{Z}(\cc_{x,r_n}p)) > \epsilon)
  \end{equation*}

  \noindent for any $K$. Then the result follows from construction of $J_p(\tau_{x,r_n}\X, \W\cc_{x,r_n})$ and $\hat{\mathfrak{X}}_{r_n}^K(p,\bm{Z}(\cc_{x,r_n}p))$.
\end{proof}

%----------------------
\subsection{Tail Bounds}\label{bp}
%----------------------

The lemmas below are used to show that the sizes of the branching processes defined in the previous subsection have exponential tails, following a line of argument due to \cite{turova_asymptotics_2012}. Let $g^\alpha_{r}(p,z) = \E[\alpha^{\mathfrak{X}_{r}(p,z)}]$, and let $T_{r}$ be the functional satisfying
\begin{equation*}
  (T_{r} \circ h)(p,z) = \int_{\R^d} \int_{\R^{d_z}} h(p',z') \text{d}\pi_r(p',z'; p,z). 
\end{equation*}

The next lemma is one of our main results in this section, establishing that the size of the branching process has exponential tails. Define
\begin{equation*}
  \psi_{r}(p,z) = \kappa \bar{f} (1+r) \int_{\R^d} \left( \int_{\R^{d_z}} \varphi(p,z;p',z')^2 \,\text{d}\Phi^*(z'') \right)^{1/2} \,\text{d}p',  
\end{equation*}

\noindent and recall the definition of $\mathcal{T}$ from \autoref{Sset}.

\begin{lemma}\label{Dexptail}
  Fix $x \in \R^d$. Suppose $f$ is bounded away from zero and infinity, and Assumptions \ref{dfrag}, \ref{sparsity}, and \ref{dreg} hold. Then for some $\alpha>1$,
  \begin{equation*}
    \sup_{(p,z)\in\mathcal{T}} g^\alpha_{r}(p,z) < \infty, 
  \end{equation*}

  \noindent for $r$ sufficiently small.
\end{lemma}
\begin{proof}
  This is a special case of Lemma C.3 of \cite{leung2019compute}.
\end{proof}

%-----------------------
\begin{lemma}\label{MKexptail}
  Define $\tilde{g}^{\alpha,K}_{r}(p,z) = \E[\alpha^{\tilde{\mathfrak{X}}^K_{r}(p,z)}]$. If \autoref{sparsity} holds, then there exists $\alpha \geq 1$ such that $\sup_{r\leq \kappa} \sup_{(p,z) \in \mathcal{T}} \E[\tilde{g}^{\alpha,K}_{r}(p,z)] < \infty$ for any $K$, with $\mathcal{T}$ defined in \autoref{Sset}.
\end{lemma}
\begin{proof}
  This is Lemma I.8 of \cite{leung2019normal}.
\end{proof}

%----------------------
\begin{lemma}\label{Mdist}
  For any $x,y \in \R^d$ and $K>0$, 
  \begin{equation*}
    \lim_{C\rightarrow\infty} \limsup_{n\rightarrow\infty} \prob\left( \max_{(p',z') \in \bm{B}_{\kappa}^K(p,\bm{Z}(\cc_{x,r_n}p))} \norm{p-p'} > C \right) = 0. 
  \end{equation*}
\end{lemma}
\begin{proof}
  This follows from the proof of Lemma I.10 of \cite{leung2019normal}. Just replace $\bm{p}_{x,r_n}$ with $\cc_{x,r_n}$.
\end{proof}

%----------------------
\subsection{Distance Bounds}\label{sexpstab}
%----------------------

We use the branching process results to show that sizes of $K$-neighborhoods, sizes of relevant sets, and radii of stabilization have exponential tails. These results are used to verify \autoref{bpestab}. We assume the network is realized according to model \autoref{model_static_generic} with 
\begin{equation*}
  \X = \{\tilde\rho_i\}_{i=1}^N \cup \{x,y\},  
\end{equation*}

\noindent where $x,y\in\R^d$ and either $N=m$, a constant in $[n/2,3n/2]$ for every $n$, or $N \sim \text{Poisson}(n)$, independent of all other primitives. Recall from \autoref{staticproof} the definitions of $D_{p,p'}$ and $M_{p,p'}$, $\bm{D}$-components $C_p$, strategic neighborhoods $C_p^+$, $K$-neighborhoods $\mathcal{N}_{\bm{A}}(p,K)$, and relevant sets $J_p$.

The first lemma shows that the distance between the ego and any alter in her $K$-neighborhood in $\bm{A}$ has exponential tails.

\begin{lemma}[$K$-Neighborhoods]\label{MKN}
  Under \autoref{sparsity}, there exists $\tilde{n}>0$ such that
  \begin{equation*} 
    \sup_{n>\tilde{n}} \sup_{m\in[n/2,3n/2]} \sup_{x,y\in\R^d} \prob\left( \max_{p\in \mathcal{N}_{\bm{A}}(x,K)} r_n^{-1}\norm{x-p} > r \right) \leq 3K(c_1\kappa \bar{f} (r/K)^d e^{-c_2 r/K})^K 
  \end{equation*}

  \noindent for $K,r$ sufficiently large.
\end{lemma}
\begin{proof}
  This follows from the proofs of Lemmas I.11 and I.12 in \cite{leung2019normal}.
\end{proof}

%-----------------------
\begin{lemma}\label{DC}
  Suppose there exist constants $\tilde{n},b_1,b_2,\epsilon_1>0$ such that
  \begin{equation}
    \sup_{n>\tilde{n}} \sup_{m\in[n/2,3n/2]} \sup_{x,y\in\R^d} \prob\left( \abs{J_x} > r \right) \leq b_1 e^{-b_2 r^{\epsilon_1}}.
    \label{AJexp}
  \end{equation}
  
  \noindent Then under \autoref{sparsity}, there exist constants $a_1,a_2,\epsilon_2>0$ such that
  \begin{equation*}
    \sup_{n>\tilde{n}} \sup_{m\in[n/2,3n/2]} \sup_{x,y\in\R^d} \prob\left( \max_{p\in J_x} r_n^{-1}\norm{x-p} > r \right) \leq a_1e^{-a_2r^{\epsilon_2}}. 
  \end{equation*}
\end{lemma}
\begin{proof}
  This is Lemma I.13 of \cite{leung2019normal}.
\end{proof}

%-----------------------
\begin{lemma}\label{JCN} 
  Under Assumptions \ref{dfrag}, \ref{sparsity}, and \ref{dreg}, there exist positive constants $\tilde{n},b_1,b_2,\epsilon$ such that
  \begin{equation*}
    \sup_{n>\tilde{n}} \sup_{m\in[n/2,3n/2]} \sup_{x,y\in\R^d} \prob\left( \abs{J_x(\tau_{x,r_n}\tilde\X_n \cup\{x,y\}, \W\cc_{x,r_n})} > r \right) \leq b_1 e^{-b_2 r^\epsilon}. 
  \end{equation*}
\end{lemma}
\begin{proof}
  This is Lemma I.14 of \cite{leung2019normal}.
\end{proof}

%----------------------------------------------------------------------

\end{document}